\documentclass[letterpaper,twocolumn,10pt]{article}
\usepackage{usenix}

\usepackage{tikz}
\usepackage{amsmath}

\usepackage{filecontents}

\usepackage{cite}
\usepackage{amsmath,amssymb,amsfonts}
\usepackage{algorithmic}
\usepackage[algo2e,ruled,vlined]{algorithm2e}
\usepackage{graphicx}
\usepackage{xcolor}
\usepackage{booktabs}       % professional-quality tables
\usepackage{multirow, makecell}       %合并行列
\usepackage{subfigure}
\usepackage{diagbox}
\usepackage{amsthm}

\usepackage{dblfloatfix}
\usepackage{afterpage}
\usepackage{array}
\usepackage{tabularx}
\usepackage{cleveref}
\usepackage{subcaption}
\usepackage{comment}

\usepackage{url}
\theoremstyle{plain}
\newtheorem{theorem}{Theorem}[section]

\theoremstyle{definition}

\theoremstyle{remark}
\newtheorem{remark}{Remark}[section]
\begin{document}
%-------------------------------------------------------------------------------

%don't want date printed
\date{}
\makeatletter
\renewcommand{\@fnsymbol}[1]{}   % 去掉 thanks 的星号
\makeatother

\title{\Large \bf TopFeaRe: Locating Critical State of Adversarial Resilience for Graphs Regarding Topology-Feature Entanglement \thanks{* The authors contribute equally. Corresponding author: Xinxin Fan.} \thanks{This work has been accepted to appear in USENIX Security'26.}\\
%  An (Incomplete) Example
}

%for single author (just remove % characters)
%\author{
%{\rm Xinxin Fan}\\
%State Key Laboratory of AI Safety, \\Institute of Computing Technology, \\Chinese Academy of Sciences; 
%University of Chinese Academy of Sciences
%\and
%{\rm Second Name}\\
%Second Institution
%% copy the following lines to add more authors
%% \and
%% {\rm Name}\\
%%Name Institution
%} % end author

\author{
	\rm Xinxin Fan$^{1,2,*}$ Wenxiong Chen$^{3,1,*}$ Quanliang Jing$^{1}$ Chi Lin$^3$ Shaoye Luo$^{1,2}$ Wenbo Song$^{3,1}$ Yunfeng Lu$^4$ 
	\\ $^1$State Key Laboratory of AI Safety, Institute of Computing Technology, Chinese Academy of Sciences 
    \\ $^2$University of Chinese Academy of Sciences, \textit{\{fanxinxin,jingquanliang, luoshaoye\}@ict.ac.cn}
    \\ $^3$Dalian University of Technology, \textit{c.lin@dlut.edu.cn} \textit{\{cwx1581015, songwenbo\}@mail.dlut.edu.cn} 
    \\ $^4$Beihang University, \textit{lyf@buaa.edu.cn}
}

\maketitle

%-------------------------------------------------------------------------------
\begin{abstract}
%-------------------------------------------------------------------------------
Graph adversarial attacks are usually produced from the two perspectives of topology/structure and node feature, both of them represent the paramount characteristics learned by today's deep learning models. Although some defense countermeasures are proposed at present, they fails to disclose the intrinsic reasons why these two aspects necessitate and how they are adequately fused to co-learn the graph representation. 
Towards this question, we in this paper propose an adversarial defense approach through locating the graph's critical state of adversarial resilience, resorting to the equilibrium-point theory in the discipline of complex dynamic system (CDS). In brief, our work has three novelties: i) Adversarial-Attack Modeling, i.e. map a graph regime into CDS, and use the oscillation of dynamic system to model the behavior of adversarial perturbation; ii) 2D Topology-Feature-Entangled Function Design for Perturbed Graph, i.e. project graph topology and node feature as two characteristic spaces, and define two-dimensional entangled perturbation functions to represent the dynamic variance under adversarial attacks; and iii) Location of Critical State of Adversarial Resilience, i.e. utilize the equilibrium-point theory to locate the graph's critical state of attack resilience resorting to the perturbation-reflected 2D function. Finally, multi-facet experiments on five commonly-used realistic datasets validate the effectiveness of our proposed approach, and the results show our approach can significantly outperform the state-of-the-art baselines under four representative graph adversarial attacks. 
\end{abstract}

%-------------------------------------------------------------------------------
\section{Introduction}
%-------------------------------------------------------------------------------
%Graphs are unstructured data that describe the relationships between entities in the real world, such as social networks \cite{social}, citation networks \cite{cite}, and transaction networks \cite{trade}. The emergence of graph neural networks (GNNs) \cite{GNN1,GNN2nc,GNN3} has significantly enhanced the ability to analyze graph data. Node classification \cite{GNN2nc}, as an important application of GNNs, aims to predict the categories of unlabeled nodes based on the features and structural information of the nodes in the graph. Understanding the information of unknown nodes can be useful in practical scenarios. For example, in recommendation systems \cite{reconmmender}, classifying customers can help recommend suitable products to them.
Graph not only demonstrates the linkage/connectivity relationship between homogeneous and heterogeneous nodes, but also bears the attributes (features) among individuals, such as social networks \cite{social}, paper-citation network \cite{cite}, and trade networks \cite{trade}, etc. Resorting to the graph topology and node features, some typical tasks can be conducted against today's deep-learning models, such as node classification, link prediction, graph alignment, etc. As the technique of deep neural network advances, a set of graph-focused deep learning models (a.k.a. target models) are proposed, including graph neural networks (GNNs) \cite{GNN1, GNN2nc, GNN3}, graph convolution network (GCN) \cite{gcnKipf}, graph attention network (GAT) \cite{GATVelickovic}, etc. 

To date, although these advanced neural networks have greatly enhanced the analytics of graph data, they are vulnerable and become highly susceptible to graph various adversarial attack (GAA) \cite{gaa1, gaa2, gaa3} by launching imperceptible perturbations via edge addition/removal, node injection, or/and feature modification, by which the target models are deceived to output incorrect inference. GAA causes severe consequence in security-critical applications, for instance, in loan-networked systems \cite{risk}, an inauthentic debtor might create multiple transaction records with other authentic users to camouflage himself/herself as a high-credit user \cite{GroupTrustFan, CSURFan, TAEffectLu}, as a result, a high reputation/trust score for this inauthentic debtor can be gained to succeed the illegitimate trade under the analytics of conventional GNNs, so do the malicious participants in other interactive (e.g. eCommerce) systems/networks. %platform, Twprotein-protein interaction network, etc.    

%\par
\begin{sloppypar}
In order to mitigate the adverseness, currently two lines of works have been proposed from the viewpoints of adversarial purification and robustness enhancement on GNNs. The former aims to eliminate the perturbations on edges/nodes, or node features to make the graph attack-resilient. The commonly-used strategy is to preprocess the contaminated graphs, for example, the representative works are GCN-SVD \cite{svd}, GCN-Jaccard \cite{jaccard}, and GNNGuard\cite{gnnguard}. The latter aims to learn robust representation through designing adequate neural (information-propagation) flow. The popular methods involve GAT \cite{gat}, RGCN \cite{rgcn}, HANG \cite{hang}, etc. Our work, aiming at exploring the generalized intrinsic critical state of attack resilience from the perspectives of graph topology and node features, belongs to the former category.  

\end{sloppypar}
\begin{figure*}[htbp]
\centering
\hspace*{0pt} 
\subfigure[Rank Growth] 
{
\begin{minipage}[b]{.28\linewidth}
\centering
\includegraphics[width=1.8in, height=1.3in]{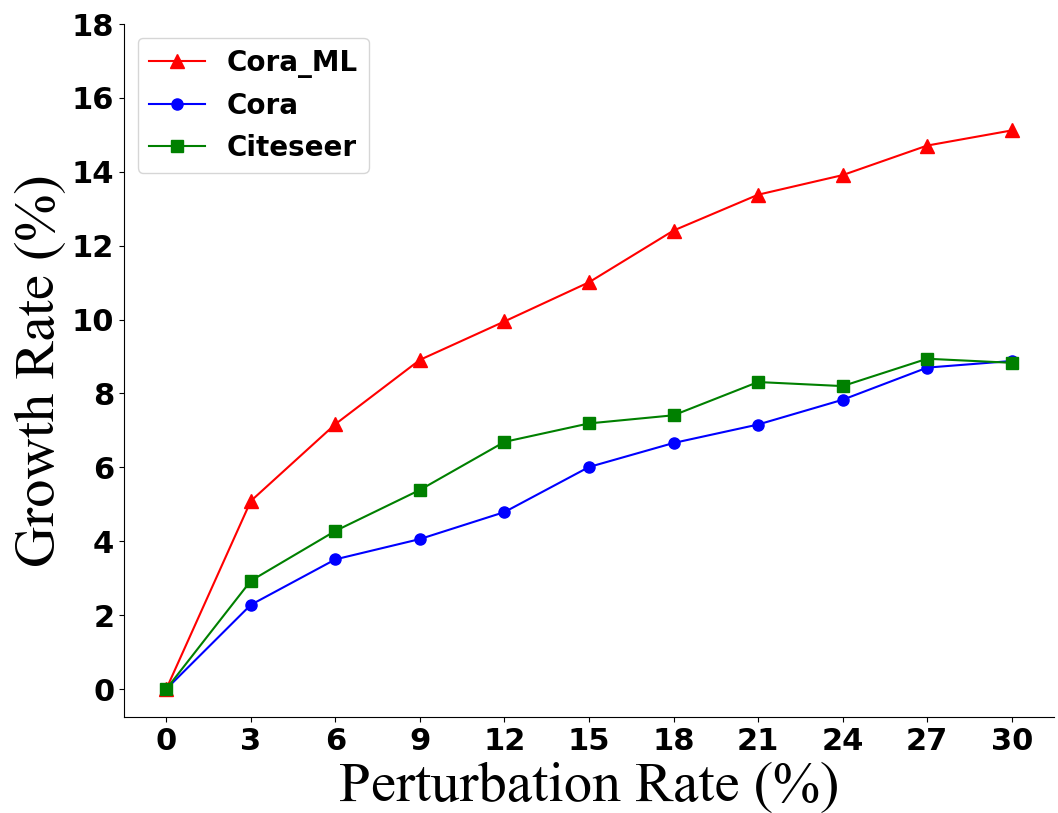}
\end{minipage}
}
\hfill 
\subfigure[Singular Value] 
{
\begin{minipage}[b]{.28\linewidth}
\centering
\includegraphics[width=1.8in, height=1.3in]{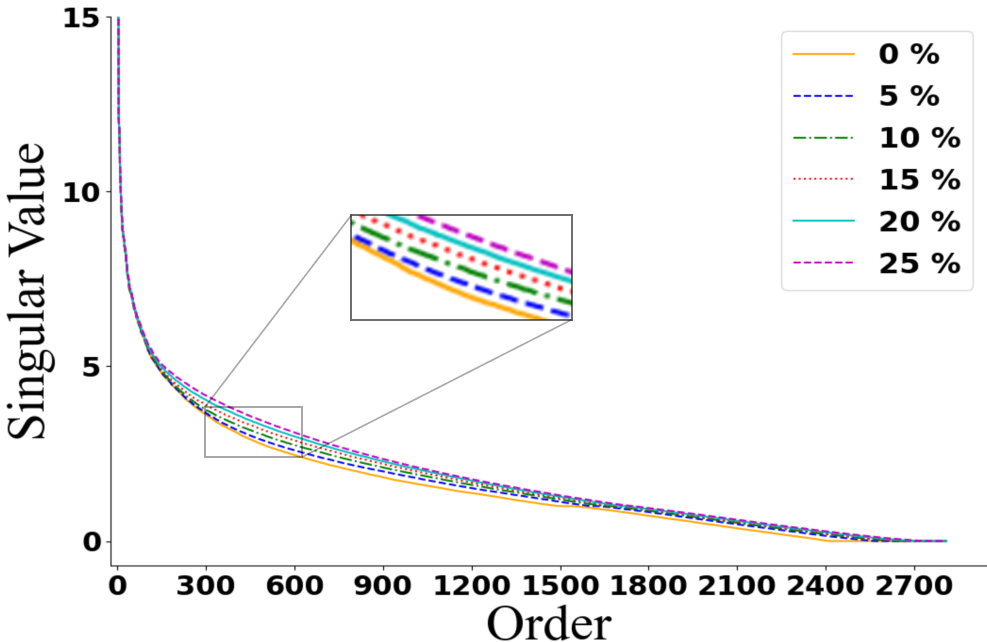}
\end{minipage}
}
\hfill 
\subfigure[Feature Smoothness] 
{
\begin{minipage}[b]{.34\linewidth}
\centering
\includegraphics[width=2.2in, height=1.3in]{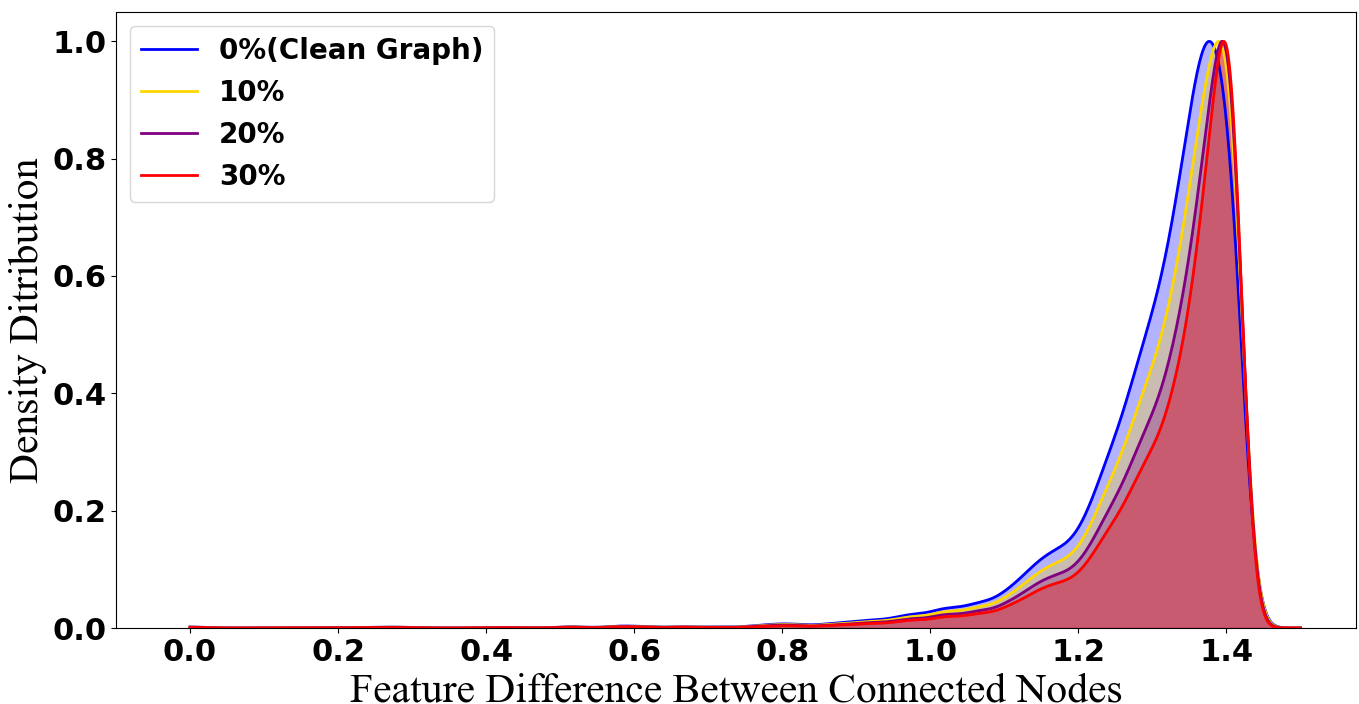}
\end{minipage}
}
\caption{Variations of adjacency-matrix rank and singular value, feature smoothness under attack Metattack}
\label{fig:AttackGraphProperty}
\end{figure*}

\par
\begin{sloppypar}
As the two basic yet pivotal elements, the topology and feature become the disturbed objects in the procedure of GAA. As well known, the topology apparently demonstrates the correlation among nodes, while the feature indicates the semantics (dis)similarity over two individuals. To make the target models learn inaccurate representation, GAA usually builds linkages between structure-/semantic-dissimilar nodes. In the same vein, our work also leverages such two elements to defeat diversities of adversarial attacks from the standpoint of searching and locating the graph's intrinsic critical state of attack resilience. In brief, our work attempts to explore the dual-resilience domains in terms of graph topology and node feature through resorting to the equilibrium-point (stability) theory of CDS, that is, the process of perturbations under GAA can be recognized as the dynamic variances (oscillations) of CDS over time, and the equilibrium point with the asymptotically-stable status is corresponding to the graph's critical attack-resilient state. Pursuant to this idea, we propose an attack-resilient approach TopFeaRe to resist various GAAs. In a nutshell, three main contributions are involved:
\par
%Our contributions can be summarized as follows:
\begin{itemize}
\item 
We exhibit what aspects are influenced by GAA, and theoretically explain why the target models are easily deceived by adversarial perturbations.  
\item 
We design topology-feature-entangled function to model the graph's dynamics in terms of topology and feature under adversarial perturbations, and formally analyze what prerequisites are required to satisfy the existence theory of asymptotically-stable equilibrium point. 
\item We conduct extensive experiments on five real-world datasets under three non-targeted and one targeted popular GAAs to validate the effectiveness of our approach, as well as the rationality of mapping adversarial perturbations into oscillations of CDSs. The experimental results show that our proposed TopFeaRe significantly outperforms the sate-of-the-art baselines. 
\end{itemize}
\end{sloppypar}

\section{Problem Statement and Our Solution}
\subsection{Affections of Adversarial Attack}
We at first investigate what affections the adversarial attack may bring in. Fig. \ref{fig:AttackGraphProperty} exhibits the variational tendencies of adjacency-matrix rank, singular value, and feature smoothness under the representative graph adversarial attack Metattack \cite{r4}, from which we can observe: i) the rank of adjacency matrix goes up apparently as the rate of adversarial perturbation (a.k.a. perturbation rate) enlarges for the three real-world datasets Cora, Cora\_ML, and Citeseer; ii) the singular value of adjacency matrix raises as well, especially in the middle segments; and iii) the density distribution of feature difference over connected nodes has a right-transfer phenomenon after adversarial attacks. Pursuant to the experimental results, we know that the GAA can at least causes the three facets of affections in terms of topology and feature. As we know, the currently-existing defense countermeasures are also proposed in consideration of these affections, such as the popular category of preprocessing methods.  
%The rest is organized as follows. Section 2 \textcolor{red}{xxxx.}

\subsection{Problem Definition}
At present, no matter the preprocessing-based adversarial purification or neural flow-based robustness enhancement, these countermeasures are witnessed to be effective against various GAAs. However, we think the following fundamental questions deserve more inspections and reasonable interpretation.   
%Upon the in-depth analysis on the existing defense methods, we can summarize the following drawbacks of currently existing defense countermeasures as follows:

\textbf{Question I: Why the existing preprocessing-based methods are effective and whether the original characteristics might be mistakenly deleted.} As well known, the preprocessing methods
%, such as GCN-SVD \cite{svd}, GCN-Jaccard \cite{jaccard} and GNNGuard \cite{gnnguard}, 
almost purify the adversarial perturbations through removing a certain ratio of contaminated edges/nodes based on the condition of low-rank of adjacency matrix, or high feature-similarity over connected nodes. However, how many edges/nodes ought to be exactly removed is not known, the existing works just depend on pre-defined ratios or similarity thresholds \cite{svd, jaccard}, that is to say, there does not exist an explicit referral to guide the ratio and threshold, which might mistakenly under/over-delete the original nodes/edges, leading to the loss of valuable feature. Furthermore, even if the mischievously-added edges or injected nodes are removed, why such operations are effective is unknown. 

\textbf{Question II: Whether to achieve the true meaning of jointly-learning of graph topology and node features.} The existing defense methods in general utilize two separate loss functions to train the topology and feature through setting different hyperparameters, which is not a joint-learning in the true sense. For instance, Pro-GNN \cite{JinMa20} defines three loss functions to enhance the robustness, one is to constrain the contribution of the properties of sparsity and low rank from the viewpoint of graph topology, the second is to measure the feature difference over pairwise nodes, and the third is to guide the graph-learning process for specific downstream tasks. Hence, how to realize and guarantee the genuine fusion of jointly-learning regarding topology and feature together, rather than through separate constraints via respective loss functions, is still blank.   

\textbf{Question III: What the relationship between graph data and neural-network learning is.} 
%\textbf{Question III: Whether Neural Network Learning Benefits from Graph Attributes.}  
Currently, the two lines of adversarial defenses develop separately, limited studies are dedicated to the relationship between graph data and neural networks. As we know, various graph-focused neural networks, such as GAT \cite{gat}, RGCN \cite{rgcn}, HANG \cite{hang}, Mid-GCN \cite{HuangJin25}, have manifested the advantage to learn robust embeddings for nodes. However, the more we want to know is, whether the graph data in a proper (equilibrium-point) state can further boost the robustness-learning capacity of such category of neural networks.

%\textcolor{red}{\textbf{Extra Module Brings Expensive Time Overheads.} To date, various graph-focused neural networks, such as GAT\cite{gat}, RGCN\cite{rgcn}, HANG\cite{hang}, Mid-GCN\cite{HuangJin25}, have manifested the advantage to learn robust embeddings for graph nodes, however, besides that a large number of hyparameters need to be fine-tuned in consideration of different combinations, some extra modules also need additional overhead to handle, such as the attention module in GAT \cite{gat}, and the mid-frequency filter module in Mid-GCN\cite{HuangJin25}.} 
%\subsection {Problem Definition} 
Keeping the three fundamental questions in mind, a solution naturally appears, i.e. whether there exists an inherent critical state of attack resilience for each graph, to answer this question, we propose a universal adversarial defense mechanism through exploring such a critical attack-resilient state resorting to the equilibrium-point theory of CDS. In brief, we first map the adversarial perturbations into the oscillation of CDS and measure the adverseness using differential function. Second, we embody topology and feature into two topology-feature-tangled functions to sketch the adversarial perturbations, by which the asymptotically-stable equilibrium point (ASEP) is inferred. Third, referring to such ASEP, we rebuild the contaminated graph to make it attack-resilient. 

\subsection{Solution Overview}
From Fig. \ref{fig:AttackGraphProperty} we know GAA usually causes the increase of adjacency-matrix rank and singular value, in addition to the right-transfer of feature-density distribution. From the viewpoint of defense, a direct yet valid idea is to descend the rank and singular value through removing the mischievously-added edges and injected nodes, and simultaneously make the feature smoothly through disconnecting the feature-dissimilar nodes that are connected maliciously. However, which dissimilar nodes should be disconnected and how many edges need to be removed remain unclear. Towards this doubt, an idea naturally emerges, namely, for each graph, whether there exists a critical attack-resilient state to resist adversarial attacks, if exists, how to find out it. To this end, regarding the topology and feature together, we employ the equilibrium-point theory of CDS to deduce the asymptotically-stable equilibrium point, by which the necessarily-removed edges are inferred to boost the attack resilience of perturbed graphs. 
\begin{center}
\fcolorbox{black}{gray!10}{\parbox{1.0\linewidth}{
\begin{remark}
\textbf{Relationship Between Oscillations and Perturbations.} CDS can converge into an asymptotically stable equilibrium point (ASEP) after continuous oscillations, such ASEP corresponds to CDS’s resilient state that restores to its original state even after long-time oscillations, that is to say, this resilient state implies high robustness. Inspired by this observation, we bridge the connection between oscillations and perturbations from the viewpoint of interpretability.
\end{remark}}}
\end{center}

\section{Adversarial Attack Modeling}
\subsection{Graph-Dynamics Metric}
\begin{sloppypar}
GAA aims to repeatedly modify pivotal edges/nodes, and create connection between feature-dissimilar nodes to fool graph neural networks. Here we refer to such a procedure as adversarial learning, which exactly sketches the dynamic variance of the victim graph under adversarial attacks. Take into consideration the complexity of adversarial perturbations and dependency on concrete topology and feature, it is difficult to give an accurate formula to exactly demonstrate the process of adversarial learning, if not impossible. 
However, the conventional handling is to model the dynamics of graph, as stated in the literature \cite{degree, gao}, one node's dynamics can be influenced by its inner self-dynamics and the outer affections by neighboring nodes. Concretely, a node $i$'s dynamics can generally be presented as a linear weighted average of the $N$ state variables each of which is associated with a node $i$, and the nonlinear influence of other nodes on node $i$. In reality, the study \cite{gao} unveils there indeed exists the nonlinear interactively-affected effect in ecological networks, e.g., the symbiotically connected species networks, such as plants and pollinators or fish and anemone networks. From the standpoint of adversarial perturbations, we also follow this idea of linear and nonlinear perturbations, and provide a universal differential-form definition to sketch the dynamics under continuous adversarial perturbations:

\begin{equation} 
\label{Equ:GeneDyn}
\begin{array}{l}
\frac{d}{{dt}}\mathord{\buildrel{\lower3pt\hbox{$\scriptscriptstyle\rightharpoonup$}} 
\over x} \left( t \right) = \mathbf{A}\mathord{\buildrel{\lower3pt\hbox{$\scriptscriptstyle\rightharpoonup$}} 
\over x} \left( t \right) + \mathbf{B}\left( { - \mathop \phi \limits^ \to  \left( {\mathord{\buildrel{\lower3pt\hbox{$\scriptscriptstyle\rightharpoonup$}} 
\over y} \left( t \right)} \right)} \right) \\ 
\mathord{\buildrel{\lower3pt\hbox{$\scriptscriptstyle\rightharpoonup$}} 
\over y} \left( t \right) = \mathbf{C}\mathord{\buildrel{\lower3pt\hbox{$\scriptscriptstyle\rightharpoonup$}} 
\over x} \left( t \right)\begin{array}{*{20}c}
{} & {} & {}  \\
\end{array} \\ 
\end{array},
\end{equation}
where $\frac{d}{{dt}}\mathord{\buildrel{\lower3pt\hbox{$\scriptscriptstyle\rightharpoonup$}} \over x} \left( t \right)$ is an $N$-dimensional vector indicating the state (connectivity) variance of all $N$ nodes during the adversarial perturbations over time $t$, and $ \mathbf{A}\mathord{\buildrel{\lower3pt\hbox{$\scriptscriptstyle\rightharpoonup$}} \over x} \left( t \right) $ denotes $N$-dimensional vector denoting the affected linear perturbations. Given the continuously-perturbed misbehavior on edges in the process of adversarial learning, variable $\mathord{\buildrel{\lower3pt\hbox{$\scriptscriptstyle\rightharpoonup$}} \over y} \left( t \right)$ is the output vector standing for the states of all nodes afterward the one-round adversarial perturbations. You can imagine from the angle of routine forward/back-propagation principle in deep learning, the output vector $\mathord{\buildrel{\lower3pt\hbox{$\scriptscriptstyle\rightharpoonup$}} \over y} \left( t \right)$ can in return to influence the dynamics of $\vec x\left( t \right) $ in a nonlinear manner, i.e. -$\phi(\mathord{\buildrel{\lower3pt\hbox{$\scriptscriptstyle\rightharpoonup$}} \over y} \left( t \right)) $. Thus, $\frac{d}{{dt}}\mathord{\buildrel{\lower3pt\hbox{$\scriptscriptstyle\rightharpoonup$}} 	\over x} \left( t \right) $ and $\vec x\left( t \right) $ are entangled and interplay iteratively, which implies the embodiment of continuous adversarial perturbations on graph topology and node features. The matrices $\mathbf{A}$, $\mathbf{B}$, $\mathbf{C}$ are the coefficients to demonstrate the detailed dynamic variations of graph regime in the process of adversarial learning.

Accordingly, the accumulated adverseness can be calculated using the integral from starting time $t_0$ to $t>0$, which equivalently corresponds to the aggregation of a period of time continuous perturbations. Therefore, mathematically, for each node $i$, the metric on time-aware dynamic variation under adversarial perturbations can be formulated as:
\begin{equation}
X_{i}\left( t \right) = \int_{t_0 }^t {x_{i}\left( \kappa  \right)d\kappa}. 
\end{equation}

Of course, the accumulated dynamics for each node are brought by both topology perturbation and feature perturbation, and finally leading to the affection on the whole graph. 
%\textcolor{red}{Need Laplace Transform?}
\end{sloppypar} 

\subsection{1D Condensed Perturbation Metric}
\begin{sloppypar}
%As aforementioned, we can define the node-state dynamic variation for each individual as shown in Eq. \eqref{Equ:GeneDyn}. However, for present-day's large-scale networks/graphs, it is impossible to quantify each individual's dynamics under continuous perturbations. To conquer this problem, some dimension-reduction handling is workable to present the unified dynamics for all nodes as a whole, that is to say, the whole graph can be treated as a unity to together consider the dynamics by adversarial attacks. 
As aforementioned in Eq. \eqref{Equ:GeneDyn}, the best way is to build a dimension for each node, which can capture the fine-grained perturbations. Nevertheless, such operation may be difficult resulting from two reasons: i) to date there is no principled guidance on how to establish an ASEP-reserved function, referring to the existing knowledge of CDS, in other words, the candidate oscillation-mapping functions need to be figured out manually; and ii) for high-dimensional function, the analysis on whether it has an ASEP would become extremely sophisticated, due to the correlation among different dimensions. Therefore, the common treatment is to use a dimension-reduction function to approximately express the dynamics of the whole graph. Specifically, the edge perturbation can, on one hand, cause the connected nodes' state variation resulting from the change of learnable semantics; on the other hand, each individual itself may have dynamics, e.g. the attribute update itself. Thereby, referring to the individual's own dynamics and pairwise connectivity-caused influence under adversarial perturbations, the universal dynamic-variation of each individual over time can be expressed as:         
%The evolution of each node's state over time in a graph network described by a dynamical system can be expressed through a set of coupled nonlinear ordinary differential equations\cite{eq1-1,eq1-2}
%\begin{equation}
%    \label{eq:1}
%    \frac {d{x}_{i}} {dt}=f(\textstyle\sum ^{N}_{j=1} {{A}_{ij}},{x}_{i})=F({x}_{i})+\textstyle\sum ^{N}_{j=1} {{A}_{ij}}G({x}_{i},{x}_{j})
%\end{equation}
\begin{equation}
\label{Equ:NodeDyna}
\frac {d{x}_{i}} {dt}=F({x}_{i})+\textstyle\sum ^{N}_{j=1} {\mathcal{A}_{ij}}G({x}_{i},{x}_{j}),
\end{equation}
where $x_i$ is the state variable of node $i$, $F(x_i)$ represents the dynamics of node $i$ itself, and $G(x_i, x_j)$ is the function describing the pairwise connectivity between nodes $i$ and $j$, ${A}_{ij}$ denotes the element of the adjacency matrix, indicating the strength of connectivity (a.k.a. semantic-correlation w.r.t. node feature).  

%For a network with \( N \) nodes, representing the network's resilience from a high-dimensional perspective requires computing Equation \eqref{eq:1} for all \( N \) nodes\cite{realistic}. However, in large-scale networks, computing the state of each node is impractical.  
%Gao et al. developed a method to simplify the analysis of network dynamical systems\cite{gao}, enabling efficient exploration and prediction of changes in the network's resilience over time, and providing insights into the network's dynamic behavior under external disturbances.

Starting from the dimension reduction, we define $\left \langle \hat y_j(x_i) \right \rangle$ to denote the average dynamics of node $i$ caused by every neighboring node $j$. Thus, we have that 
\begin{equation}
\label{Equ:AvgDyn}
\textstyle\sum_{j=1}^{N}\mathcal{A}_{ij}G({x}_{i},{x}_{j})={{s}_{i}}^{in}{\left \langle \hat {y}_{j}({x}_{i})\right \rangle},
\end{equation}  
where ${{s}_{i}}^{in}=\textstyle\sum_{j=1}^{N}\mathcal{A}_{ij}$, and $\left \langle \hat y_j(x_i) \right \rangle$ represents the weighted average of the connectivity strength $ G({x}_{i},{x}_{j})$ over all in-degree neighbors of node $i$.

%In the process of analyzing Equation \eqref{Equ:NodeDyna} using dimensionality reduction methods, firstly consider the scalar \({y}_{i}\) associated with the \(i\)-th node and its weighted in-degree \({{s}_{i}}^{in}\), and denote them as \({y}_{j}({x}_{i})=G({x}_{i},{x}_{j})\) and \({{s}_{i}}^{in}=\textstyle\sum_{j=1}^{N}{A}_{ij}\), respectively.Therefore,we have
%\begin{equation}
%    \label{eq:2}
%    \textstyle\sum_{j=1}^{N}{A}_{ij}G({x}_{i},{x}_{j})={{s}_{i}}^{in}{\left \langle {y}_{j}({x}_{i})\right \rangle}_{j\in i}
%\end{equation}
%in Equation \eqref{eq:2}, \( \left \langle {y}_{j}({x}_{i})\right \rangle=  \textstyle\sum_{j=1}^{N}{A}_{ij}G({x}_{i},{x}_{j}) \slash {{s}_{i}}^{in}\) can be defined as the weighted average of \( G({x}_{i},{x}_{j}) \) over the in-degree neighbors of the \( i \)-th node, with the weights given by \( \textstyle\sum_{j=1}^{N}{A}_{ij} \).
\end{sloppypar}
\par
\begin{sloppypar}
%To simplify the interactions between nodes, the heterogeneous mean field (HMF)\cite{hmf} approximation is used to compute \({\left \langle {y}_{j}({x}_{i})\right \rangle}_{j\in i}\), representing the average behavior of all nodes in the entire network system.
%Assume that, regardless of node \(i\), the probability \(P(j|i)\) of node \(i\) having an in-degree neighbor at node \(j\) is proportional to the out-degree \({{s}_{j}}^{out}\) of node \(j\), i.e., \(P(j|i)\propto {{s}_{j}}^{out}\). To normalize this proportional relationship, let the normalization factor be the sum of the out-degrees of all nodes \(j\), denoted as \(\textstyle\sum_{j=1}^{N}{{s}_{j}}^{out}\). Then, \(P(j|i)\) can be written as: \(P(j|i)={{s}_{j}}^{out}\slash\textstyle\sum_{j=1}^{N}{{s}_{j}}^{out} \). Since \((P(j|i)\) reflects the influence of node \(j\) on node \(i\), it is treated as the weight of \({y}_{j}({x}_{i})\). Through weighted averaging, can obtain

To unify the entangled and sophisticated adversarial perturbations over all pairs of nodes, we employ the classic and commonly-used heterogeneous mean field (HMF) theory \cite{hmf} to approximate $ \left \langle \hat y_j(x_i) \right \rangle$. Normally, the affection probability $P(j|i)$ of node $i$ placed by its neighboring node $j$ is usually proportional to the $j$'s out-degree \({{s}_{j}}^{out}\), i.e. \(P(j|i)\propto {{s}_{j}}^{out}\). To normalize the affection, the regular way is to regard the normalization factor as the sum of all nodes' out-degrees, i.e. $\textstyle\sum_{j=1}^{N}{{s}_{j}}^{out}$. Thus, the affection probability by node $j$ can defined as $P(j|i)={{s}_{j}}^{out}\slash\textstyle\sum_{k=1}^{N}{{s}_{k}}^{out}$. Accordingly, the weighted average dynamics of node $i$ can be formulated as,
\begin{equation}
\begin{aligned}
\label{Equ:HMFDyna}
{\left \langle \hat{y}_{j}({x}_{i})\right \rangle}_{HMF} &= \textstyle\sum_{j=1}^{N}P(j|i)\hat{y}_{j}({x}_{i}) =\frac{\frac{1}{N}\textstyle\sum_{j=1}^{N}{{s}_{j}}^{out}\hat{y}_{j}({x}_{i})}{\frac{1}{N}\textstyle\sum_{j=1}^{N}{{s}_{j}}^{out}}.
\end{aligned}
\end{equation}

As reported in literature \cite{degree}, if the degree correlation between nodes is small, each individual can be considered to draw from the same distribution, and the average behavior of the graph as a whole can be used to approximate the local behavior of each single node. That is to say, for all nodes we have the approximation equation ${\left \langle \hat{y}_{j}({x}_{i})\right \rangle}\approx {\left \langle \hat{y}_{j}({x}_{i})\right \rangle}_{HMF}$. As we know, in the interactive networks, such as Twitter and Facebook, the users (nodes) usually have independent activities, leading to small correlation among different individuals, thus, we think the statement above well-matches the reality and the approximation equation is reasonable. The dataset statistics in Table \ref{tab:1} also states this point, i.e., the degree correlations of the five popular datasets are all near to zero. To simplify Eq. \eqref{Equ:HMFDyna}, we define the mathematical operator $\boldsymbol{\Phi}$ as
\begin{equation}
\label{eq:5}
\boldsymbol{\Phi}(\mathrm{\mathbf{\hat Y}})=\frac{{\mathbf{1}}^{\top}\mathcal{A}\mathrm{\mathbf{\hat y}}}{{\mathbf{1}}^{\top}\mathcal{A}\mathrm{\mathbf{1}}}=\frac{\textstyle\sum_{i=1}^{N}\textstyle\sum_{j=1}^{N}\mathcal{A}_{ij}{\hat y}_{j}}{\textstyle\sum_{i=1}^{N}\textstyle\sum_{j=1}^{N}\mathcal{A}_{ij}},
\end{equation}
where $\mathrm{\mathbf{\hat Y}}$ = ${\left({\hat y}_{1}, {\hat y}_{2}, \cdots, {\hat y}_{N}\right)}^\top$, $\mathbf{1}$ = $\left(1, 1, \cdots, 1\right)^\top$. Hence, Eq. \eqref{Equ:NodeDyna} can be rewritten as
\begin{equation}
\label{eq:6}
\frac{{d{x}_{i}}}{{d}t} \cong F({x}_{i})+{{s}_{i}}^{in}\mathrm{\boldsymbol{\Phi}}(G({x}_{i},\mathrm{\mathbf{x}})),
\end{equation}
where $\mathbf{x}$=$\left(x_{1}, \cdots, x_{N}\right)^\top$, $G(x_{i},\mathbf{x})$=$\left(G(x_{i}, x_{1}), \cdots, G(x_{i}, x_{N})\right)^\top$, and $\mathrm{\boldsymbol{\Phi}}(G({x}_{i},\mathrm{\mathbf{x}}))$ denotes the average edge weight between node $i$ and its neighboring nodes, under HMF theory, it can be further approximated as $G({x}_{i},\boldsymbol{\Phi}(\mathrm{\mathbf{x}}))$. Writing Eq. \eqref{eq:6} in vector form, we have  
\begin{equation}
\label{eq:7}
\frac{{d{\mathrm{\mathbf{x}}}}}{{d}t}= F(\mathrm{\mathbf{x}})+{\mathrm{\mathbf{s}}}^{in}\odot G(\mathrm{\mathbf{x}},\boldsymbol{\Phi}(\mathrm{\mathbf{x}})),
\end{equation}
where $G(\mathbf{x}, \boldsymbol{\Phi}(\mathbf{x}))=(G(x_1, \boldsymbol{\Phi}(\mathbf{x})), \cdots, G(x_N, \boldsymbol{\Phi}(\mathbf{x})))^\top$, $\mathbf{s}^{in}=(s_1^{in}, \cdots, s_N^{in})^\top$, symbol $\odot$ represents the Hadamard product between vectors, i.e. pairwise elements multiply. Then, substituting the linear mathematical operator $\boldsymbol{\Phi}$ into Eq. \eqref{eq:7}, 
\begin{equation}
\begin{aligned}
\label{eq:8}
\frac{d{\boldsymbol{\Phi}}({\mathbf{x}})}{dt} &= {\boldsymbol{\Phi}}(F(\mathbf{x}) + {\mathbf{s}}^{in} \odot G(\mathbf{x}, \boldsymbol{\Phi}(\mathbf{x}))) \\
&= {\boldsymbol{\Phi}}(F(\mathbf{x})) + \boldsymbol{\Phi}({\mathbf{s}}^{in} \odot G(\mathbf{x}, \boldsymbol{\Phi}(\mathbf{x}))) \\
&\cong F(\Phi(\mathbf{x})) + \boldsymbol{\Phi}({\mathbf{s}}^{in}) G(\boldsymbol{\Phi}(\mathbf{x}), \boldsymbol{\Phi}(\mathbf{x}))
\end{aligned}.
\end{equation}
Pursuant to HMF theory \cite{hmf}, we can obtain $\boldsymbol{\Phi}(F(\mathbf{x})) \cong F(\boldsymbol{\Phi}(\mathbf{x}))$ and $\boldsymbol{\Phi}({\mathbf{s}}^{in} \odot G(\mathrm{\mathbf{x}},\boldsymbol{\Phi}(\mathrm{\mathbf{x}}))) \cong \boldsymbol{\Phi}({\mathbf{s}}^{in})G(\boldsymbol{\Phi}(\mathrm{\mathbf{x}}),\boldsymbol{\Phi}(\mathrm{\mathbf{x}}))$. 

%节点自身状态的改变均受到其邻近节点的影响，进而导致整个系统状态的改变，所以这种平均邻近节点活动可以来刻画系统整体的平均有效状态
As previously analyzed, an individual's perturbation could render the variations of its neighbors' states, further giving rise to state-change of the whole graph from the global viewpoint. In view of this, the weighted average state of all nodes can be used to sketch the whole graph's state $x_{Gra}$ in a unity way, 
\begin{equation}
\label{eq:9}
{x}_{Gra}=\frac{\left \langle{s}^{out}x \right \rangle}{\left \langle{s}^{out} \right \rangle}.
\end{equation}
%it represents the weighted average of the states \( x_i \) of all nodes, with the weights being \( {s_i}^{out} \). 
We define a node's state as its degree centrality, stemming from the intuition that the basic yet pivotal degree property plays an important role for the robustness, the subsequent experiments also verify this point. On the other hand, the in-degree of a node can be straightly leveraged to execute perturbation, thus, we take it into account and define a control parameter for the topology-feature-entangled perturbations as
\begin{equation}
\label{eq:10}
{\beta}_{Gra}=\frac{\left \langle{s}^{out}{s}^{in} \right \rangle}{\left \langle{s}^{out} \right \rangle}.
\end{equation}

Substituting $\boldsymbol{\Phi}(\mathbf{x})$=$x_{Gra}$ and $\boldsymbol{\Phi}({\mathbf{s}}^{in})$=$\beta_{Gra}$ into Eq. \eqref{eq:8}, a one-dimensional state-variation function can be inferred from the viewpoint of whole graph, 
\begin{equation}
\label{eq:11}
\frac {dx_{Gra}} {dt}=F(x_{Gra})+{\beta }_{Gra}G(x_{Gra}, x_{Gra}).
\end{equation}

The equation above is a coupled nonlinear ordinary differential equation, $F(x_{Gra})$ and $G(x_{Gra}, x_{Gra})$ respectively reflect the linear and nonlinear properties, here we leverage the 1D form of Michaelis–Menten Equation \cite{alon} to represent the nonlinearity, and introduce a constant to present the linearity,
\begin{equation}
\label{eq:12}
\frac {d{x}_{Gra}} {dt}=\mathcal{B}{x}_{Gra}+{\beta}_{Gra}\frac{\left({{x}_{Gra}}\right)^{h}}{1+ \left({{x}_{Gra}}\right)^{h}}, 
\end{equation}
where $\mathcal{B}$ and $h$ are two adaptive parameters to represent the node-state dynamics under continuous perturbations for different graph regimes. Of course, Michaelis–Menten Equation is just one selection, other functions that can appropriately describe the dynamics of graph perturbations are applied as well, this is because there does not exist guiding principle on how to establish such an dynamics-mapping function. For clarity, we treat the condensed $x_{Gra}$ and $\beta_{Gra}$ as the dynamic variable and crucial parameter. That is to say, the affection on node-state dynamics by perturbation is subject to $\beta_{Gra}$, and for an anticipated critical attack-resilient state $x_0 \in (0,+\infty)$, there exists an instantaneous state $\beta_{Gra}$ corresponding to $x_{Gra}=x_0$. Easy to understand that, the higher the value of $x_{Gra}$, the better the graph's attack resilience (robustness). Extremely, if $x_{Gra} \to 0$, it means the entire graph is destroyed completely.
\end{sloppypar}

%\begin{sloppypar}
%Let \( \frac {d{x}_{0}} {dt} = 0\), i.e., \( f(\beta,x_0) = 0 \), and this is called the steady-state solution of the system. We can use the steady-state solution to derive a function\(\beta_{eff}(x_{eff})\) of \( \beta_{eff} \) in terms of \( x_{eff} \)
%\begin{equation}
%    \label{eq:13}
%    {\beta }_{eff}({{x}_{eff}})=-\frac{B(1+{{x}_{eff}}^{h})}{{{x}_{eff}}^{h-1}}
%\end{equation}
%by swapping the \( x \)- and \( y \)-axes, we can obtain the function \( x_{eff}(\beta_{eff}) \) of \( x_{eff} \) in terms of \( \beta_{eff} \). We will use Equation \eqref{eq:12} to explore the dynamic behavior of the graph network under adversarial graph attacks. by swapping the \( x \)- and \( y \)-axes, we can obtain the function \( x_{eff}(\beta_{eff}) \) of \( x_{eff} \) in terms of \( \beta_{eff} \). The one-dimensional simplified equation applied to graph data is shown in the figure. We will use equation (12) to explore the dynamic behavior of the graph network under adversarial graph attacks.
%\end{sloppypar}

\subsection{2D Entangled Perturbation Metric}
%Applying one-dimensional simplification methods to graph datasets may not be suitable, because models like graph neural networks (GNNs) learn node representations by considering both the connectivity and feature information of the nodes to make predictions. On the other hand, Graph Adversarial Attacks (GAAs) design attack strategies based on the model's learning characteristics, with the core goal of disrupting these two types of relationships between nodes. As shown in Figure \ref{fig:AttackGraphProperty}, when the graph is subjected to adversarial perturbations, the feature differences between nodes increase, and the rank of the adjacency matrix gradually rises. In this case, we believe that the graph's resilience will decrease as the perturbation rate increases.
%\par
%To more accurately represent the impact of perturbations on the graph, we extend the one-dimensional simplification method to a two-dimensional approach for dimensionality reduction of the graph network. In this case, two nonlinear ordinary differential equations are used as our simplification model, leading to two dynamic variables and two independent critical parameters\cite{2D}. According to Equation \eqref{eq:1}, the dynamic behavior of each node is extended to the following form
%图神经网络的学习过程本质上是基于节点的结构和特征来进行信息传播与更新的。因此，节点的结构状态 r和特征状态 q在动态系统中相互耦合。在对抗攻击中，攻击者试图通过扰动图的结构或节点特征来破坏模型的学习效果。扰动可能通过改变图的邻接关系（结构扰动）或改变节点之间的特征关系来实现。公式中的γr和γq是在结构与特征上节点之间连接关系强度有关的关键参数，合理反映了扰动对节点之间结构和特征关系的影响，例如，γr和γq的减小可以模拟对抗攻击对图的破坏程度。当这些参数增大时，可能意味着节点之间的连接性或特征差异更强，从而使得扰动对图的影响更大，这与对抗攻击中的破坏效果是一致的。
\begin{sloppypar}
In essence, both topology and feature play important roles for the graph representation for GNNS. The more sophisticated and powerful adversarial attacks also deceive target models through disturbing the two radical attributes, as shown in Fig. \ref{fig:AttackGraphProperty}, not only do the adjacency-matrix's rank and singular value ascend after attack, but also the distribution of feature difference has a transfer. That is to say, graph topology and node feature are entangled, this is because the edge-connection is intentionally created among feature-dissimilar nodes with the purpose of misleading the node's semantic learning, implying the distinction of feature also influences the edge perturbation during adversarial attack; on the other hand, the topology relying on the edges also disrupts the feature learning. Thus, we ought to redefine the two-factor function into $G(r_i,r_j)$=$H(r_i,r_j,q_{j})$, where $r_i$ and $q_i$ can be respectively deemed as graph-topology variable and node-feature variable, $H(r_i,r_j,q_{j})$ indicates the entanglement of topology and feature, i.e. the perturbation on node $i$ can potentially affect the variations of edge-connection to node $j$ and the feature-correlation between node $j$ while executing adversarial perturbation. Therefore, we measure the topology-feature-entangled perturbations, and extend the predefined topology-specific dynamics function into two nonlinear ordinary differential functions, 
\begin{equation}
\begin{aligned}
\label{eq:14}
\frac {d{r}_{i}} {dt} &= F({r}_{i}) + \sum ^{N}_{j=1} \mathcal{A}_{ij} H({r}_{i},{r}_{j},{q}_{j}) \\
\frac {d{q}_{i}} {dt} &= F({q}_{i}) + \sum ^{N}_{j=1} \mathcal{A}_{ij} H({q}_{i},{q}_{j},{r}_{j})
\end{aligned},
\end{equation}
where $r_i$ and $q_i$ are used to represent the graph-topology dynamics and node-feature dynamics of node $i$. Pursuant to Eq. \eqref{eq:12}, we can infer the following formulas:
%more detailed information is expressed as follows \cite{eq15-1,eq15-2}
\begin{equation}
\begin{aligned}
\label{eq:15}
\frac {d{r}_{i}} {dt} &= -m_i {r}_{ij}+\frac{\textstyle\sum ^{N}_{j=1}{{\gamma }_{ij}}^{r}{q}_{j}}{1+\textstyle\sum ^{N}_{j=1}{{\gamma }_{ij}}^{r}{q}_{j}}\\
\frac {d{q}_{i}} {dt} &= -c_i {q}_{ij}+\frac{\textstyle\sum ^{N}_{j=1}{{\gamma }_{ij}}^{q}{r}_{j}}{1+\textstyle\sum ^{N}_{j=1}{{\gamma }_{ij}}^{q}{r}_{j}}
\end{aligned},
\end{equation}
where $m_i$ and $c_i$ represent the degree centrality of node $i$ and the sum of the feature distinctions between node $i$ and its neighboring nodes respectively, ${\gamma_{ij}}^{r}$ and ${\gamma_{ij}}^{q}$ are two crucial parameters used to measure the connectivity (correlation) strength between nodes $i$ and $j$, $\gamma_{ij}=0$ indicates no connection. Clearly, $\gamma_{ij}$ is determined by both node degree and feature discrepancy. Naturally, they can be defined as
\begin{equation}
\begin{aligned}
\label{eq:16}
{{\gamma }_{ij}}^{r} &= {\varepsilon }_{ij}\frac{1}{\mathcal{K}_{i}}\\
{{\gamma }_{ij}}^{q} &= {\varepsilon }_{ij}\frac{1}{\mathcal{F}_{i}}\\
\end{aligned},    
\end{equation}
where $\varepsilon_{ij}$=$1$ indicates node $i$ is connected to node $j$, otherwise $\varepsilon_{ij}$ = $0$. $\mathcal{K}_i$ and $\mathcal{F}_i$ respectively represent the number of neighbors of node $i$ and the sum of feature distinctions between node $i$ and its neighbors. 

Referring to the weighted average state (Eq. \eqref{eq:9}), we set $q_{Gra}=\frac{\left \langle{\mathcal{F}} x \right \rangle}{\left \langle{\mathcal{F}} \right \rangle}$. The difference between $q_{Gra}$ and $r_{Gra}$ (equal to $x_{Gra}$) lies in that the dynamic state of each node $i$ is weighted by $\mathcal{F}_{i}$. Then, considering the entire graph, we have the following two condensed entanglement formulas:
\begin{equation}
\begin{aligned}
\label{eq:17}
\sum_{i=1}^{N} \sum_{j=1}^{N} \gamma_{ij}^{r} q_{j} &\cong \langle \gamma_{r} \rangle q_{Gra} \\
\sum_{i=1}^{N} \sum_{j=1}^{N} \gamma_{ij}^{q} r_{j} &\cong \langle \gamma_{q} \rangle r_{Gra}
\end{aligned}.
\end{equation}

To make the two crucial parameters accurately reflect the graph's perturbations, similarly, we leverage the weighted average manner to process $\gamma_{ij}^{r}$ and $\gamma_{ij}^{q}$,
\begin{equation}
\begin{aligned}
\label{eq:18}
\left \langle {\gamma }_{r}\right \rangle&=\frac{\textstyle\sum_{i=1}^{N}\textstyle\sum_{j=1}^{N}{{\gamma }_{ij}}^{r}\times {\mathcal{F}_{i}}}{\textstyle\sum_{i=1}^{N}{\mathcal{F}_{i}}} \\
\left \langle {\gamma }_{q}\right \rangle&=\frac{\textstyle\sum_{i=1}^{N}\textstyle\sum_{j=1}^{N}{{\gamma }_{ij}}^{q}\times {\mathcal{K}_i}}{\textstyle\sum_{i=1}^{N}{\mathcal{K}_i}}
\end{aligned},
\end{equation}
where $\left \langle \gamma_{r}\right \rangle$ and $\left \langle \gamma_{q}\right \rangle$ are obtained by computing the feature distinction-weighted average and the neighbor-count weighted average. In the same way, from the standpoint of dimension reduction, the perturbations in the unified form are reflected by $\left \langle \gamma_{r}\right \rangle$ and $\left \langle \gamma_{q}\right \rangle$. Finally, by calculating the average of the sum of degree centralities of all nodes and the average of the sum of feature distinctions of all nodes to their respective neighboring individuals, the average degree centrality $\mathcal{M}$ and average feature distinction $\mathcal{C}$ are inferred, the node-state dynamic-variation Eq. \eqref{eq:15} can be rewritten as
\begin{equation}
\begin{aligned}
\label{eq:19}
\frac {d{r}_{Gra}} {dt} &= -\mathcal{M}{r}_{Gra}+\frac{\left \langle {\gamma }_{r}\right \rangle \left({{q}_{Gra}}\right)^h}{1+\left \langle {\gamma }_{r}\right \rangle \left({{q}_{Gra}}\right)^h}\\
\frac {d{q}_{Gra}} {dt} &= -\mathcal{C}{q}_{Gra}+\frac{\left \langle {\gamma }_{q}\right \rangle \left({{r}_{Gra}}\right)^h}{1+\left \langle {\gamma }_{q}\right \rangle \left({{r}_{Gra}}\right)^h}
\end{aligned}.
\end{equation} 

\begin{center}
\fcolorbox{black}{gray!10}{\parbox{1.0\linewidth}{
\begin{remark}
Our method regards structure topology and node features, and searches for the ASEP for both of them. Eq. \eqref{eq:14} presents the topology and feature dynamics under adversarial perturbations. The subsequent Theorem \ref{Theorem: CondnsFormlaAsymStab} proves the asymptotic stability of two-dimensional dynamic-variation mapping function (Eq. \eqref{eq:19}) deduced from Eq. \eqref{eq:14}.
\end{remark}
}}
\end{center}
\end{sloppypar}  

\subsection{Generalizability}
Although our work provides a generalized theoretical framework for multi-dimensional functions, however, how to find a proper dynamic and ASEP-reserved function is still hard today. Thus, TopFeaRe utilizes HME theory to condense adversarial perturbations of whole graph as a 1D function. For graph data, HMF theory enables to leverage the basic in/out-degree of nodes to approximate the whole graph’s information into 1D condensed expression, thus for whatever type of graphs, as long as there exist nodes and edges, HMF can always be applied.
Furthermore, although Michaelis-Menten equation is used to design ASEP surface, we can alternatively choose another equation with the only requirement of satisfying the existence of ASEP. Thus, we think our theoretical assumptions can applied into diverse real-world graphs.

\section{Theoretical Analytics}
\subsection{Critical Attack-Resilient State Prerequisite}
During adversarial attack, each node would undergoes continuous dynamics, and it is hard to calibrate the node's concrete and instantaneous variation state in the course of continuous perturbations epoch by epoch. Hence, one viable solution is to map the whole graph's dynamic-variation process into a multi-object CDS, in this way, we can analyze the node-state dynamics from a global horizon. Moreover, for CDSs one of most important properties is the stability (a.k.a. equilibrium point). As stated in \cite{XinxinFan21}, a dynamic system has three statuses: non-stability, stability, and asymptotic stability, wherein the former implies the system departs from the original status and evolves into a chaotic state; the middle represents the system changes in a bounded scope; the latter indicates the system departs from the original status, and finally converges at an equilibrium point no matter how long it may experience. 

We think the three statuses of CDS can well-match the dynamics of graph under the continuous adversarial perturbations. Starting from adversarial defense, the more stable the dynamic system is, the more resilient it would be against adversarial attacks, thus we want to go a step further and anticipate the graph undergoing whatever perturbations can finally evolve into the asymptotically-stable state, corresponding to the ASEP. Toward this end, we introduce the following theorem to warrant the existence of asymptotic stability. 
\begin{sloppypar}
\begin{theorem} 
\label{Theorem: AsymptoticStability}
\textbf{$\text{(}$The Existence of Asymptotic Stability State of Graph \cite{XinxinFan21}$\text{)}$} 
The dynamic-variance Eq. \eqref{Equ:GeneDyn} has an asymptotically-stable equilibrium point if matrix $ \mathbf{A} $ is Hurwitz, the output $\mathord{\buildrel{\lower3pt\hbox{$\scriptscriptstyle\rightharpoonup$}}\over y} \left( t \right)$ and its associated nonlinear-mapping input $\phi \left( {\mathord{\buildrel{\lower3pt\hbox{$\scriptscriptstyle\rightharpoonup$}} \over y} \left( t \right)} \right) $ satisfy the condition that $ \mathord{\buildrel{\lower3pt\hbox{$\scriptscriptstyle\rightharpoonup$}} 
\over \phi } \left( {\mathord{\buildrel{\lower3pt\hbox{$\scriptscriptstyle\rightharpoonup$}} 
\over y} \left( t \right)} \right)^T  \cdot \left[ {\mathord{\buildrel{\lower3pt\hbox{$\scriptscriptstyle\rightharpoonup$}} 
\over y} \left( t \right) - M\mathord{\buildrel{\lower3pt\hbox{$\scriptscriptstyle\rightharpoonup$}} 
\over \phi } \left( {\mathord{\buildrel{\lower3pt\hbox{$\scriptscriptstyle\rightharpoonup$}} 
\over y} \left( t \right)} \right)} \right] > 0 $, where matrix $ M = diag\left( {{\raise0.7ex\hbox{$1$} \!\mathord{\left/{\vphantom {1 {k_1 }}}\right.\kern-\nulldelimiterspace}	\!\lower0.7ex\hbox{${k_1 }$}}, \cdots ,{\raise0.7ex\hbox{$1$} \!\mathord{\left/ {\vphantom {1 {k_p }}}\right.\kern-\nulldelimiterspace}
\!\lower0.7ex\hbox{${k_p }$}}} \right)$, $k_i$$>$0. Furthermore, given $ \psi  = diag\left( {\chi  _1 , \cdots ,\chi  _p } \right) $, there exists constant $\chi _i  \ge 0 $, such that $ (1 + \lambda _k \chi  _i ) \ne 0 $ for each eigenvalue $ \lambda _k $ of matrix $ \mathbf{A} $, 
%and meanwhile $ M + \left( {I + s\psi } \right)G\left( s \right) $ is strictly positive real.
and meanwhile, $ M + \left( {I + s\psi } \right) \mathbf{C}\left( {sI - \mathbf{A}} \right)^{ - 1} \mathbf{B} $ is strictly positive real. The number of diagonal elements in matrices $M$ and $\psi$, as well as the number of eigenvalues of matrix $\mathbf{A}$, depends on the dimension of graph.
\end{theorem}

\begin{center}
\fcolorbox{black}{gray!10}{\parbox{1.0\linewidth}{
\begin{remark}
\textbf{Theorem \ref{Theorem: AsymptoticStability}} demonstrates the generalized guarantee for the graph's asymptotic stability. Towards the condensed topology-feature-entangled dynamics formula, the prerequisite that the perturbed graph can finally converges into an ASEP is figured out, i.e. in what perturbation domain such ASEP can be found out. 
\end{remark}
}}
\end{center}

Next, we infer the following theorem to present the boundary of perturbation domain in which the ASEP can be located.
\begin{theorem}
\label{Theorem: PerturbDomain}
\textbf{$\text{(}$Boundary of Perturbation-Domain$\text{)}$} To achieve the asymptotic stability equilibrium state of perturbed graph, the boundary of adversarial perturbations must satisfy the inequality $\frac{\left({{x}_{Gra}}\right)^{h-1}}{1+ \left({{x}_{Gra}}\right)^{h}} \leqslant k$.
\end{theorem}

\textit{Proof.} See Appendix \ref{Prf_PerturbDomain}.

Thereby, by determining the maximum value of $ \frac{\phi \left( {\mathord{\buildrel{\lower3pt\hbox{$\scriptscriptstyle\rightharpoonup$}} \over y} \left( t \right)} \right)}{\mathord{\buildrel{\lower3pt\hbox{$\scriptscriptstyle\rightharpoonup$}} \over y} \left( t \right)} $, we can identify the perturbation domain to which the element $k$ belongs. If the graph is affected by perturbations from GAAs and $k$ falls into the perturbation domain, then the state of perturbed graph will ultimately converge to the ASEP. To obtain such critical state of adversarial resilience, we set $\frac {d{r}_{graph}} {dt}=0$ and $\frac {d{q}_{graph}} {dt}=0$, 
%we can obtain a steady-state solution equation similar to Equation \eqref{eq:13},
and calculate the topology-feature-entangled equations for $r_{Gra}$ and $q_{Gra}$ w.r.t. the parameters $\left \langle {\gamma }_{r}\right \rangle $ and $\left \langle {\gamma }_{q}\right \rangle$,  
\begin{equation}
\begin{aligned}
\label{eq:20}
{r}_{Gra} \left( \left \langle {\gamma }_{r}\right \rangle, \left \langle {\gamma }_{q}\right \rangle \right) &= \frac{\left \langle {\gamma }_{r}\right \rangle\left \langle {\gamma }_{q}\right \rangle+\mathcal{M}\mathcal{C}}{-\mathcal{M} \left(\mathcal{C}-\left \langle {\gamma }_{r}\right \rangle \right)\left \langle {\gamma }_{q}\right \rangle}\\
{q}_{Gra} \left(\left \langle {\gamma }_{r}\right \rangle, \left \langle {\gamma }_{q}\right \rangle\right) &= \frac{\left \langle {\gamma }_{r}\right \rangle\left \langle {\gamma }_{q}\right \rangle+\mathcal{M}\mathcal{C}}{\mathcal{C} \left(-\mathcal{M}-\left \langle {\gamma }_{q}\right \rangle \right)\left \langle {\gamma }_{r}\right \rangle}
\end{aligned}.
\end{equation}

Furthermore, we can adjust the variation of ${q}_{Gra}$, such that there exists a variable $\theta$ satisfying ${r}_{Gra} = \theta{q}_{Gra} $
\begin{equation}
\begin{aligned}
\label{eq:21}
{q}_{Gra}\left(\left \langle {\gamma }_{r}\right \rangle, \left \langle {\gamma }_{q}\right \rangle \right) &= \frac{\left \langle {\gamma }_{r}\right \rangle\left \langle {\gamma }_{q}\right \rangle+\mathcal{M}\mathcal{C}}{-\theta \mathcal{M}\left(\mathcal{C}-\left \langle {\gamma }_{r}\right \rangle \right)\left \langle {\gamma }_{q}\right \rangle}
\end{aligned}.
\end{equation}

We recognize the parts $\frac{\left \langle {\gamma }_{r}\right \rangle{q}_{Gra}}{1+\left \langle {\gamma }_{r}\right \rangle{q}_{Gra}}$ and $\frac{\left \langle {\gamma }_{q}\right \rangle{r}_{Gra}}{1+\left \langle {\gamma }_{q}\right \rangle{r}_{Gra}}$ of Eq. \eqref{eq:19} ($h$=1) as the nonlinear perturbations. Given that the control parameters $\left \langle {\gamma }_{r}\right \rangle$ and $\left \langle {\gamma }_{q}\right \rangle$ can be separated, thus we can obtain the approximate nonlinear parts $\frac{{q}_{Gra}}{1+{q}_{Gra}}$ and $\frac{{r}_{Gra}}{1+{r}_{Gra}}$. Since the number of elements in matrix $M$ depends on the graph's dimension, for one-dimensional mapping, the matrix $M$ contains only a single element, we decompose the two-dimensional entangled mapping functions into two one-dimensional mapping function with each describing the topology state and feature state respectively. We define $\frac{1}{k_{r}}$ and $\frac{1}{k_{q}}$ as the diagonal elements of the one-dimensional mapping function, corresponding to the matrix $M$ of the original $N$-dimensional graph. According to Theorem \ref{Theorem: PerturbDomain}, when $k_{r}$ and $k_{q}$ satisfy $k_{r} \geqslant \frac{1/\theta}{1+({r}_{Gra}/\theta)}$ and $k_{q}\geqslant \frac{\theta}{1+(\theta{r}_{Gra})}$, respectively, the two-dimensional entanglement function would have an ASEP. Clearly, as $r_{Gra} \to 0$, the minimum values of $k_{r}$ and $k_{q}$ can be attained. Therefore, the constants $k_{r}$ and $k_{q}$ in Theorem \ref{Theorem: PerturbDomain} belong to the intervals $[1/\theta, \infty)$ and $[\theta, \infty)$ respectively.

In addition, to prove the asymptotic stability of topology-feature-entangled Eq. \eqref{eq:19}, we offer the following theorem.

\begin{theorem}
\label{Theorem: CondnsFormlaAsymStab}
\textbf{$\text{(}$Asymptotic Stability of Two-Dimensional Dynamic-Variation Mapping Equation$\text{)}$}
Decompose the two-dimensional mapping function (Eq. \eqref{eq:19}) into two the one-dimensional and analyze the asymptotic stability of each separately. According to the structure of Eq. \eqref{Equ:GeneDyn}, let $\mathbf{A}_{r}$ and $\mathbf{B}_{r}$ correspond to $-\mathcal{M}$ and $\left \langle {\gamma }_{q}\right \rangle $ in the dynamic equations describing the state of graph topology, while $\mathbf{A}_{q}$ and $\mathbf{B}_{q}$ correspond to $-\mathcal{C}$ and $\left \langle {\gamma }_{r}\right \rangle$ in the dynamic equations describing the state of node feature. As long as $\mathbf{A}_{r}$, $\mathbf{B}_{r}$, $\mathbf{A}_{q}$ and $\mathbf{B}_{q}$ can form strictly positive real matrix $\frac{1}{k_{r}} + \left( {I + s\psi } \right) \mathbf{C}\left( {sI - \mathbf{A}_{r}} \right)^{ - 1} \mathbf{B}_{r}$ and $\frac{1}{k_{q}} + \left( {I + s\psi } \right) \mathbf{C}\left( {sI - \mathbf{A}_{r}} \right)^{ - 1} \mathbf{B}_{r}$, the two-dimensional dynamic equation exhibits asymptotic stability.
\end{theorem}  

\textit{Proof.} See Appendix \ref{Prf_CondnsFormlaAsymStab} 
\end{sloppypar}

\subsection{Surface of Asymptotic Stability}
\begin{sloppypar}
%We can construct a two-dimensional model using the Equation \eqref{eq:20} to represent the entire network. 
The aforementioned theoretical framework mainly presents that ASEP can be used as an indicator to reconstruct a robust graph, by which our method can purify the contaminated graph to restore its robust state.
The two variables $r_{Gra}$ and $q_{Gra}$ reflect the unified perturbation of the graph under adversarial perturbations, and the inferred trajectory points (ASEPs) $\left({r}_{Gra}, \left \langle {\gamma }_{r}\right \rangle, \left \langle {\gamma }_{q}\right \rangle \right)$ and $\left({q}_{Gra}, \left \langle {\gamma }_{r}\right \rangle, \left \langle {\gamma }_{q}\right \rangle \right)$ will form a surface in theory to satisfy the asymptotically-stable property. %In other words, such a situation denotes the graph's critical state of attack resilience.  

Note that, the values $r_{Gra}$ and $q_{Gra}$ in Eq. \eqref{eq:20} is our theoretical anticipation through bridging the mapping relationship between stability and resilience. To verify the correlation is reasonable and our defined entanglement function can indeed capture the perturbations, we exhibit the actual point distribution using real-world datasets. Towards that goal, the dataset Cora\_ML is utilized to output the surface under three popular GAAs, i.e. Metattack \cite{r4}, CE-PGD \cite{pgd}, and DICE \cite{dice}. The dataset statistics information is presented in Table \ref{tab:1}. %wherein the degree correlation is very low even negative for Cora\_ML and Cora, which indicates the our proposed HMF-based one-dimension condensed formula is reasonable to sketch the node-state dynamics of graph as a whole under continuous perturbations.
In our experiments, the perturbation rate increases gradually from from 3\% to 30\% with the increment 3\%, and the execution is depicted in Algorithm \ref{alg:1}, which aims to sketch an ASEP curved plane that well-matches the different-level adversarial perturbations, that is to say, the proposed theoretical framework provides the prerequisite on how to design a proper ASEP curved plane.

Algorithm \ref{alg:1} takes the adjacency matrix of graph data and its feature matrix as input, and outputs two dynamic variables and key parameters corresponding to the graph data. In the loop, we can obtain the structural state (degree centrality) $x$ of each node, the feature weight $f$ related to its node state, the number of neighboring nodes $n$ of each node, while $s_{out}$ serves as the structural weight of the node state. Finally, the weighted average of $\sum\nolimits_{j = 1}^N {\gamma _{ij}^r }$ and $\sum\nolimits_{j = 1}^N {\gamma _{ij}^q }$ is calculated. It can be seen that from Eq. \ref{eq:19} the weighted averages of the node states based on structural weights and feature weights are computed separately to obtain dynamic variables $r_{Gra}$ and $q_{Gra}$, The calculation of the two variables is adapted from Eq. \ref{eq:9}. 
In principle, for a graph regime, no matter degree centrality, betweenness, or closeness, the graph's robustness can be obtained once the ASEP obtained.
\end{sloppypar}

\begin{algorithm2e}[tb]
\caption{ASEP Distribution under Perturbations}
\label{alg:1}
\SetAlgoLined			% 增添end行
\DontPrintSemicolon		% 不显示行末尾的分号
\SetKwInOut{Input}{\textbf{Input}}		% Set the Input
\SetKwInOut{Output}{\textbf{Output}}	% set the Output

\Input{Graph $ G =\left(\mathcal{A},\mathcal{X}\right) $, $\mathcal{A}$ and $\mathcal{X}$ are adjacency and feature matrices}
\Output{$r_{Gra}, q_{Gra}, \left \langle \gamma_r \right \rangle, \left \langle \gamma_q \right \rangle $}

$r_{Gra} \gets 0, q_{Gra} \gets 0, \left \langle {\gamma }_{r} \right \rangle \gets 0, \left \langle {\gamma }_{q} \right \rangle \gets 0$ 

$S_{out} \gets \mathrm{sum}\left(\mathcal{A}, \mathrm{axis}=1\right)$ \;

\For{each \(\mathrm{node}\) in \( G \)} {
$x \gets \mathrm{getTopology}\left(\mathrm{node}, \mathcal{A}\right)$ \;
$f \gets \mathrm{getFeatureDistinction}\left(\mathrm{node}, \mathcal{X} \right)$\;
$n \gets \mathrm{getNeighbors}\left(\mathrm{node}, \mathcal{A}\right)$\;
$s\_vector.append(x)$\;
$f\_vector.append(f)$\;
$n\_vector.append(n)$\;
}

\For{each \(\mathrm{node}\) in \( G \)} {
$\left \langle {\gamma }_{r} \right \rangle \gets f\_vector[\mathrm{node}] / n\_vector[\mathrm{node}] + \left \langle {\gamma }_{r} \right \rangle$\;
$\left \langle {\gamma }_{q} \right \rangle \gets n\_vector[\mathrm{node}] / f\_vector[\mathrm{node}] + \left \langle {\gamma }_{q} \right \rangle$\;
}

$r_{Gra} \gets \mathrm{sum}(S_{out} \times s\_vector) / \mathrm{sum}(S_{out})$\;
$q_{Gra} \gets \mathrm{sum}(f\_vector \times s\_vector) / \mathrm{sum}(f\_vector)$\;
$\left \langle {\gamma }_{r} \right \rangle \gets \left \langle {\gamma }_{r} \right \rangle / \mathrm{sum}(f\_vector)$\;
$\left \langle {\gamma }_{q} \right \rangle \gets \left \langle {\gamma }_{q} \right \rangle / \mathrm{sum}(n\_vector)$\;

\Return \ \ \( r_{Gra}, q_{Gra}, \left \langle {\gamma }_{r} \right \rangle, \left \langle {\gamma }_{q} \right \rangle \)
\end{algorithm2e}

The experimental results are drawn in Fig. \ref{fig:x} and Fig. \ref{fig:q}, from which we observe that: i) different datasets may have different theoretical surfaces to present the entangled variables $r_{Gra}$ and $q_{Gra}$; ii) the topology state and feature state of the actual network Cora\_ML converge into the inferred 2D theoretical surface of our defined two-dimensional condensed formula, which implies there exists a latent correlation indeed between asymptotic stability and attack resilience, and our bridging is reasonable and effective; and iii) as the rate of adversarial perturbations ascends, the two entangled variables exhibit a decreasing tendency under the three adversarial attacks, which indicates the attack resilience gradually diminishes with the increase of perturbation strength.

\begin{figure}[tb]
	\centering
	\hspace*{0pt} 
	\subfigure[Metattack] 
	{
		\begin{minipage}[b]{.3\linewidth}
			\centering
			\includegraphics[width=1.14in, height=1.in]{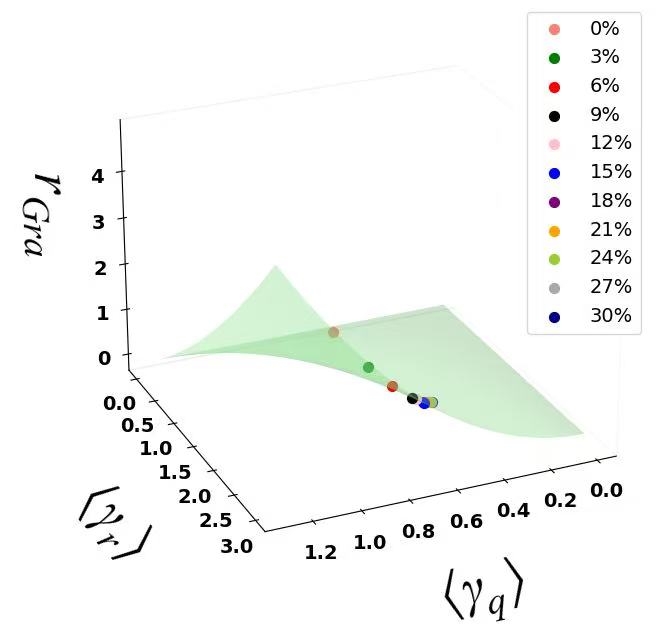}
		\end{minipage}
	}
	\hfill 
	\subfigure[CE-PGD] 
	{
		\begin{minipage}[b]{.3\linewidth}
			\centering
			\includegraphics[width=1.14in, height=1.in]{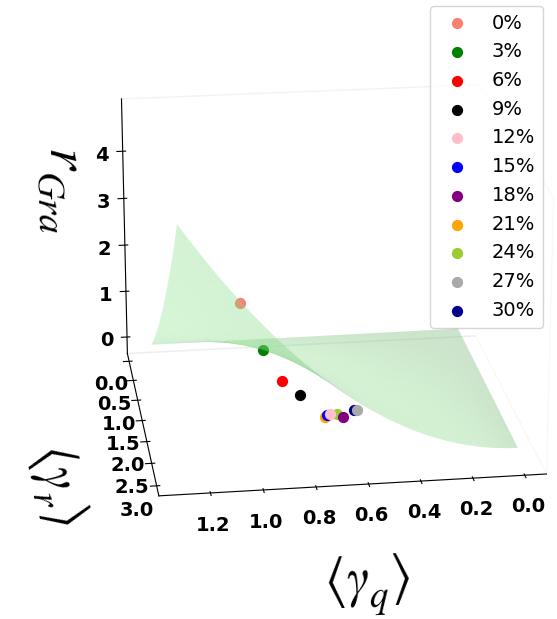}
		\end{minipage}
	}
	\hfill 
	\subfigure[DICE] 
	{
		\begin{minipage}[b]{.3\linewidth}
			\centering
			\includegraphics[width=1.14in, height=1.in]{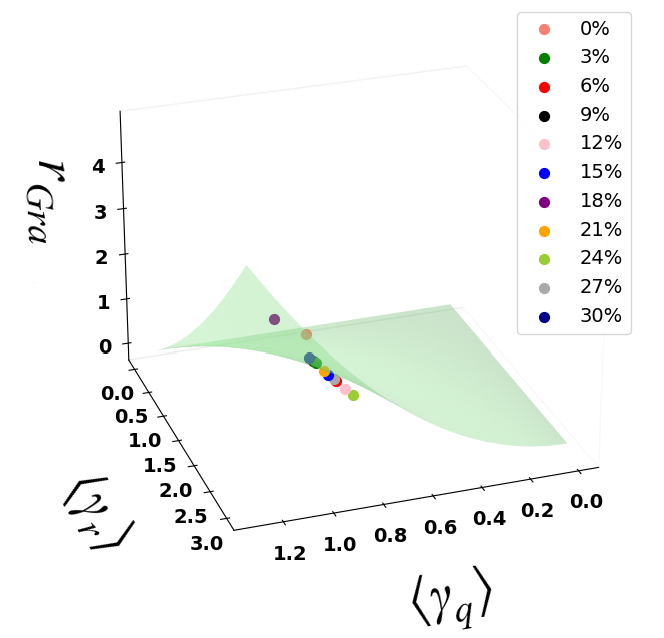}
		\end{minipage}
	}
	\caption{ASEP surface with $r_{Gra}$ on Cora\_ML}
	\label{fig:x}
\end{figure}

\begin{figure}[tb]
	\centering
	\hspace*{0pt} 
	\subfigure[Metattack] 
	{
		\begin{minipage}[b]{.3\linewidth}
			\centering
			\includegraphics[width=1.14in, height=1.in]{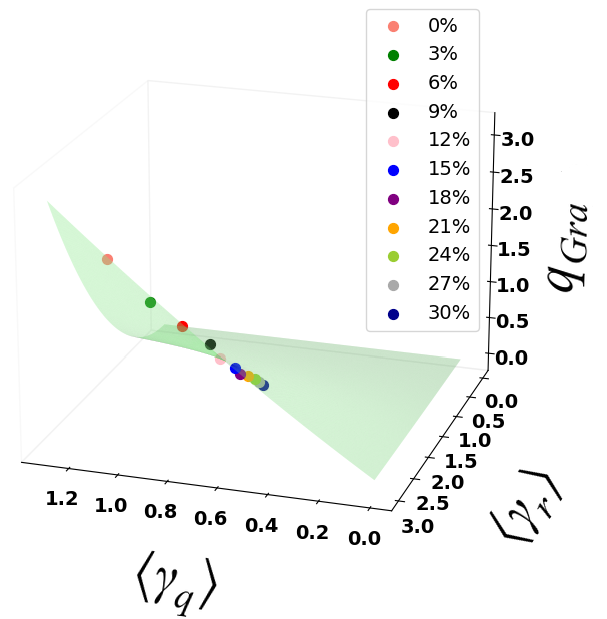}
		\end{minipage}
	}
	\hfill 
	\subfigure[CE-PGD] 
	{
		\begin{minipage}[b]{.3\linewidth}
			\centering
			\includegraphics[width=1.14in, height=1.in]{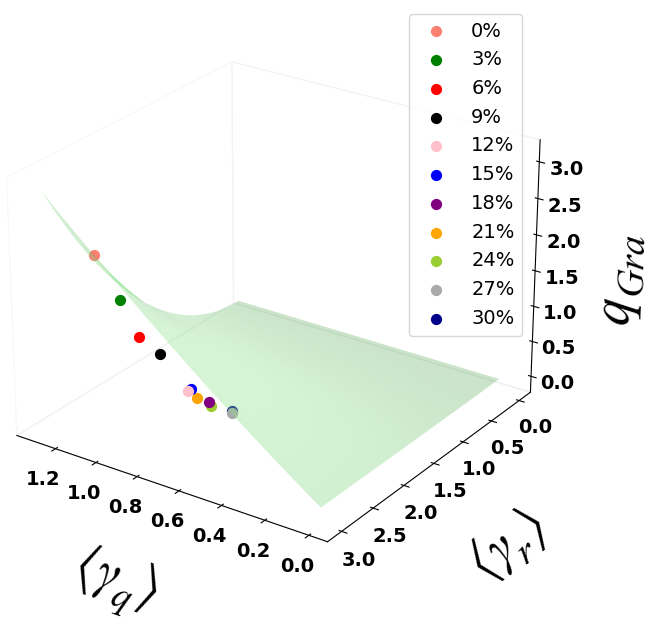}
		\end{minipage}
	}
	\hfill 
	\subfigure[DICE] 
	{
		\begin{minipage}[b]{.3\linewidth}
			\centering
			\includegraphics[width=1.14in, height=1.in]{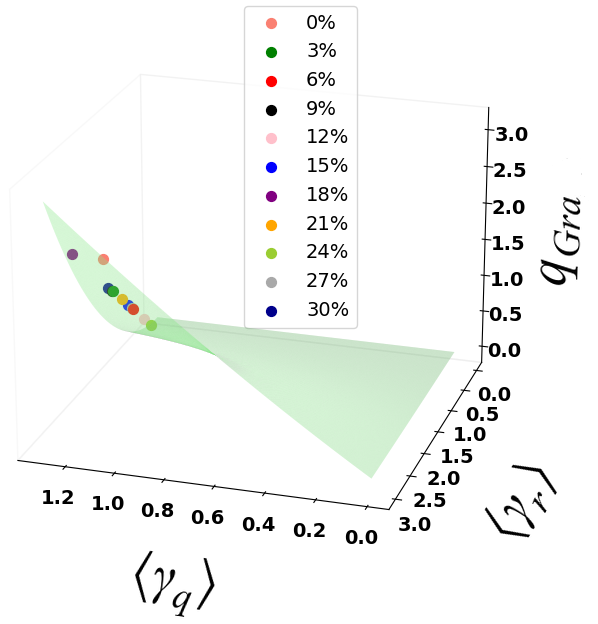}
		\end{minipage}
	}
	\caption{ASEP surface with $q_{Gra}$ on Cora\_ML}
	\label{fig:q}
\end{figure}

\section{Implementation}
\begin{sloppypar}
%Through the observation of the two-dimensional simplification model, we can subsequently implement a graph optimization-based defense method to restore the resilience of the graph. 
In this section, we detail the implementation on adversarial defense. After the entangled mapping of topology and feature, we know that the dual variables $\left(\left \langle \gamma_{r}\right \rangle, \left \langle \gamma_{q} \right \rangle \right)$ strictly affect the robustness. From Fig. \ref{fig:AttackGraphProperty} we see the GAAs induce the increase of adjacency-matrix's rank and singular values, as well as right-transfer of the density distribution of feature discrepancy. Hence, through this phenomenon, a conclusion can be drawn, that is, the GAAs aim to introduce perturbed edges into the original graph, and these edges are inclined to create links between those nodes that have dissimilar features. To counter that, we need to identify and remove such mischievously-perturbed edges. Towards this purpose, we at first give the following analysis on the basic GCN.
\end{sloppypar}

In order to achieve accurate prediction for graph data, GCN  \cite{gcnKipf} usually considers the neighbor's affection to learn node representation. The propagation between layers in neural network is defined as
%\begin{equation}
%    \label{eq:24}
$ {\hat H}^{(l+1)}=\sigma({\tilde{D}}^{-\frac{1}{2}}\tilde{A}{\tilde{D}}^{-\frac{1}{2}}{\hat H}^{(l)}{W}^{(l)}), $
%\end{equation}
where $\hat H$ represents the node features at each layer, $\tilde{A}$ is the sum of adjacency matrix and identity matrix, $\tilde{D}$ denotes the diagonal degree-matrix of $\tilde{A}$, and $W$ is the trainable weight matrix. In the process of aggregating node features, GCN updates the representation vector of each node through $l$-hop message passing \cite{k-hop}.
At $l = 0$, it means the node only aggregates its own feature information, while at $l = 1$, the node collects information from its 1-hop neighbors, and so on. Therefore, GCN relies on nearby nodes to execute representation learning for downstream task. If two nodes are similar, the learning effect of GCN on one node will be better when it transmits information to the other node. Conversely, if they are dissimilar, the information transfer leads to incorrect learning.

For each pair of nodes, we resort to their common neighboring nodes to determine their topology similarity. 
Instead of only considering the direct neighbors, a more complex approach are implemented, i.e. set $l = 2$. 
%We use the $k$-hop method to identify key nodes, where two nodes are considered similar if they share the same 2-hop neighbors. 
That is to say, we use the two sets of neighboring nodes within 2-hop as parameters to calculate Jaccard similarity, 
%\begin{equation}
%    \label{eq:25}
$ J\left({Set}_{i},{Set}_{j}\right)=\frac{\left |{Set}_{i}\cap {Set}_{j} \right |}{\left |{Set}_{i}\cup  {Set}_{j} \right |} $,
%\end{equation}
where $Set_i$ denotes the 2-hop neighboring individuals of node $i$. 
On the other hand, the cosine similarity is employed to calculate the pairwise feature similarity via the feature vectors,
%\begin{equation}
%    \label{eq:26}
$   C(f_{i},f_{j})=\frac{{f}_{i}\cdot {f}_{j}}{\lVert{f}_{i}\rVert \lVert{f}_{j} \rVert} $.
%\end{equation}
The larger the Jaccard similarity $J$ and cosine similarity $C$, the more similar the pairwise nodes $i$ and $j$. Then, to balance the affections of topology and feature in the procedure of node representation learning, we fuse the neighbor-based similarity and feature-based similarity in a proper way, i.e. assign different weights ${\mathcal{W}}_{J}$ and ${\mathcal{W}}_{C}$ to them,
\begin{equation}
\label{eq:27}
S({node}_{i},{node}_{j})={\mathcal{W}}_{J}\times J+{\mathcal{W}}_{C}\times C.
\end{equation}
% \begin{equation}
%     Similarity({node}_{i},{node}_{j})={\mathcal{W}}_{J}\times J({Set}_{i},{Set}_{j})+{\mathcal{W}}_{C}\times C({feature}_{i},{feature}_{j})
% \end{equation}

% \begin{equation}
%     (\textstyle\sum_{j=1}^{N}{{\gamma }_{(i+1)j}}^{(r)}>\textstyle\sum_{j=1}^{N}{{\gamma }_{ij}}^{(r)})\cup (\textstyle\sum_{j=1}^{N}{{\gamma }_{(i+1)j}}^{(q)}>\textstyle\sum_{j=1}^{N}{{\gamma }_{ij}}^{(q)})
% \end{equation}

\begin{sloppypar}
We calculate the similarity between all pairs of nodes and sort in ascending order in a list. By removing the lower similarity edges at the beginning of the list, we can optimize the graph topology into the attack-resilient state. However, to minimize the risk of mistakenly deleting important edges, the two pivotal variables $\left \langle {\gamma }_{r}\right \rangle$ and $\left \langle {\gamma }_{q}\right \rangle$ are referred as the criterion to decide whether to delete or not,
%Based on the known decreasing trend under GAAs, we can identify the nodes connected by the edges marked for removal. Furthermore, we can use the critical parameters of the nodes themselves as additional constraints to guide the decision on whether to remove low-similarity edges
%γr与γq是所有节点连接强度之和的，30式是反映j节点与周围相连节点的连接强度之和，在防御过程中，我们通过删除低相似度的边优化图结构来恢复图的韧性。如果一个节点j与其邻居节点的关系较强，那么该节点在图中就具有更大的影响力，该节点的特征信息在图卷积传播过程中对预测结果的影响更大，我们防御的目标是识别并移除那些通过扰动引入的低相似度边，这些边通常会使节点从错误的邻居那里接收到不准确的信息，引入30式约束条件能够找到既低相似度又连接强度弱的边。
\begin{equation}
\begin{aligned}
\label{eq:28}
\textstyle\sum_{j=1}^{N}{{\gamma }_{ij}}^{r}|_{(m+1)} &> \textstyle\sum_{j=1}^{N}{{\gamma }_{ij}}^{r}|_{m}\\
\textstyle\sum_{j=1}^{N}{{\gamma }_{ij}}^{q}|_{(m+1)} &> \textstyle\sum_{j=1}^{N}{{\gamma }_{ij}}^{q}|_{m}
\end{aligned}.
\end{equation}
\begin{algorithm2e}[tb]
	\caption{Adversarial Defense Implementation}
	\label{alg:2}
	\SetAlgoLined
	\DontPrintSemicolon
	\SetKwInOut{Input}{\textbf{Input}}
	\SetKwInOut{Output}{\textbf{Output}}
	
	\Input{Graph $G =(\mathcal{A}, \mathcal{X})$, Weights $\mathcal{W}_{J}$, $\mathcal{W}_{C}$, rate of adversarial perturbation $\alpha$}
	\Output{Attack-resilient graph $\hat{G}$}
	
	Initialize $\mathrm{node}_i^{(\gamma)}, \mathrm{node}_j^{(\gamma)} \gets 0, \mathrm{similarity\_list} \gets []$ 
	
	\For{each $\mathrm{node}_i, \mathrm{node}_j$ in $ G $} {
		$\mathrm{Set}_i, \mathrm{Set}_j \gets \mathrm{l-HopNeighbors}(\mathrm{node}_i, \mathrm{node}_j, \mathcal{A})$\;
		$\mathrm{f}_i, \mathrm{f}_j \gets \mathrm{getFeature}(\mathrm{node}_i, \mathrm{node}_j, \mathcal{X})$\;
		$ J \gets \mathrm{JaccardSimilarity}(\mathrm{Set}_i, \mathrm{Set}_j)$\;
		$ C \gets \mathrm{cosineSimilarity}(\mathrm{f}_i, \mathrm{f}_j)$\;
		$ S \gets \mathcal{W}_{J} \times J + \mathcal{W}_{C} \times C $\;
		$ \mathrm{similarity\_list.append}(S)$\;
	}    
	
	$ L \gets \mathrm{sortAscending}(\mathrm{similarity\_list}) $\;
	$ \mathrm{initial\_length} \gets \mathrm{len}(L) $\;
	
	\While{ $\mathrm{True} $} {
		$ \hat{G}, \mathrm{node}_i, \mathrm{node}_j, L \gets \mathrm{removeEdge}(G, L)$\;
		$ \mathrm{node}_i^{(\gamma)}, \mathrm{node}_j^{(\gamma)} \gets \mathrm{getValues}\left(\mathcal{A}, \mathrm{node}_i, \mathrm{node}_j \right)$\;
		
		\If{$ \mathrm{node}_i^{(\gamma)}|_{m} < \mathrm{node}_i^{(\gamma)}|_{(m-1)} $ \textbf{or} $ \mathrm{node}_j^{(\gamma)}|_{m} < \mathrm{node}_{j}^{(\gamma)}|_{(m-1)} $} {
			$ \hat{G}, L \gets \mathrm{restoreEdge}(\mathrm{node}_i, \mathrm{node}_j)$\;
		}
		
		\If{ $ \mathrm{len}(L) < \alpha \times \mathrm{initial\_length} $} {
			break\;
		}
	}
	\Return \ \ $ \hat{G} $ 
\end{algorithm2e}

%That is to say, %after the $m$-th edge deletion operation, the topology variable of the node becomes $\textstyle\sum_{j=1}^{N}{{\gamma }_{ij}}^{(r)}$. 
The condition to execute the ($m$+1)th edge-deletion operation must satisfy the above constraint, 
%$\textstyle\sum_{j=1}^{N}{{\gamma }_{ij}}^{(r)}|_{(m+1)} > \textstyle\sum_{j=1}^{N}{{\gamma }_{ij}}^{(r)}|_m$; 
otherwise the deleted edge will be restored. The changes in the feature parameters follow a similar logic, satisfying either of the two constraint conditions is sufficient. The implementation procedure is described in Algorithm \ref{alg:2}.
It takes the adjacency matrix of graph, feature matrix, weights on Jaccard similarity and cosine similarity, an adjustable optimization ratio $\alpha$ as input, and outputs the optimized graph in its desired robust state. First, calculate the weighted similarity between nodes using Eq. \ref{eq:28} and store them in a similarity list in ascending order. Furthermore, within the loop identify and remove undesired edges. Each time an edge with low similarity over nodes is removed, the connection strength between nodes, namely ${\gamma_{ij}^r}$ and ${\gamma_{ij}^q}$, will be calculated in each iteration of the loop. If, after removing an edge with low similarity, it is found that the connection strength between these two nodes is actually weaker than it was before the removal. The constraint constructs by Eq. \ref{eq:27}. Therefore, the previously-removed but desired edge will be restored, and parameter $\alpha$ determines the optimization ratio.

\end{sloppypar}

\section{Experiment Evaluation}
%\begin{sloppypar}
%We evaluate the performance of our method through answering four important research questions.
%\begin{itemize}
%\item \textcolor{red}{\textbf{RQ1:} Compared to the conventional preprocessing-base baselines, how does our ASEP-referred adjacency matrix-purification method perform?}
%\item \textbf{RQ2} Does the purified image help different GNNs learn better?
%\item \textbf{RQ3} While enhancing the adversarial robustness of the image, do its graph attribute rank, singular values, and feature smoothness optimize towards that of a clean image?
%\item \textbf{RQ4} Our defense strategy is based on prior knowledge of changes in graph robustness in two-dimensional models, and these disturbances fall under non-targeted attacks. Is our method equally effective when facing targeted attacks?
%\end{itemize}
%\end{sloppypar}

\subsection{Configuration}

\begin{table}[tb]
\caption{Dataset Statistics}
\label{tab:1}
\fontsize{9}{10}\selectfont
\begin{tabular}{cccccc}
\toprule
% &Feature&Classification&Degree Correlation
\textbf{Dataset}&\textbf{Nodes}&\textbf{Edges}&\textbf{Feat.}&\textbf{Classes}&\textbf{Deg. Cor.}\\
\midrule
Cora\_ML & 2810 & 7981 & 2879 & 7 & -0.0855\\
Cora & 2485 & 5069 & 1433 & 7 & -0.0588\\
Citeseer & 2110 & 3668 & 3703 & 6 & 0.0438\\
Photo & 7487 & 119043 & 745 & 8 & 0.0247\\
Pubmed & 19717 & 44325 & 500 & 3 & 0.0395\\

\bottomrule
\end{tabular}
\end{table}

The configurations are detailed in Appendix \ref{EperimentConfig}: i) \textbf{Datasets}: five commonly-used scalable realistic graph datasets: Cora, Cora\_ML, Citeseer, Amazon Photo, and PubMed. The statistics are sketched in Table \ref{tab:1}; ii) \textbf{GAAs:} three non-targeted attacks Metattack \cite{r4}, CE-PGD \cite{pgd}, and DICE \cite{dice}, as well as the targeted attack Nettack\cite{nettack}; iii) \textbf{Baselines:} three neural network architectures-focused models: GCN \cite{gcnKipf}, GAT \cite{GATVelickovic}, and HANG \cite{hang}, as well as three purification/preprocessing-based models: GCN-SVD \cite{svd}, GCN-Jaccard \cite{jaccard}, and Mid-GCN \cite{HuangJin25}; and iv) \textbf{Execution Settings:} the running environment and hyperparameters settings are detailed.

\begin{table*}[htbp]
\centering
%\caption{ACCURACY OF NODE CLASSIFICATION ON Cora\_ML (\textbf{BOLD}-THE BEST)}
\caption{Accuracy of node classification on Cora\_ML (\textbf{Bold}-the best)}
\label{tab:Cora_ML1}
% \begin{tabular}{|c|c|c|c|c|c|c|c|} % 三列，均为居中对齐
	\begin{tabular}{|>{\centering\arraybackslash}m{0.6cm}|>{\centering\arraybackslash}m{1.95cm}|>{\centering\arraybackslash}m{1.98cm}|>{\centering\arraybackslash}m{1.98cm}|>{\centering\arraybackslash}m{1.98cm}|>{\centering\arraybackslash}m{1.98cm}|>{\centering\arraybackslash}m{1.98cm}|>{\centering\arraybackslash}m{1.98cm}|}
		\hline
		\textbf{GAA} & \diagbox{{\textbf{Mod.}}}{{\textbf{RAP}}} & \textbf{0\%} & \textbf{5\%} & \textbf{10\%} & \textbf{15\%} & \textbf{20\%} & \textbf{25\%}\\ % 表头
		\hline
		
		\multirow{3}{*}{\makecell[c]{Meta\\ttack}} & GCN-SVD & 0.8271$\pm$0.0080 & 0.8164$\pm$0.0084 & 0.8054$\pm$0.0092 & 0.7857$\pm$0.0095 & 0.6967$\pm$0.1935 & 0.6279$\pm$0.1729\\ % 第一行
		\cline{2-8} % 只画第二行到第八列的横线
		& GCN-Jaccard &0.8423$\pm$0.0094&0.7896$\pm$0.0087&0.7441$\pm$0.0123&0.7028$\pm$0.0145&0.6552$\pm$0.0294&0.5923$\pm$0.0471\\ % 第二行
		\cline{2-8} % 只画第二行到第八列的横线
		& TopFeaRe &  {\textbf{0.8524$\pm$0.0096}} \(\textcolor{red}{(\uparrow \textbf{1.01\%})}\)& {\textbf{0.8293$\pm$0.0088}} \(\textcolor{red}{(\uparrow \textbf{1.29\%})}\)& {\textbf{0.8219$\pm$0.0114}} \(\textcolor{red}{(\uparrow \textbf{1.65\%})}\)& {\textbf{0.8150$\pm$0.0107}} \(\textcolor{red}{(\uparrow \textbf{2.93\%})}\)& {\textbf{0.8086$\pm$0.0104}} \(\textcolor{red}{(\uparrow \textbf{11.19\%})}\) & {\textbf{0.7920$\pm$0.0128}} \(\textcolor{red}{(\uparrow \textbf{16.41\%})}\)\\ % 第三行
		
		\hline
		
		\multirow{3}{*}{\makecell[c]{CE-\\PGD}} & GCN-SVD & 0.8271$\pm$0.0080 & 0.8229$\pm$0.0066 & 0.8164$\pm$0.0080 & 0.8118$\pm$0.0122 & 0.8110$\pm$0.0089 & 0.8080$\pm$0.0055\\ % 第一行
		\cline{2-8} % 只画第二行到第八列的横线
		& GCN-Jaccard & 0.8423$\pm$0.0094 & 0.8334$\pm$0.0106 & 0.8248$\pm$0.0112 & 0.8153$\pm$0.0086 & 0.8131$\pm$0.0109 & 0.8117$\pm$0.0082\\ % 第二行
		\cline{2-8} % 只画第二行到第八列的横线
		& TopFeaRe &  {\textbf{0.8524$\pm$0.0096}} \(\textcolor{red}{(\uparrow \textbf{1.01\%})}\)& {\textbf{0.8357$\pm$0.0092}} \(\textcolor{red}{(\uparrow \textbf{0.23\%})}\)& {\textbf{0.8327$\pm$0.0112}} \(\textcolor{red}{(\uparrow \textbf{0.79\%})}\)& {\textbf{0.8307$\pm$0.0086}} \(\textcolor{red}{(\uparrow \textbf{1.54\%})}\)& {\textbf{0.8305$\pm$0.0092}} \(\textcolor{red}{(\uparrow \textbf{1.74\%})}\) & {\textbf{0.8294$\pm$0.0090}} \(\textcolor{red}{(\uparrow \textbf{1.77\%})}\)\\ % 第三行
		
		\hline
		
		\multirow{3}{*}{DICE} & GCN-SVD & 0.8271$\pm$0.0080 & 0.8148$\pm$0.0065 & 0.8000$\pm$0.0064 & 0.7901$\pm$0.0089 & 0.7730$\pm$0.0119 & 0.7590$\pm$0.0087\\ % 第一行
		\cline{2-8} % 只画第二行到第八列的横线
		& GCN-Jaccard & 0.8423$\pm$0.0094 & 0.8300$\pm$0.0092 & 0.8177$\pm$0.0101 & 0.8061$\pm$0.0088 & 0.7915$\pm$0.0135 & 0.7752$\pm$0.0128\\ % 第二行
		\cline{2-8} % 只画第二行到第八列的横线
		& TopFeaRe &  {\textbf{0.8524$\pm$0.0096}} \(\textcolor{red}{(\uparrow \textbf{1.01\%})}\)& {\textbf{0.8354$\pm$0.0089}} \(\textcolor{red}{(\uparrow \textbf{0.54\%})}\)& {\textbf{0.8291$\pm$0.0113}} \(\textcolor{red}{(\uparrow \textbf{1.14\%})}\)& {\textbf{0.8242$\pm$0.0103}} \(\textcolor{red}{(\uparrow \textbf{1.81\%})}\)& {\textbf{0.8160$\pm$0.0122}} \(\textcolor{red}{(\uparrow \textbf{2.45\%})}\) & {\textbf{0.8078$\pm$0.0079}} \(\textcolor{red}{(\uparrow \textbf{3.26\%})}\)\\ % 第三行
		
		\hline
	\end{tabular}
\end{table*}

\begin{table*}[htbp]
\centering
%\caption{ACCURACY OF NODE CLASSIFICATION UNDER COMBINATION WITH Ours USING METATTACK (\textbf{BOLD}-THE BEST)}
\caption{Accuracy of node classification on Citeseer under combination with Ours using Metattack}
\label{tab:Combine-Citeseer}
% \begin{tabular}{|c|c|c|c|c|c|c|c|} % 三列，均为居中对齐
\begin{tabular}{|>{\centering\arraybackslash}m{0.6cm}|>{\centering\arraybackslash}m{1.95cm}|>{\centering\arraybackslash}m{1.98cm}|>{\centering\arraybackslash}m{1.98cm}|>{\centering\arraybackslash}m{1.98cm}|>{\centering\arraybackslash}m{1.98cm}|>{\centering\arraybackslash}m{1.98cm}|>{\centering\arraybackslash}m{1.98cm}|}
\hline
\textbf{\makecell[c]{Data\\set}} & \diagbox{{\textbf{Mod.}}}{{\textbf{RAP}}} & \textbf{0\%} & \textbf{5\%} & \textbf{10\%} & \textbf{15\%} & \textbf{20\%} & \textbf{25\%}\\ % 表头

\hline
\multirow{6}{*}{\makecell[c]{Cite\\seer}} & GCN & 0.7211$\pm$0.0099 & 0.6661$\pm$0.0219 & 0.6132$\pm$0.0325 & 0.5538$\pm$0.0273 & 0.4935$\pm$0.0261 & 0.4556$\pm$0.0318\\ 
\cline{2-8} 
& TopFeaRe &  {\textbf{0.7228$\pm$0.0105}} \(\textcolor{red}{(\uparrow \textbf{0.17\%})}\)& {\textbf{0.7159$\pm$0.0170}} \(\textcolor{red}{(\uparrow \textbf{4.98\%})}\)&  {\textbf{0.7005$\pm$0.0162}} \(\textcolor{red}{(\uparrow \textbf{8.73\%})}\)& {\textbf{0.6885$\pm$0.0170}} \(\textcolor{red}{(\uparrow \textbf{13.47\%})}\)& {\textbf{0.6754$\pm$0.0147}} \(\textcolor{red}{(\uparrow \textbf{18.19\%})}\) & {\textbf{0.6647$\pm$0.0189}} \(\textcolor{red}{(\uparrow \textbf{20.91\%})}\)\\ 
\cline{2-8}
& GAT & 0.7335$\pm$0.0084&0.7140$\pm$0.0161&0.6876$\pm$0.0256&0.6511$\pm$0.0301&0.6170$\pm$0.0203&0.5960$\pm$0.0283\\ 
\cline{2-8} 
& GAT-Our &  {\textbf{0.7393$\pm$0.0127}} \(\textcolor{red}{(\uparrow \textbf{0.58\%})}\)& {\textbf{0.7274$\pm$0.0094}} \(\textcolor{red}{(\uparrow \textbf{1.34\%})}\)& {\textbf{0.7210$\pm$0.0120}} \(\textcolor{red}{(\uparrow \textbf{3.34\%})}\)& {\textbf{0.7105$\pm$0.0099}} \(\textcolor{red}{(\uparrow \textbf{5.94\%})}\)& {\textbf{0.6991$\pm$0.0137}} \(\textcolor{red}{(\uparrow \textbf{8.21\%})}\) & {\textbf{0.6934$\pm$0.0201}} \(\textcolor{red}{(\uparrow \textbf{9.74\%})}\)\\
\cline{2-8}
& HANG & \textbf{0.7406$\pm$0.0116}&0.7268$\pm$0.0117&0.7075$\pm$0.0272&0.6851$\pm$0.0156&0.6722$\pm$0.0214&0.6678$\pm$0.0305\\ 
\cline{2-8} 
& HANG-Our &  {0.7399$\pm$0.0104} \(\textcolor{red}{(\downarrow \textbf{0.07\%})}\)& {\textbf{0.7348$\pm$0.0126}} \(\textcolor{red}{(\uparrow \textbf{0.80\%})}\)& {\textbf{0.7281$\pm$0.0118}} \(\textcolor{red}{(\uparrow \textbf{2.06\%})}\)& {\textbf{0.7160$\pm$0.0106}} \(\textcolor{red}{(\uparrow \textbf{3.09\%})}\)& {\textbf{0.7095$\pm$0.0145}} \(\textcolor{red}{(\uparrow \textbf{3.73\%})}\) & {\textbf{0.7093$\pm$0.0110}} \(\textcolor{red}{(\uparrow \textbf{4.15\%})}\)\\
\cline{2-8}
& MidGCN &      0.7452$\pm$0.0032&0.7237$\pm$0.0067&0.7105$\pm$0.0066&0.6584$\pm$0.0106&0.6700$\pm$0.0052&0.6306$\pm$0.0099\\ 
\cline{2-8} 
& MidGCN-Our &  {\textbf{0.7466$\pm$0.0047}} \(\textcolor{red}{(\uparrow \textbf{0.14\%})}\)& {\textbf{0.7307$\pm$0.0027}} \(\textcolor{red}{(\uparrow \textbf{0.70\%})}\)& {\textbf{0.7127$\pm$0.0080}} \(\textcolor{red}{(\uparrow \textbf{0.22\%})}\)& {\textbf{0.6626$\pm$0.0102}} \(\textcolor{red}{(\uparrow \textbf{0.42\%})}\)& {\textbf{0.6712$\pm$0.0054}} \(\textcolor{red}{(\uparrow \textbf{0.12\%})}\) & {\textbf{0.6349$\pm$0.0134}} \(\textcolor{red}{(\uparrow \textbf{0.43\%})}\)\\

\hline
\end{tabular}
\end{table*}

\begin{table*}[htbp]
\centering
\caption{Ablation experiment on node classification accuracy under Metattack}
\label{tab:Eq.}
% \begin{tabular}{|c|c|c|c|c|c|c|c|} % 三列，均为居中对齐
\begin{tabular}{|>{\centering\arraybackslash}m{0.65cm}|>{\centering\arraybackslash}m{1.88cm}|>{\centering\arraybackslash}m{1.98cm}|>{\centering\arraybackslash}m{1.98cm}|>{\centering\arraybackslash}m{1.98cm}|>{\centering\arraybackslash}m{1.98cm}|>{\centering\arraybackslash}m{1.98cm}|>{\centering\arraybackslash}m{1.98cm}|}
\hline
\textbf{\makecell[c]{Data\\set}} & \diagbox{{\textbf{Mod.}}}{{\textbf{RAP}}} & \textbf{0\%} & \textbf{5\%} & \textbf{10\%} & \textbf{15\%} & \textbf{20\%} & \textbf{25\%}\\ % 表头
\hline

\multirow{3}{*}{\makecell[c]{Cora\\\_ML}} & \makecell[c]{w/: Feature} 
& 0.8512$\pm$0.0033 & 0.8160$\pm$0.0045 & 0.7576$\pm$0.0031 & 0.7182$\pm$0.0059 & 0.6963$\pm$0.0087 & 0.6307±0.0082\\ % 第一行
\cline{2-8} % 只画第二行到第八列的横线
& \makecell[c]{w/: Topology} &0.8515$\pm$0.0031&0.8150$\pm$0.0043&0.7594$\pm$0.0030&0.7241$\pm$0.0041&0.7011$\pm$0.0071&0.6415$\pm$0.0081\\ % 第二行
\cline{2-8} % 只画第二行到第八列的横线
& w/: Both 
& {\textbf{0.8524$\pm$0.0096}} \(\textcolor{red}{(\uparrow \textbf{0.09\%})}\)& {\textbf{0.8293$\pm$0.00886}} \(\textcolor{red}{(\uparrow \textbf{1.33\%})}\)& {\textbf{0.8219$\pm$0.0114}} \(\textcolor{red}{(\uparrow \textbf{6.25\%})}\)& {\textbf{0.8150$\pm$0.0107}} \(\textcolor{red}{(\uparrow \textbf{9.09\%})}\)& {\textbf{0.8086$\pm$0.0104}} \(\textcolor{red}{(\uparrow \textbf{10.75\%})}\) & {\textbf{0.7920$\pm$0.0128}} \(\textcolor{red}{(\uparrow \textbf{15.05\%})}\)\\ % 第三行

\hline

\multirow{3}{*}{Cora} & \makecell[c]{w/: Feature}
& 0.8230$\pm$0.0047 & 0.7864$\pm$0.0152 & 0.7358$\pm$0.0116 & 0.6309$\pm$0.0133 & 0.6282$\pm$0.0189 & 0.6282$\pm$0.0189\\ % 第一行
\cline{2-8} % 只画第二行到第八列的横线
& \makecell{w/: Topology} &0.8232$\pm$0.0044&0.7893$\pm$0.0129&0.7364$\pm$0.0093&0.6242$\pm$0.0112&0.6322$\pm$0.0221&0.5451$\pm$0.0231\\ % 第二行
\cline{2-8} % 只画第二行到第八列的横线
& w/: Both &  {\textbf{0.8326$\pm$0.0083}} \(\textcolor{red}{(\uparrow \textbf{0.94\%})}\)& {\textbf{0.8032$\pm$0.0128}} \(\textcolor{red}{(\uparrow \textbf{1.39\%})}\)& {\textbf{0.7886$\pm$0.0131}} \(\textcolor{red}{(\uparrow \textbf{5.22\%})}\)& {\textbf{0.7807$\pm$0.012}} \(\textcolor{red}{(\uparrow \textbf{14.98\%})}\)& {\textbf{0.7709$\pm$0.0126}} \(\textcolor{red}{(\uparrow \textbf{13.87\%})}\) & {\textbf{0.7920$\pm$0.0128}} \(\textcolor{red}{(\uparrow \textbf{16.38\%})}\)\\ % 第三行

\hline

\multirow{3}{*}{\makecell[c]{Cite\\seer}} & \makecell[c]{w/: Feature}
& 0.7107$\pm$0.0023 & 0.6857$\pm$0.0062 & 0.6578$\pm$0.0115 & 0.6174$\pm$0.0079 & 0.6192$\pm$0.0071 & 0.5755$\pm$0.0092\\ % 第一行
\cline{2-8} % 只画第二行到第八列的横线
& \makecell[c]{w/: Topology} &0.6192$\pm$0.0071&0.6857$\pm$0.0071&0.6857$\pm$0.0071&0.6067$\pm$0.0076&0.6083$\pm$0.0098&0.5566$\pm$0.0160\\ % 第二行
\cline{2-8} % 只画第二行到第八列的横线
& w/: Both &  {\textbf{0.7228$\pm$0.0105}} \(\textcolor{red}{(\uparrow \textbf{1.21\%})}\)& {\textbf{0.7159$\pm$0.017}} \(\textcolor{red}{(\uparrow \textbf{3.02\%})}\)& {\textbf{0.7005$\pm$0.0162}} \(\textcolor{red}{(\uparrow \textbf{1.48\%})}\)& {\textbf{0.6885$\pm$0.017}} \(\textcolor{red}{(\uparrow \textbf{7.11\%})}\)& {\textbf{0.6754$\pm$0.0147}} \(\textcolor{red}{(\uparrow \textbf{5.62\%})}\) & {\textbf{0.6647$\pm$0.0189}} \(\textcolor{red}{(\uparrow \textbf{8.92\%})}\)\\ % 第三行

\hline
\end{tabular}
\end{table*}

\subsection{Performance}
\textbf{Comparison with Preprocessing-based Baselines.} First, we compare our approach to two adjacency-matrix-based preprocessing approaches: GCN-SVD and GCN-Jaccard. Since our TopFeaRe enhances the adversarial resilience through introducing constraints (Eq. \eqref{eq:28}) under the guidance of ASEP during the purification of the adjacency matrix, we input the purified graphs into the GCN to execute node-classification tasks. Table \ref{tab:Cora_ML1} and Tables \ref{tab:Cora1}-\ref{tab:PubMed} (Appendix \ref{Exp:Preprocess&Neural}) display the accuracy across the five datasets under three non-targeted GAAs.
\par
From the experimental results, we can observe that: i) our TopFeaRe significantly enhances the adversarial resilience compared to GCN-SVD and GCN under all the three adversarial attacks. For example, on Cora\_ML under the RAP of 25\%, TopFeaRe outperforms the second-best baseline by 16.41\%, 1.77\%, and 3.26\%, respectively. Furthermore, at RAP 0\%, which means there are no adversarial perturbations, our method can further promote the adversarial resilience of the clean graph by removing a small number of dissimilar edges. This indicates our proposed method not only performs exceptionally well in adversarial environments, but it also enhances the learning capability of the graph topology under clean conditions. 
In other words, it is unnecessary to know clean graph and attack signals beforehand, once a graph (even contaminated) is obtained, we can employ existing attacks to draw ASEP surface, then enhance it by such ASEP.

%Additionally, regarding the Metattack attack, it inflicts greater damage on the graph topology compared to CE-PGD and DICE. The results also show that under such severe adversarial interference, TopFeaRe can effectively purify the perturbed graph, giving rise to the significant improvement of adversarial resilience. 
%This finding answers \textbf{RQ1} and underscores the importance of graph purification mitigation while facing strong adversarial attacks.
%\subsubsection{Effects of Purified Graphs on Learning}

\textbf{Comparison with Neural Network-Optimized Baselines.} 
%To study \textbf{RQ2}, 
We combine our method into GCN, GAT, HANG, and Mid-GCN, and compare to their original version under Metattack attack, that is to say, we in advance modify the adjacency matrix referring to ASEP, then perform the baselines. The experimental results are shown in Table \ref{tab:Combine-Citeseer} and Tables \ref{tab:Combine-Cora_ML}-\ref{tab:Combine-PubMed} (Appendix \ref{Exp:Preprocess&Neural}), from which, we can see that our TopFeaRe can significantly boost the adversarial resilience of these neural network-optimized baselines under RAP ranging from 5\%-25\%, with the accuracy promotion of GCN, GAT, HANG, and Mid-GCN on average [13.095\%, 8.67\%, 1.06\%, 0.43\%] on Cora\_ML, [18.75\%, 7.78\%, 4.04\%, 1.81\%] on Cora, [13.26\%, 5.71\%, 2.77\%, 0.38\%] on Citeseer, [4.54\%, 9.39\%, 1.01\%, 1.23\%] on Amazon Photo, and [7.30\%, 14.86\%, 7.29\%, 2.26\%] on PubMed respectively.

%When no adversarial attacks exist, i.e. RAP 0\%, only HANG-Our showed a slight decrease of 0.07\%, 0.94\%, and 0.01\% on the datasets Citeseer, Amazon Photo, and PubMed, while the other cases all show clear improvements. Under adversarial perturbations, as RAP increases, the resistance assistance provided by our method becomes more pronounced. For example, on Cora dataset, TopFeaRe improves GCN by 31.24\%, GAT-Our promotes GAT by 15.26\%, and HANG-Our boosts HANG by 7.25\%.

\par
%It is noteworthy that HANG itself has strong robustness against adversarial perturbations. Therefore, in the combination of TopFeaRe and HANG, we observe that the less robust GNNs (such as GCN) benefit more significantly from TopFeaRe. This indicates that TopFeaRe not only enhances the adversarial resilience of the models but also effectively improves the performance of different GNNs, particularly when facing stronger adversarial attacks. Additionally, models combined with TopFeaRe demonstrate better learning capabilities when handling complex graph structures, further validating the effectiveness of TopFeaRe in enhancing the robustness of GNNs.

\subsection{Ablation Studies}
\textbf{Effect of Topology and Feature.} 
To inspect the effects of topology and feature, we also run a set of ablation experiments from facets: i) only using graph topology-based ASEP to guide the purification of adjacency matrix, called TopRe; and ii) only using node features-based ASEP to guide the purification of adjacency matrix, called FeaRe. The results are listed in Table \ref{tab:Eq.}, from which we can observe that for the three datasets, using the single feature (topology or feature) can obtain similar performance, however, when both are combined, the accuracy can promote by 1.33\% to 10.75\% for Cora\_ML, 1.39\% to 16.41\% for Cora, and 3.02\% to 8.92\% for Citeseer under the adversarial attack Metattack. 
Additionally, we run anther set of experiments with/without two parameters: degree centrality and feature distinction. The performance in Table \ref{tab:WithoutPara} is obviously better in consideration of the parameters.

\begin{table*}[bt]
\centering
\caption{Accuracy of node classification w/(w/o) degree centrality and feature distinction under Metattack}
\label{tab:WithoutPara}
\begin{tabular}{|>{\centering\arraybackslash}m{1.2cm}|>{\centering\arraybackslash}m{1.98cm}|>{\centering\arraybackslash}m{1.98cm}|>{\centering\arraybackslash}m{1.98cm}|>{\centering\arraybackslash}m{1.98cm}|>{\centering\arraybackslash}m{1.98cm}|>{\centering\arraybackslash}m{1.98cm}|}
\hline
\textbf{Dataset} & \textbf{RAP} & \textbf{5\%} & \textbf{10\%} & \textbf{15\%} & \textbf{20\%} & \textbf{25\%}\\ % 表头
\hline
\multirow{2}{*}{Citeseer} 
& \makecell[c]{w/o}
& 0.6831$\pm$0.0045	& 0.6560$\pm$0.0100 & 0.6117$\pm$0.0078 & 0.6175$\pm$0.0072 & 0.5870$\pm$0.0100\\ % 第一行
\cline{2-7} % 只画第二行到第八列的横线
& \makecell[c]{w/}
& 0.7159$\pm$0.0170 & 0.7005$\pm$0.0162 & 0.6885$\pm$0.0170 & 0.6754$\pm$0.0147 & 0.6647$\pm$0.0189\\ % 第一行
\hline
\end{tabular}
\end{table*}

\textbf{HMF vs. Degree Correlations.} We run experiments w.r.t. different-level degree correlations (DCs) via pruning on Citeseer, the results in Table \ref{tab:DegreeCorrelation} show DC has little influence.

\begin{table*}[bt]
\centering
\caption{Accuracy of node classification w.r.t. different degree correlations on Citeseer}
\label{tab:DegreeCorrelation}
\begin{tabular}{||>{\centering\arraybackslash}m{0.8cm}|>{\centering\arraybackslash}m{2.2cm}|>{\centering\arraybackslash}m{1.0cm}||>{\centering\arraybackslash}m{0.8cm}|>{\centering\arraybackslash}m{2.2cm}|>{\centering\arraybackslash}m{1.cm}||>{\centering\arraybackslash}m{0.8cm}|>{\centering\arraybackslash}m{2.2cm}|>{\centering\arraybackslash}m{1.cm}|}
\hline
\textbf{RAP} & \textbf{Accuracy} & \textbf{DCs} & \textbf{RAP} & \textbf{Accuracy} & \textbf{DCs} & \textbf{RAP} & \textbf{Accuracy} & \textbf{DCs}\\ % 表头
\hline
\multirow{3}{*}{\makecell[c]{5\%}} & 0.7125$\pm$0.0036 & 0.1965  & \multirow{3}{*}{\makecell[c]{15\%}} & 0.7029$\pm$0.0060 & 0.2242 & \multirow{3}{*}{\makecell[c]{25\%}} & 0.7165$\pm$0.0086 & 0.2270\\ % 第一行
\cline{2-3} \cline{5-6} \cline{8-9} 
& 0.7274$\pm$0.0070 & 0.0147 & & 0.7152$\pm$0.0074 & 0.0281& & 0.7226$\pm$0.0074 & 0.0312 \\
\cline{2-3} \cline{5-6} \cline{8-9}
& 0.7267$\pm$0.0062 & 0.0015 & & 0.7104$\pm$0.0049 & 0.0174 & & 0.7137$\pm$0.0057 & 0.0220 \\
\hline
\end{tabular}
\end{table*}

\subsection{Properties Variation}
This subsection studies whether our proposed method has a reverse affection on the adversarial phenomena in Fig. \ref{fig:AttackGraphProperty}. %To this end, we analyzed the purified graph referring to ASEP, in terms of the three attributes: feature smoothness, adjacency matrix's rank and singular values, and compared with the corresponding states under the adversarial perturbations.

%\subsubsection{Rank Optimization}
\begin{figure}[bt] % 用不带星号的 figure 环境，就不会跨栏，但是会上下浮动。可以去查一下关于浮动体的控制。
	\centering
	\subfigure[Rank growth under attacks]{
		\includegraphics[width=1.5in, height=1.2in]{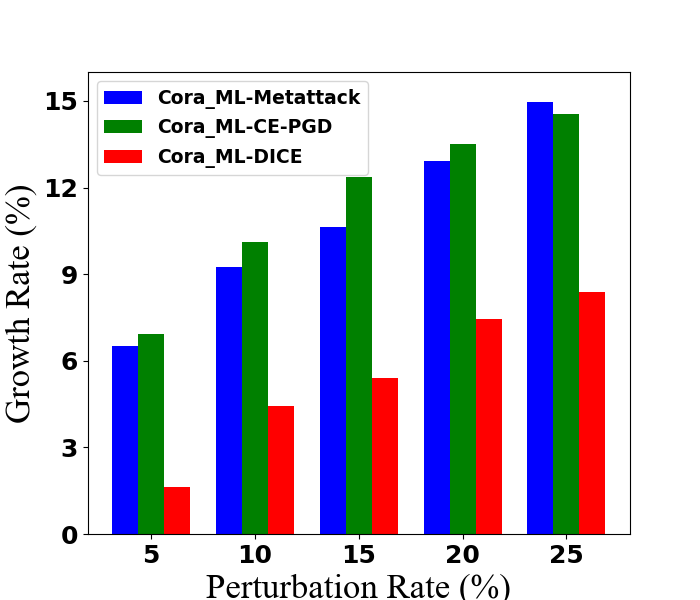}}
	\subfigure[Rank decreased by TopFeaRe]{
		\includegraphics[width=1.5in, height=1.2in]{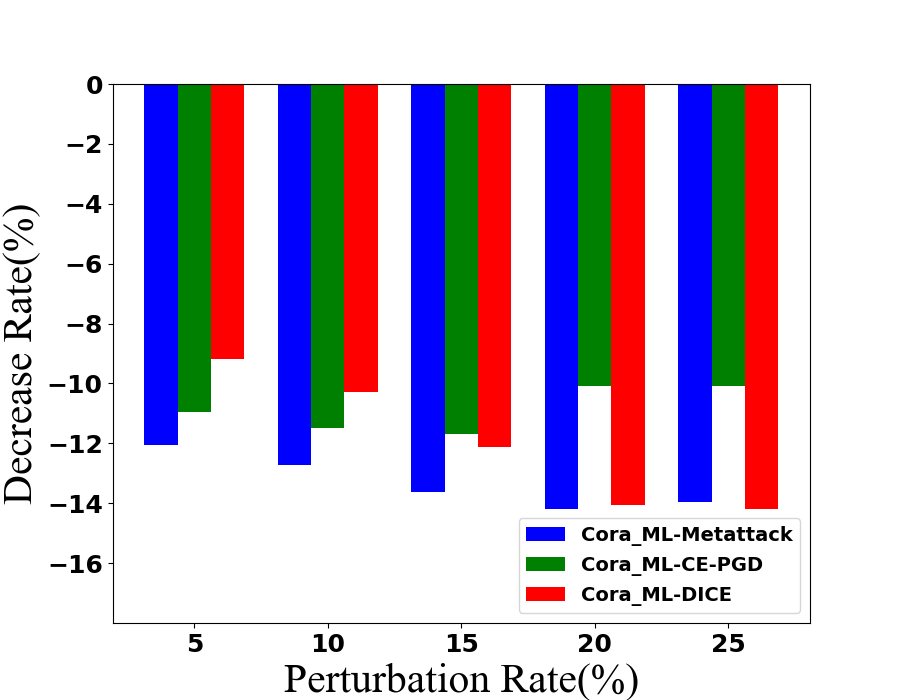}}
	\caption{Rank variation of adjacency matrix of Cora\_ML.}
	\label{fig:Rank Cora_ML}
\end{figure}	
\begin{figure}[bt]	
	%\begin{minipage}[b]{.45\linewidth}
	\subfigure[Rank growth under attacks]{
		\includegraphics[width=1.6in, height=1.2in]{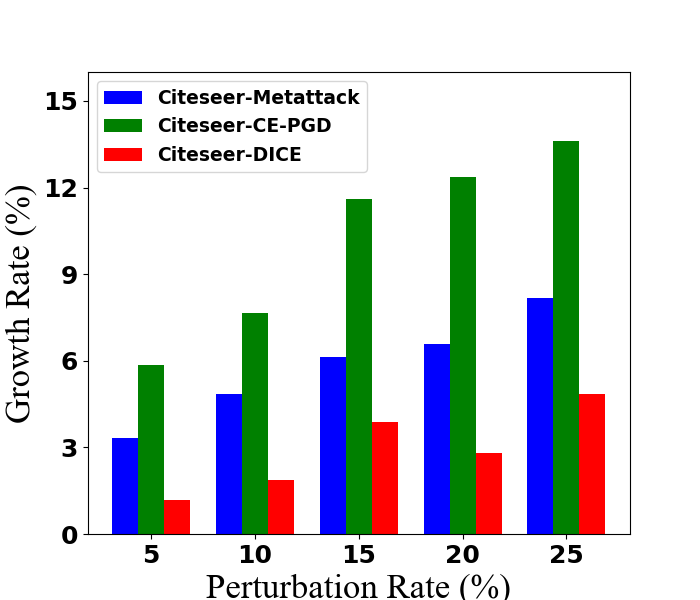}}
	\subfigure[Rank decreased by TopFeaRe]{
		\includegraphics[width=1.55in, height=1.2in]{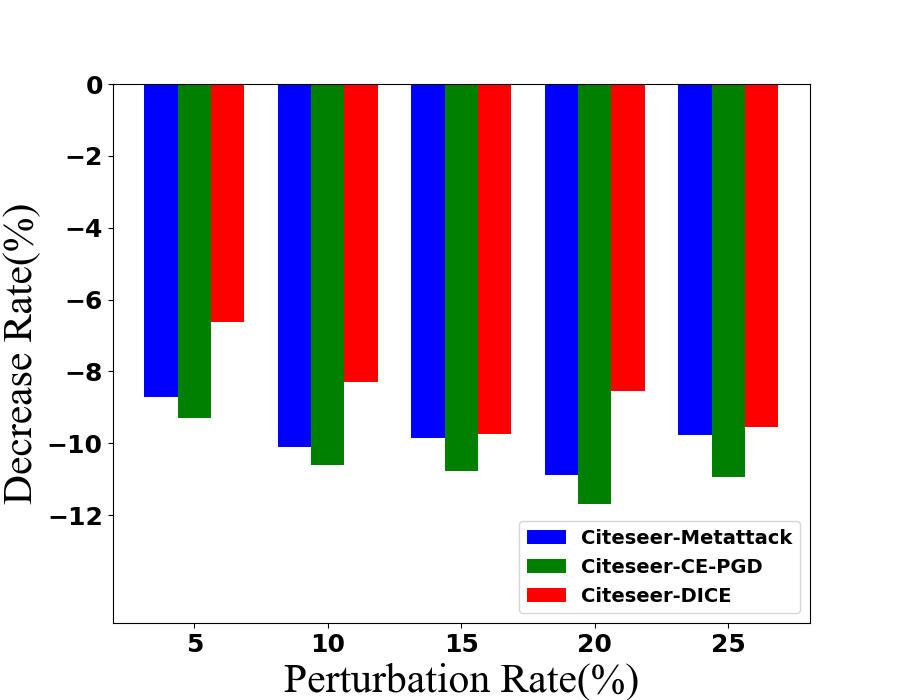}}
	\caption{Rank variation of adjacency matrix of citeseer.}
	\label{fig:Rank citeseer}
	%	\end{minipage}
\end{figure}
\textbf{Rank Variation.} Under the three GAAs, we analyzed the rank growth rates and examined the rank descent rates after purifying by TopFeaRe, as shown in Fig. \ref{fig:Rank Cora_ML} and Fig. \ref{fig:Rank citeseer}, the results indicate as RAP gradually increases from 5\% to 25\%, the rank of adjacency matrix correspondingly raises. This phenomenon suggests adversarial attacks tend to establish connections between nodes, leading to a more complex and dense graph structure. In response, our TopFeaRe significantly reduces the rank through effectively removing the edges over dissimilar nodes. This process not only mitigates the impact of adversarial perturbations, but it also optimizes the graph's topology, making it closer to a robust state.
%\begin{figure}[htbp]
%	\centering
%	\hfill
%	\subfigure[Rank growth under attacks]{
%		\includegraphics[width=1.6in, height=1.3in]{Figure/rank/cora_ml_grow.png}}
%	\hfill
%	\subfigure[Rank decreased by TopFeaRe]{
%		\includegraphics[width=1.6in, height=1.3in]{Figure/rank/cora_ml_decrease.png}}
%	\caption{Rank variation of adjacency matrix of Cora\_ML.}
%	\label{fig:Rank Cora_ML}
%\end{figure}
%% %秩图2
%% \begin{figure}[htbp]
%	% \centering
%	% \subfigure[Rank growth under attacks]{
%		% 		\includegraphics[scale=0.16]{Figure/rank/cora_grow.png}}
%	% \hfill
%	% \subfigure[Rank decrease with TopFeaRe]{
%		% 		\includegraphics[scale=0.135]{Figure/rank/cora_decrease.png}}
%	% \caption{Rank variation of adjacency matrix of Cora.}
%	% \label{fig:Rank Cora}
%	% \end{figure}
%%秩图3
%\begin{figure}[htbp]
%	\centering
%	\hfill
%	\subfigure[Rank growth under attacks]{
%		\includegraphics[width=1.6in, height=1.3in]{Figure/rank/citeseer_grow.png}}
%	\hfill
%	\subfigure[Rank decreased by TopFeaRe]{
%		\includegraphics[width=1.6in, height=1.3in]{Figure/rank/citeseer_decrease.png}}
%	\caption{Rank variation of adjacency matrix of citeseer.}
%	\label{fig:Rank citeseer}
%\end{figure}

%\subsubsection{Singular Value Variation}
\textbf{Singular-Value Variation.} To further analyze the changes of singular values, we select Cora\_ML and Citeseer datasets to observe the performance under the three attacks. Singular values are part of the spectral properties of a graph, reflecting its structural information. GAAs can alter the spectral properties of a graph, making an originally simple graph more complex to deceive the target models, which usually renders the increase in singular values. As shown in Fig. \ref{fig:sv cora_ml} and Fig. \ref{fig:sv citeseer}, all three attacks introduce unnecessary edges into the perturbed graphs. However, our TopFeaRe successfully reduces the singular values of the adjacency matrices for both datasets, effectively mitigating the impact of adversarial attacks.
%\subsubsection{Feature Smoothing Variation}
\begin{sloppypar}
\textbf{Feature-Smoothness Variation.} Feature smoothness is a crucial foundation for effective learning from different datasets. Seen from the previous experiments, we know adversarial attacks causing feature non-smoothness can mislead the target models during both training and inference. We present the changes between perturbed graphs and the clean at RAD 25\% on Cora\_ML and Citeseer datasets, as shown in Fig. \ref{fig:fs cora_ml} and Fig. \ref{fig:fs citeseer}. In clean graphs, feature values are typically concentrated in the regions with low feature distinctions, whereas after adversarial perturbations, this phenomenon is reversed. The purified graphs by TopFeaRe cause the edges in high feature-distinction regions to disperse, while the features become concentrated in the areas with low feature distinctions.
\end{sloppypar}

\subsection{Against Targeted Adversarial Attack}
%To study \textbf{RQ2}, 
We also use the dataset Citeseer to validate the effectiveness of our proposed method against targeted adversarial attack-Nettack. We compare TopFeaRe with other baseline methods, the experiments are shown in Table \ref{tab:Citeseer_Nettack}, from which, we can observe that our proposed TopFeaRe obviously outperforms other baselines across all perturbation counts from 0 to 5. Notably, our TopFeaRe surpasses the popular graph purification methods, such as GCN-SVD and GCN-Jaccard. To sum up, our proposed TopFeaRe demonstrates exceptional performance not only in resisting non-targeted attacks but also in effectively mitigating targeted attacks.

%奇异值1
\begin{figure}[tb]
\centering
\subfigure[Singular-Value under Metattack] % 为子图添加标题
{
\begin{minipage}[b]{.3\linewidth}
\centering
\includegraphics[width=1.25in, height=0.75in]{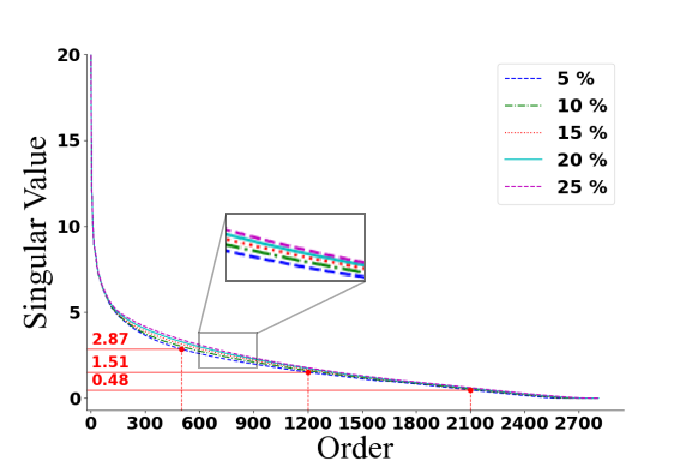}
\end{minipage}
}
\hfill % 在这里插入水平填充，使子图之间有间距
\subfigure[Singular-Value under CE-PGD] % 为子图添加标题
{
\begin{minipage}[b]{.3\linewidth}
\centering
\includegraphics[width=1.25in, height=0.75in]{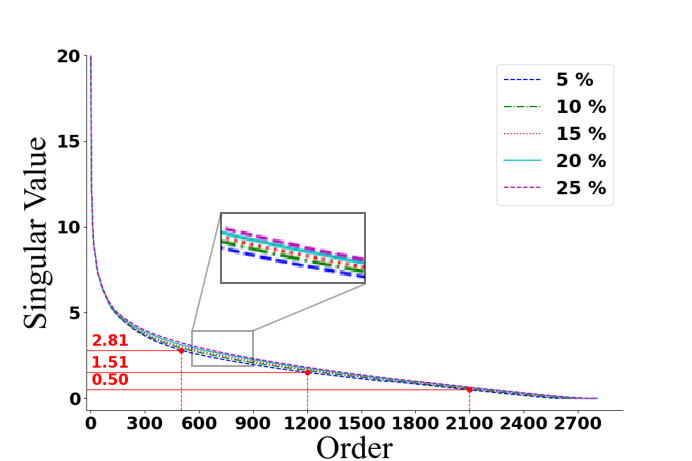}
\end{minipage}
}
\hfill % 在这里插入水平填充
\subfigure[Singular-Value under DICE] % 为子图添加标题
{
\begin{minipage}[b]{.3\linewidth}
\centering
\includegraphics[width=1.25in, height=0.75in]{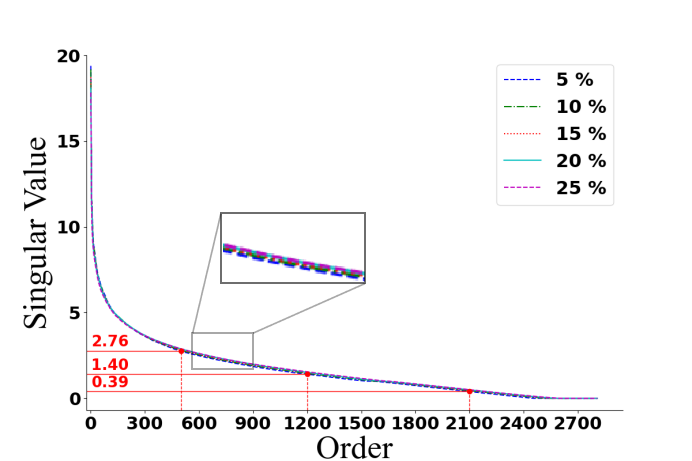}
\end{minipage}
}

\subfigure[Singular-Value decrease (w.r.t. Metattack)] % 为子图添加标题
{
\begin{minipage}[d]{.3\linewidth}
\centering
\includegraphics[width=1.25in, height=0.75in]{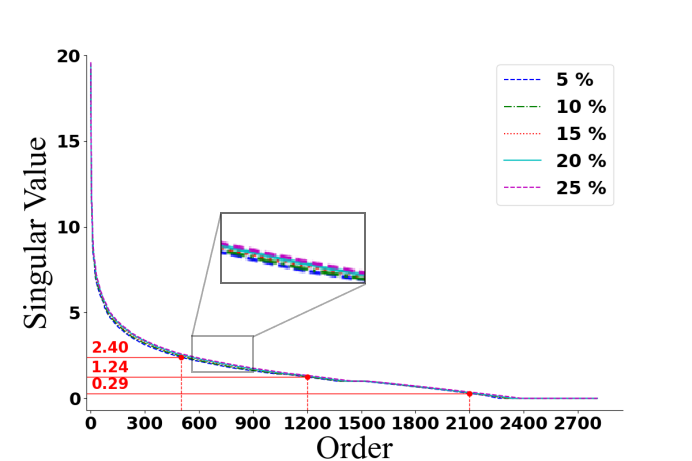}
\end{minipage}
}
\hfill % 在这里插入水平填充，使子图之间有间距
\subfigure[Singular-Value decrease (w.r.t. CE-PGD)] % 为子图添加标题
{
\begin{minipage}[e]{.3\linewidth}
\centering
\includegraphics[width=1.25in, height=0.75in]{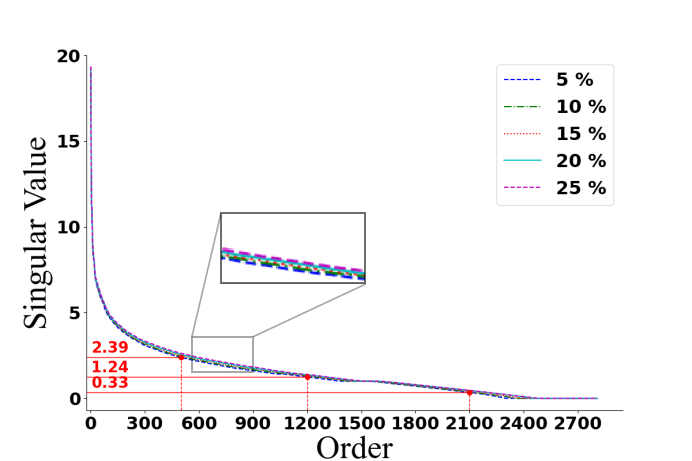}
\end{minipage}
}
\hfill % 在这里插入水平填充
\subfigure[Singular-Value decrease (w.r.t. DICE)] % 为子图添加标题
{
\begin{minipage}[f]{.3\linewidth}
\centering
\includegraphics[width=1.25in, height=0.75in]{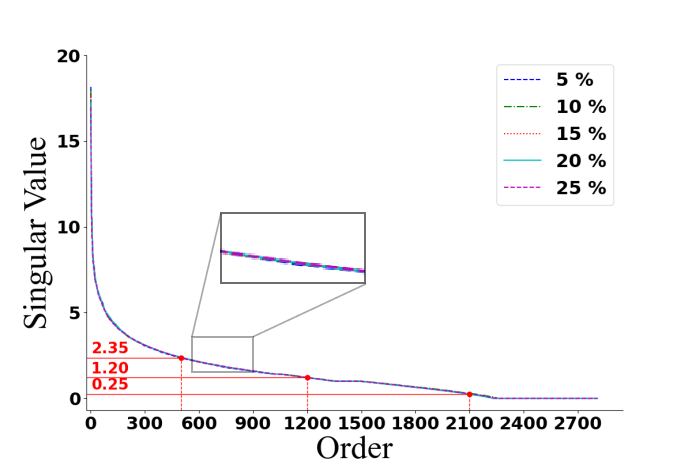}
\end{minipage}
}
\caption{Singular-Value variation of Cora\_ML} 
\label{fig:sv cora_ml}
\end{figure}
\begin{figure}[tb]
	\centering
	\subfigure[Singular-Value under Metattack] % 为子图添加标题
	{
		\begin{minipage}[b]{.3\linewidth}
			\centering
			\includegraphics[width=1.25in, height=0.75in]{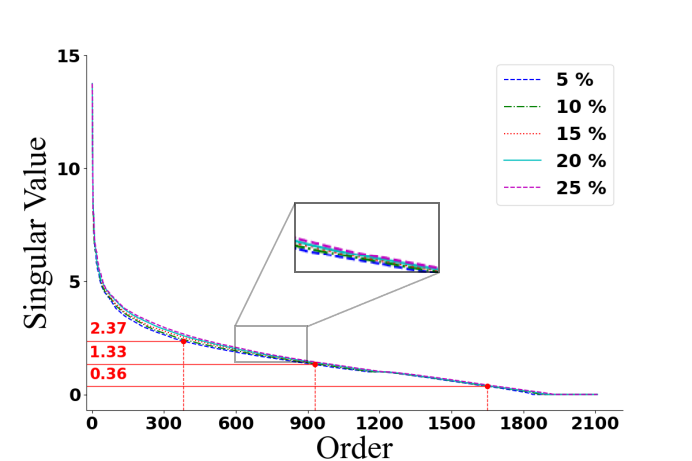}
		\end{minipage}
	}
	\hfill % 在这里插入水平填充，使子图之间有间距
	\subfigure[Singular-Value under CE-PGD] % 为子图添加标题
	{
		\begin{minipage}[b]{.3\linewidth}
			\centering
			\includegraphics[width=1.25in, height=0.75in]{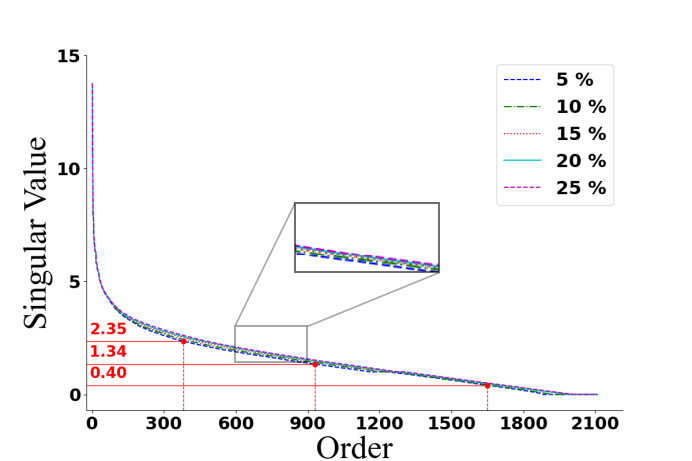}
		\end{minipage}
	}
	\hfill % 在这里插入水平填充
	\subfigure[Singular-Value under DICE] % 为子图添加标题
	{
		\begin{minipage}[b]{.3\linewidth}
			\centering
			\includegraphics[width=1.25in, height=0.75in]{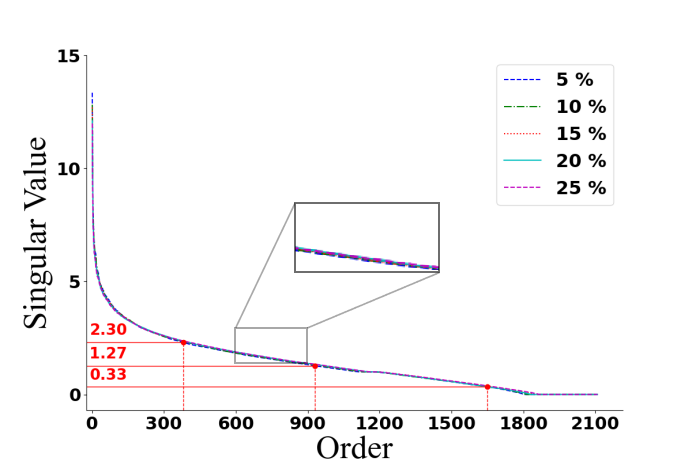}
		\end{minipage}
	}
	
	\subfigure[Singular-Value decrease (w.r.t. Metattack)] % 为子图添加标题
	{
		\begin{minipage}[d]{.3\linewidth}
			\centering
			\includegraphics[width=1.25in, height=0.75in]{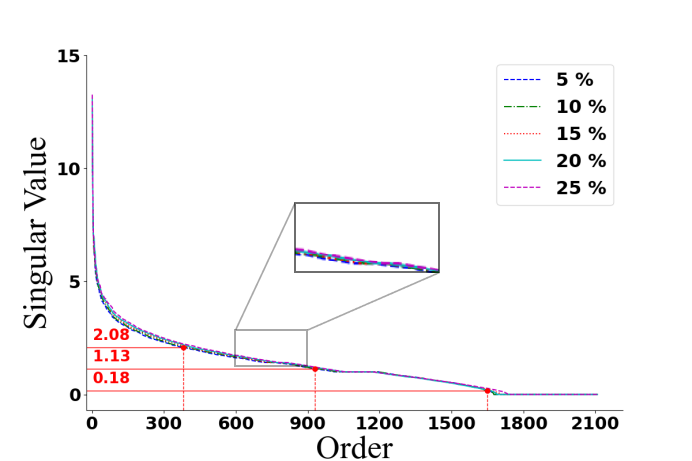}
		\end{minipage}
	}
	\hfill % 在这里插入水平填充，使子图之间有间距
	\subfigure[Singular-Value decrease (w.r.t. CE-PGD)] % 为子图添加标题
	{
		\begin{minipage}[e]{.3\linewidth}
			\centering
			\includegraphics[width=1.25in, height=0.75in]{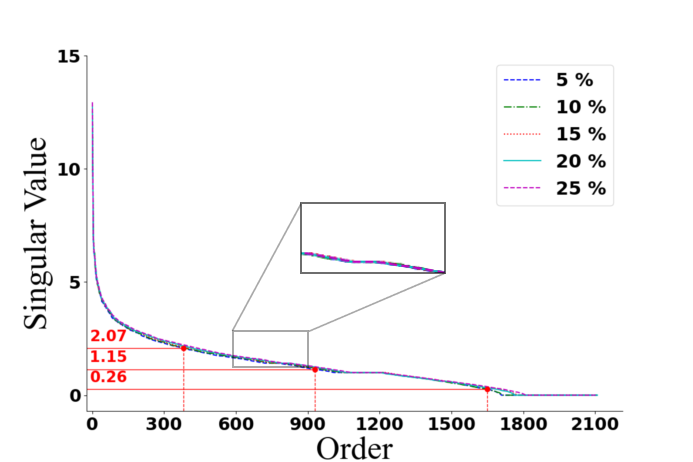}
		\end{minipage}
	}
	\hfill % 在这里插入水平填充
	\subfigure[Singular-Value decrease (w.r.t. DICE)] % 为子图添加标题
	{
		\begin{minipage}[f]{.3\linewidth}
			\centering
			\includegraphics[width=1.25in, height=0.75in]{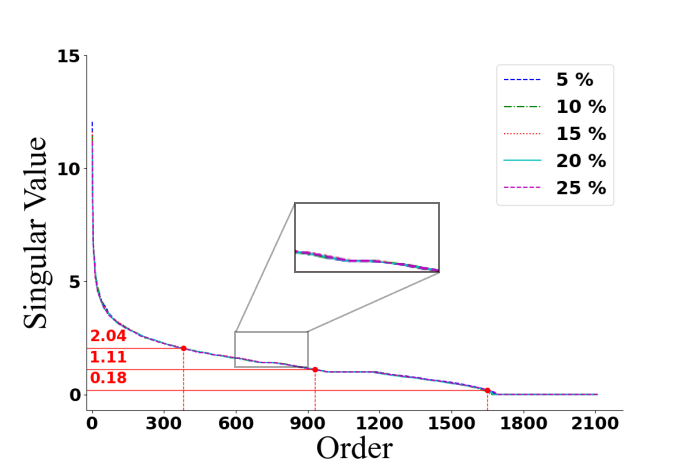}
		\end{minipage}
	}
	
	\caption{Singular-Value variation of Citeseer}
	\label{fig:sv citeseer}
\end{figure}
%特征平滑1
\begin{figure}[tb]
	\centering
	\subfigure[Feature-Smoothness under Metattack] % 为子图添加标题
	{
		\begin{minipage}[b]{.3\linewidth}
			\centering
			\includegraphics[width=1.25in, height=0.7in]{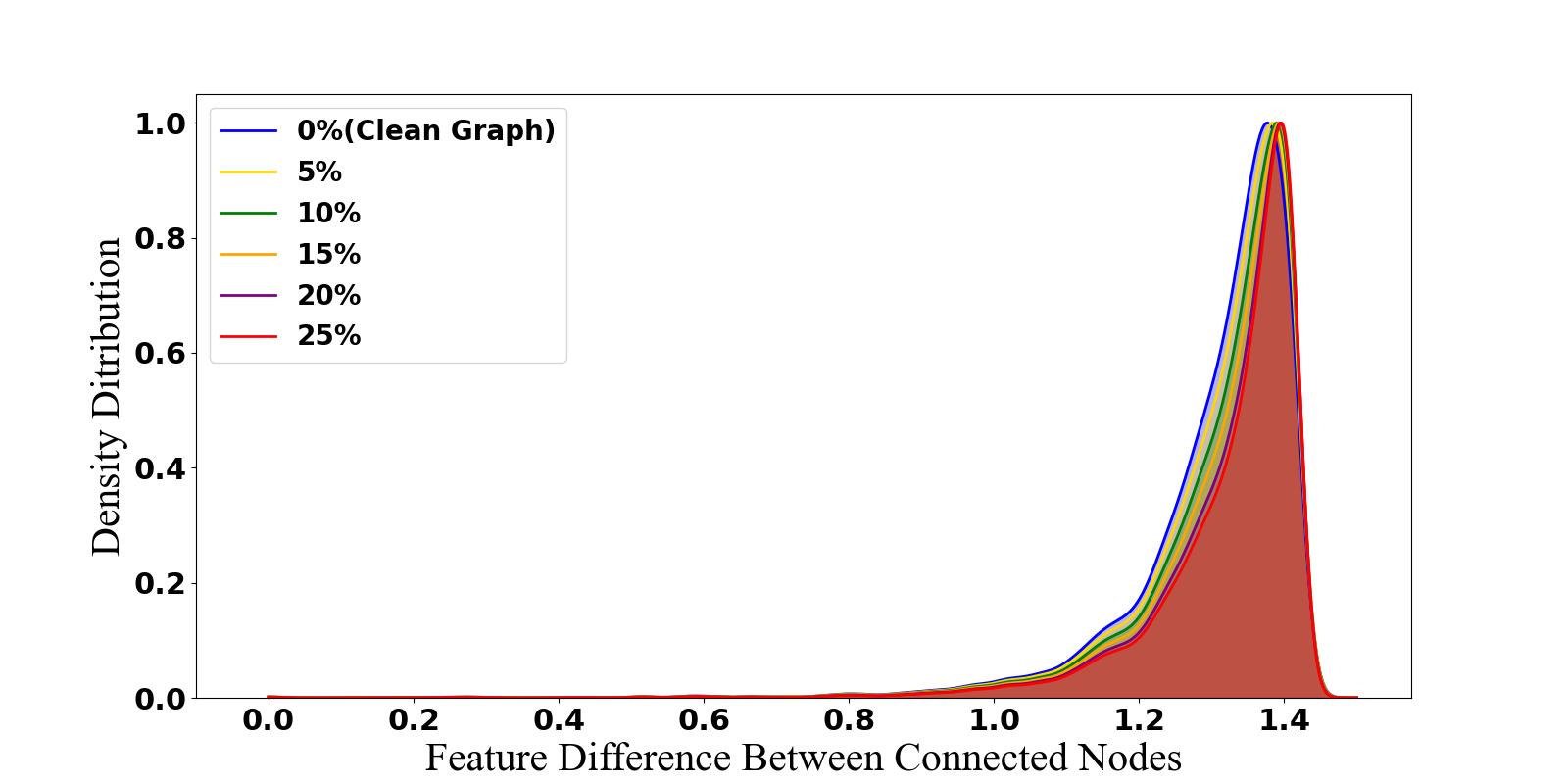}
		\end{minipage}
	}
	\hfill
	\subfigure[Feature-Smoothness under CE-PGD] % 为子图添加标题
	{
		\begin{minipage}[b]{.3\linewidth}
			\centering
			\includegraphics[width=1.25in, height=0.7in]{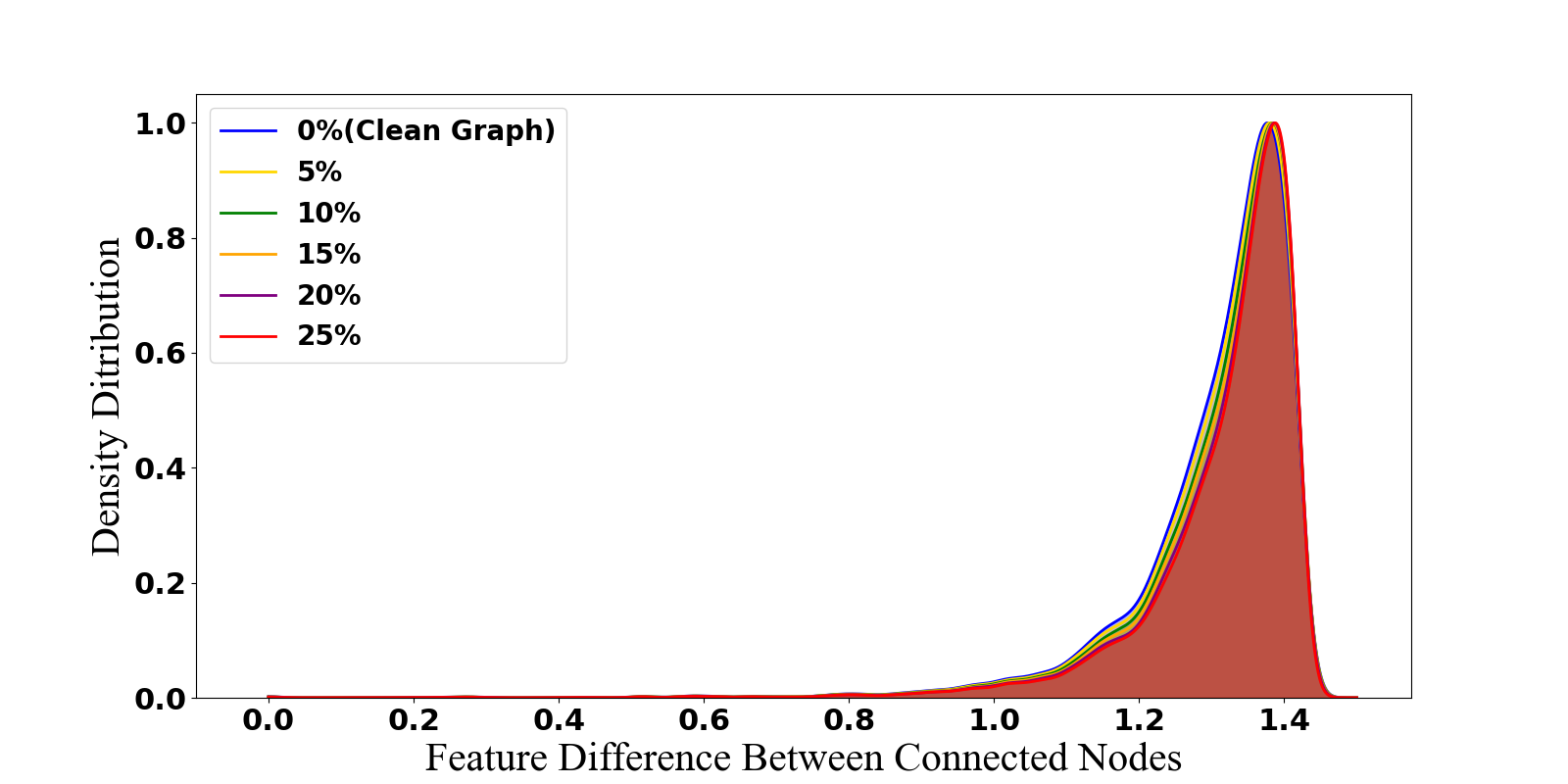}
		\end{minipage}
	}
	\hfill
	\subfigure[Feature-Smoothness under DICE] % 为子图添加标题
	{
		\begin{minipage}[b]{.3\linewidth}
			\centering
			\includegraphics[width=1.25in, height=0.7in]{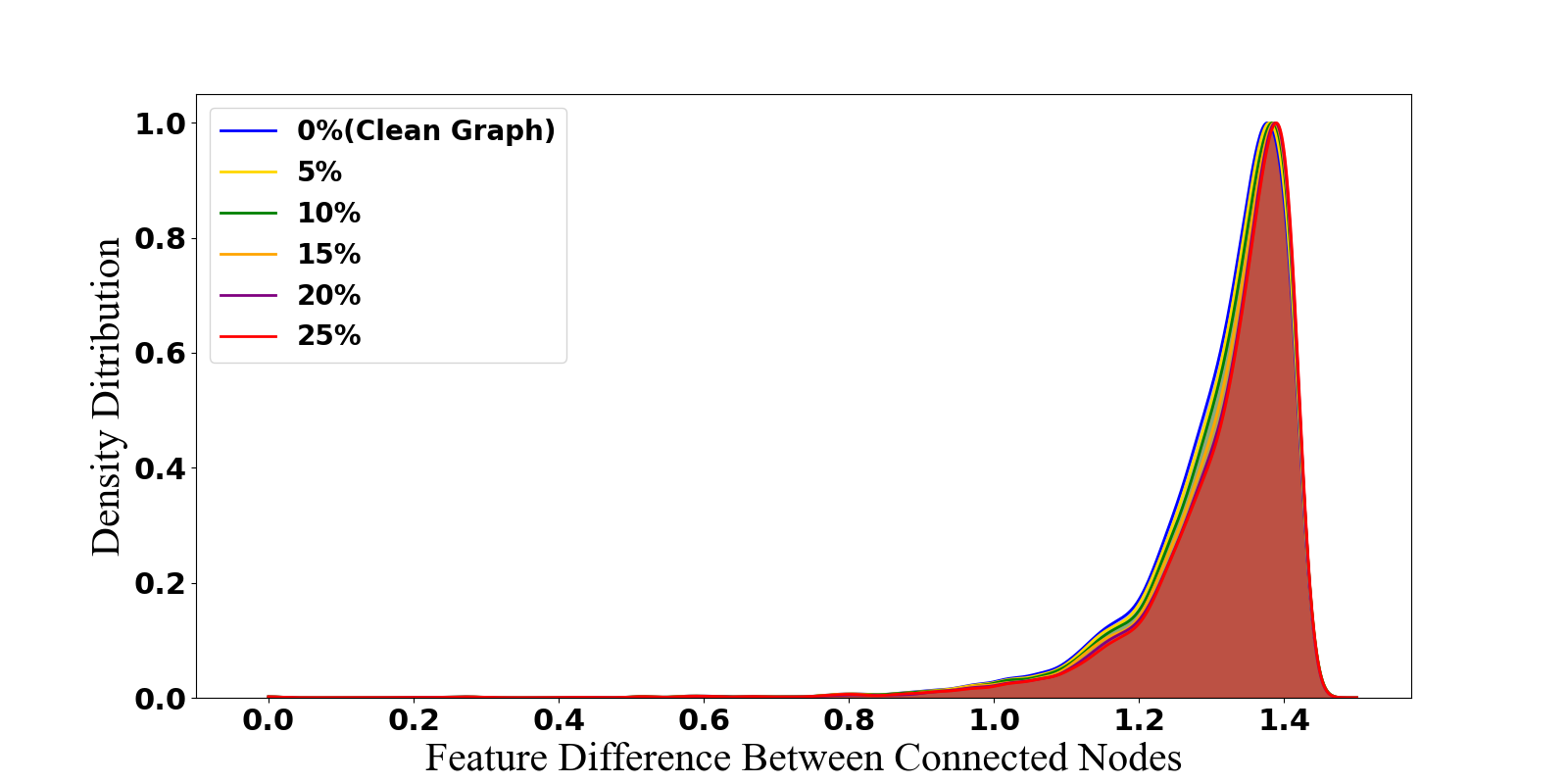}
		\end{minipage}
	}
	
	\subfigure[Feature-Smoothness: Metattack vs. TopFeaRe] % 为子图添加标题
	{
		\begin{minipage}[b]{.3\linewidth}
			\centering
			\includegraphics[width=1.25in, height=0.7in]{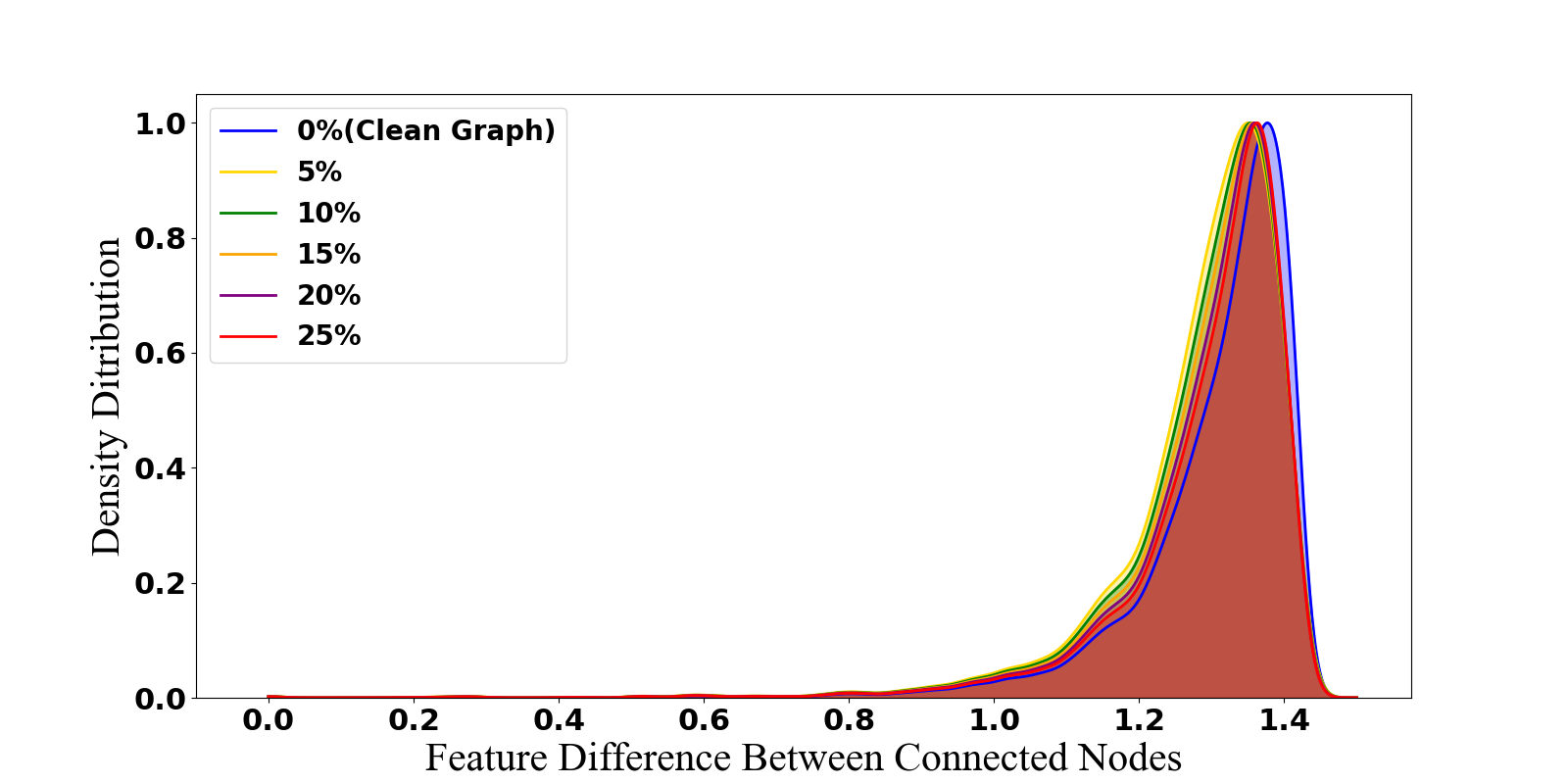}
		\end{minipage}
	}
	\hfill
	\subfigure[Feature-Smoothness: CE-PGD vs. TopFeaRe] % 为子图添加标题
	{
		\begin{minipage}[b]{.3\linewidth}
			\centering
			\includegraphics[width=1.25in, height=0.7in]{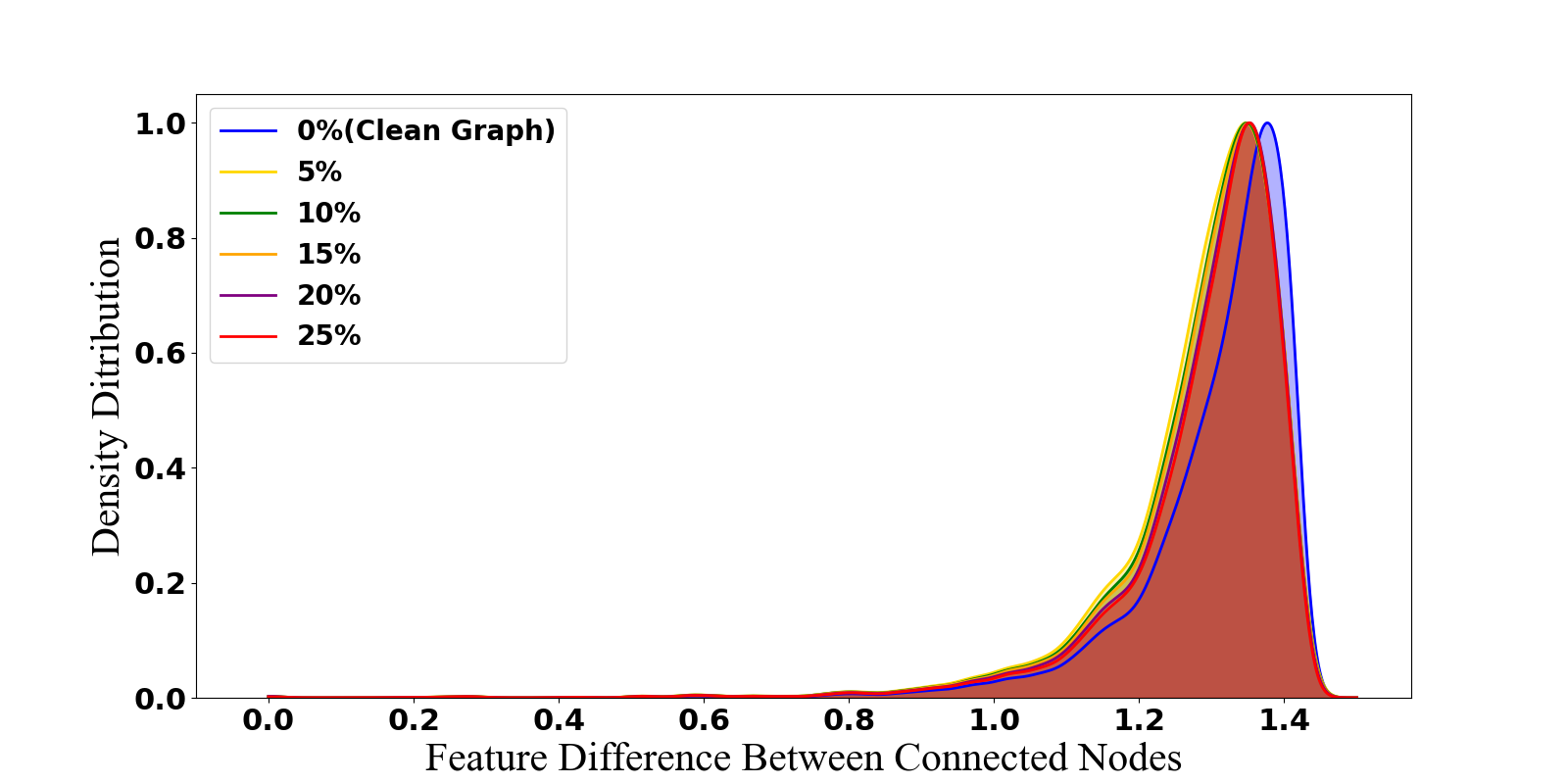}
		\end{minipage}
	}
	\hfill
	\subfigure[Feature-Smoothness: DICE vs. TopFeaRe] % 为子图添加标题
	{
		\begin{minipage}[b]{.3\linewidth}
			\centering
			\includegraphics[width=1.25in, height=0.7in]{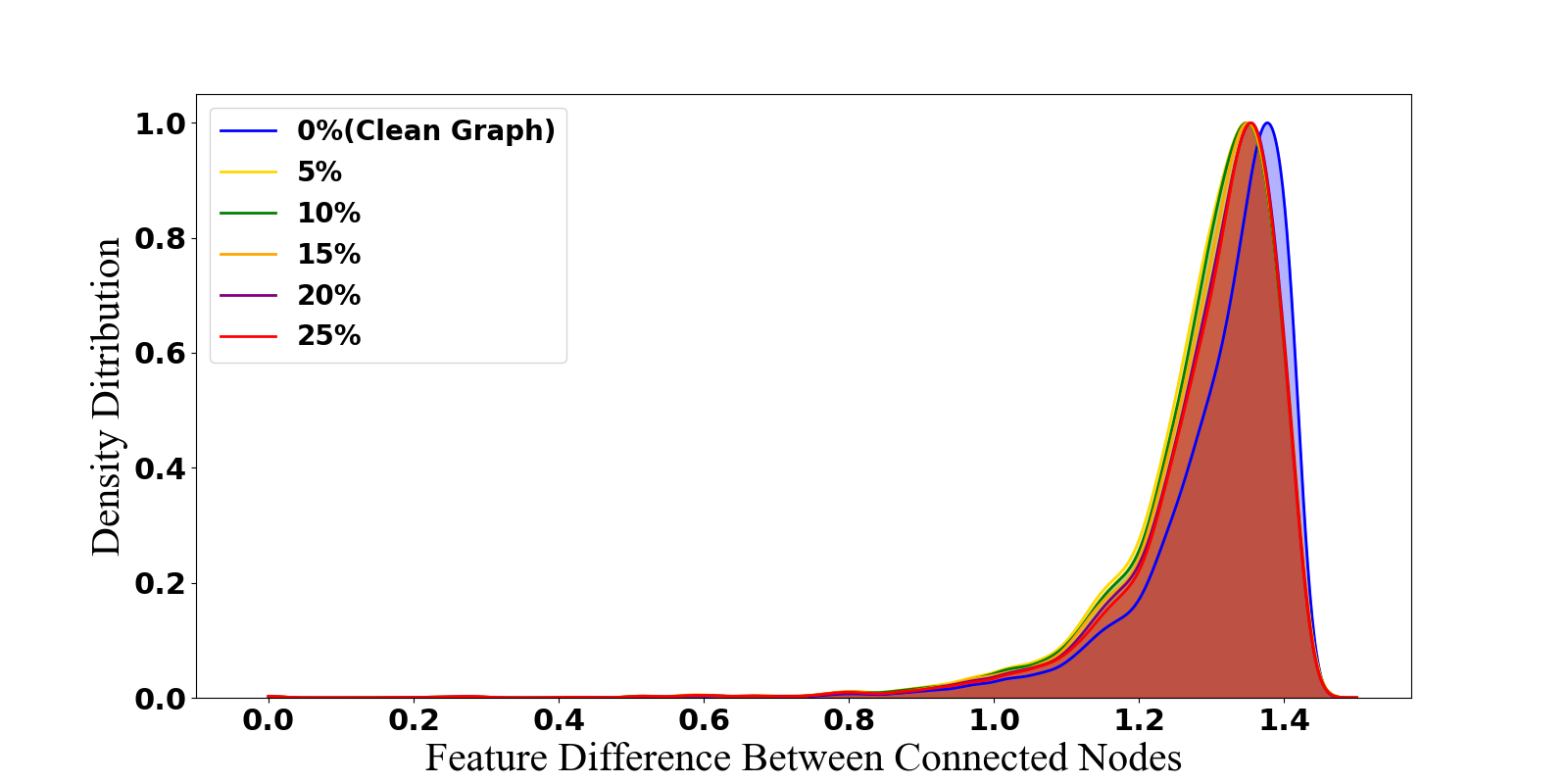}
		\end{minipage}
	}
	
	\caption{Feature-Smoothness variation of Cora\_ML} 
	\label{fig:fs cora_ml}
\end{figure}
\begin{figure}[tb]
	\centering
	\subfigure[Feature-Smoothness under Metattack] % 为子图添加标题
	{
		\begin{minipage}[b]{.3\linewidth}
			\centering
			\includegraphics[width=1.25in, height=0.7in]{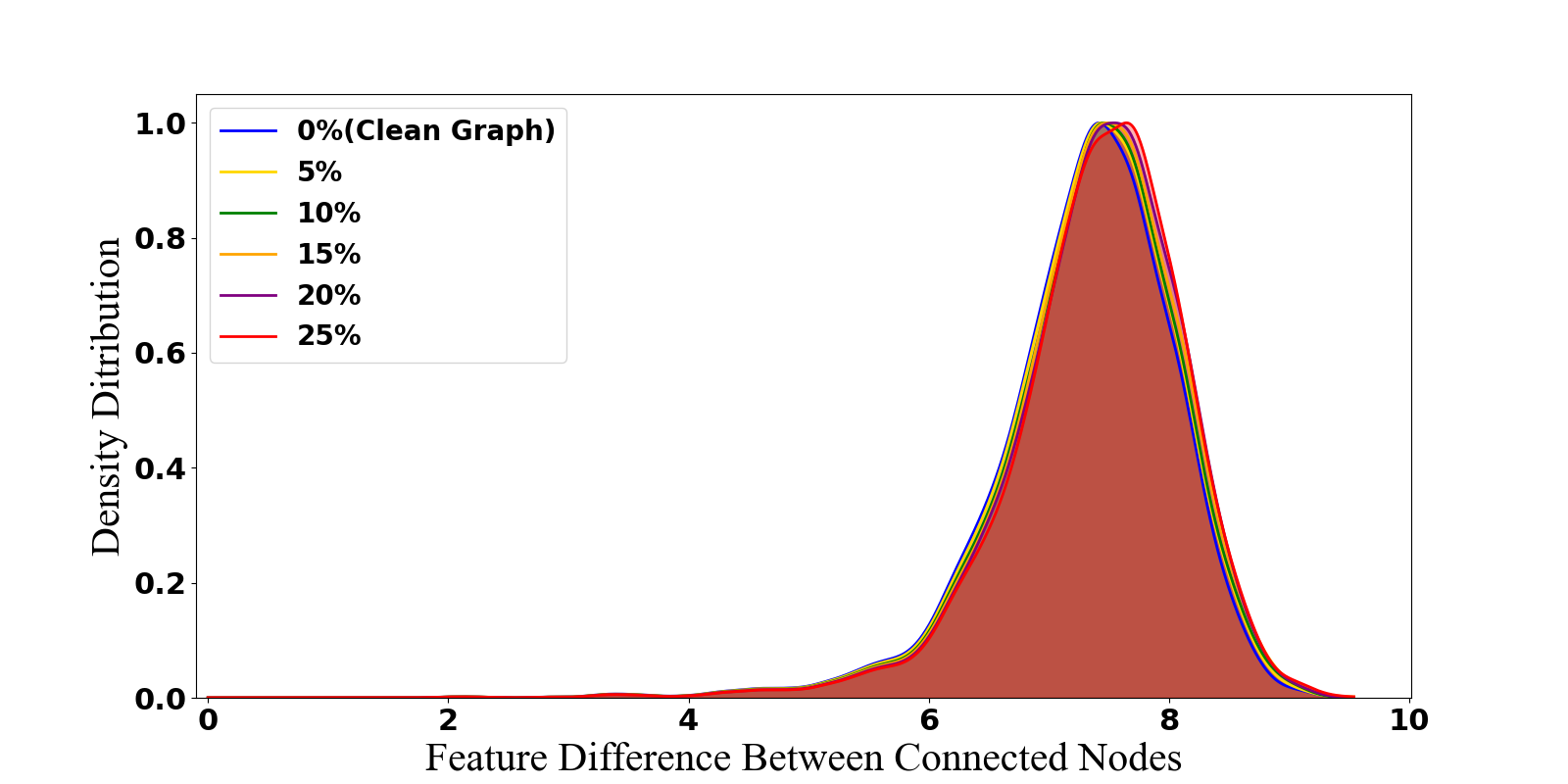}
		\end{minipage}
	}
	\hfill
	\subfigure[Feature-Smoothness under CE-PGD] % 为子图添加标题
	{
		\begin{minipage}[b]{.3\linewidth}
			\centering
			\includegraphics[width=1.25in, height=0.7in]{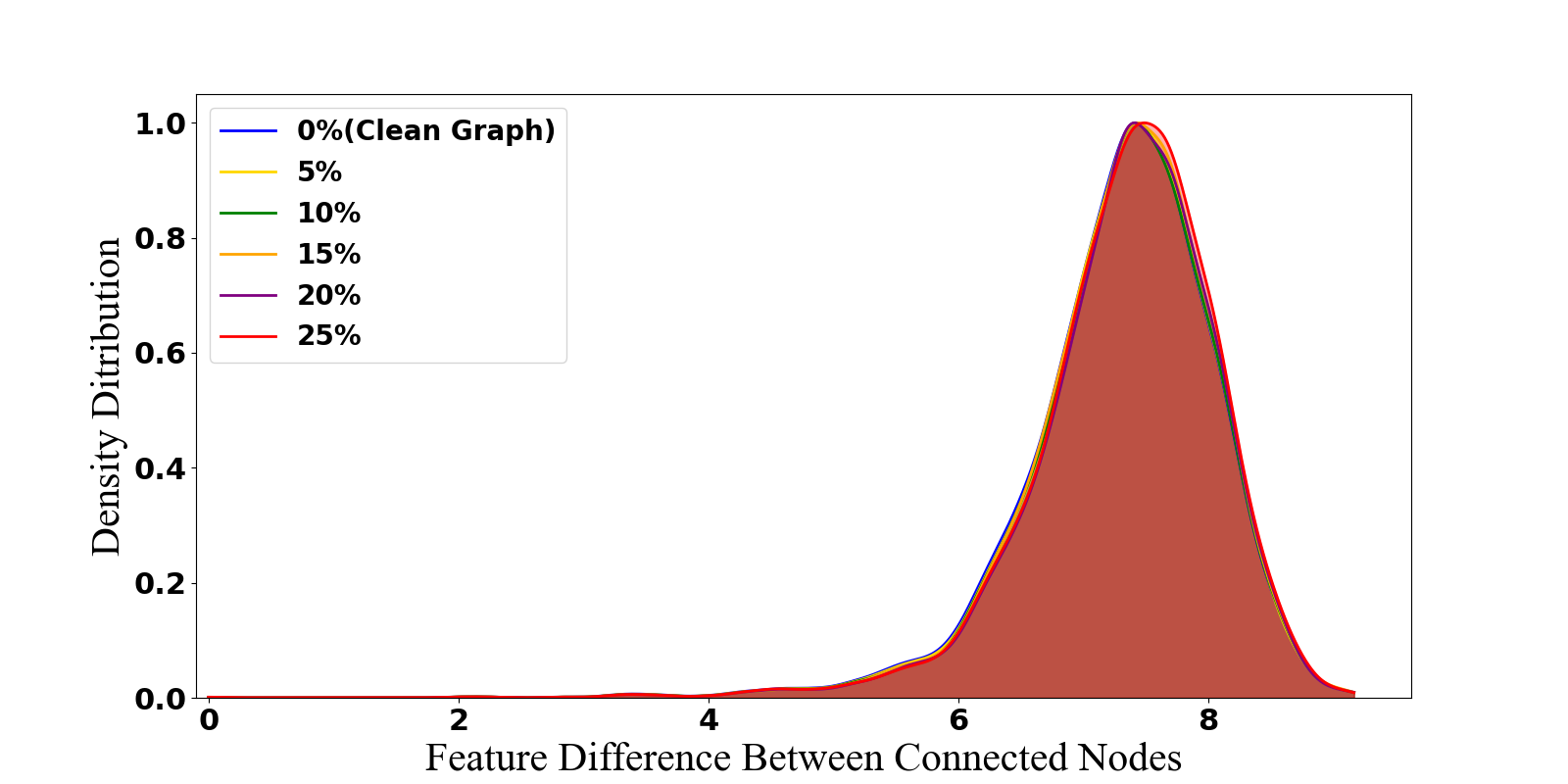}
		\end{minipage}
	}
	\hfill
	\subfigure[Feature-Smoothness under DICE] % 为子图添加标题
	{
		\begin{minipage}[b]{.3\linewidth}
			\centering
			\includegraphics[width=1.25in, height=0.7in]{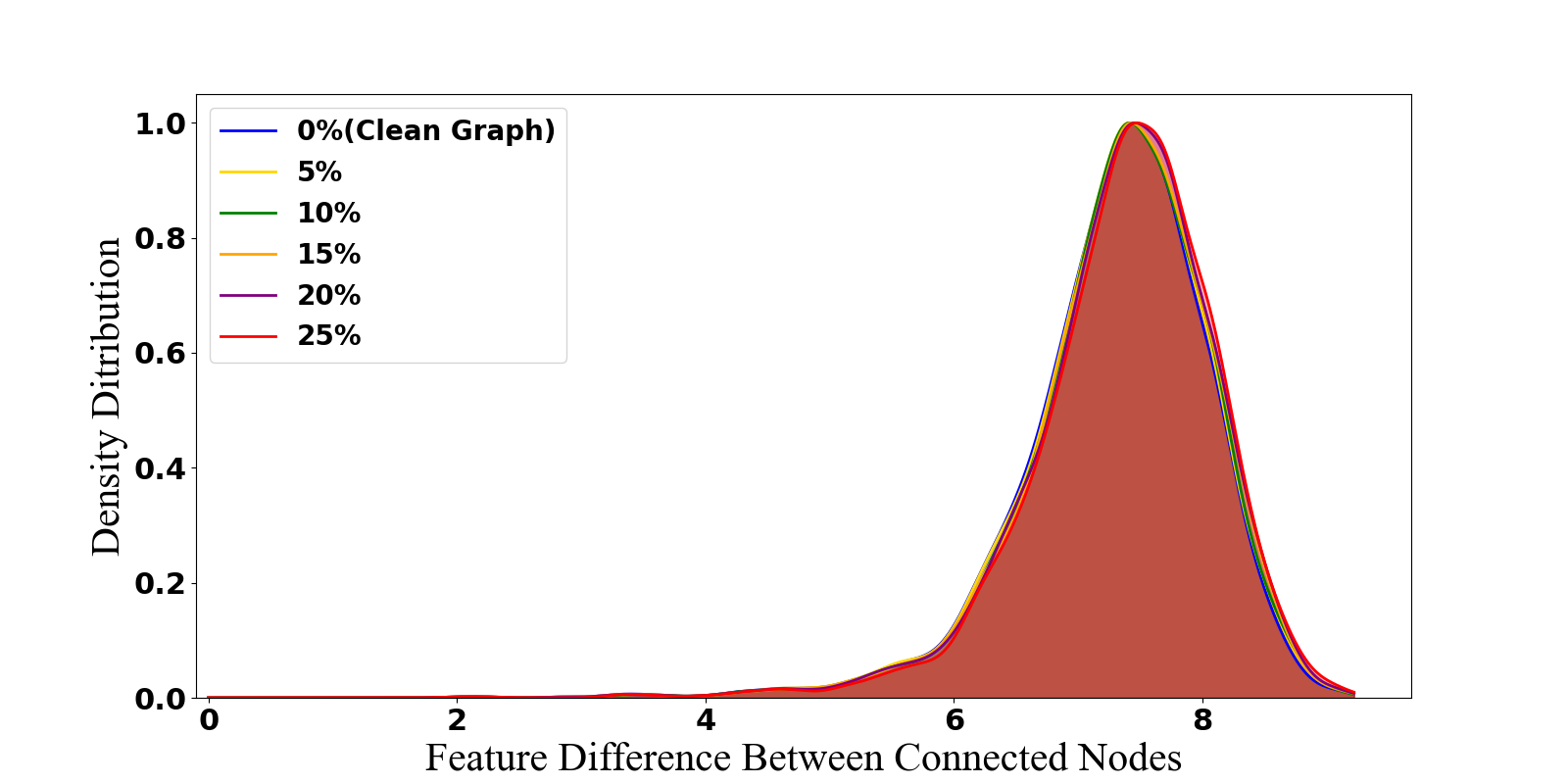}
		\end{minipage}
	}
	
	\subfigure[Feature-Smoothness: Metattack vs. TopFeaRe] % 为子图添加标题
	{
		\begin{minipage}[b]{.3\linewidth}
			\centering
			\includegraphics[width=1.25in, height=0.7in]{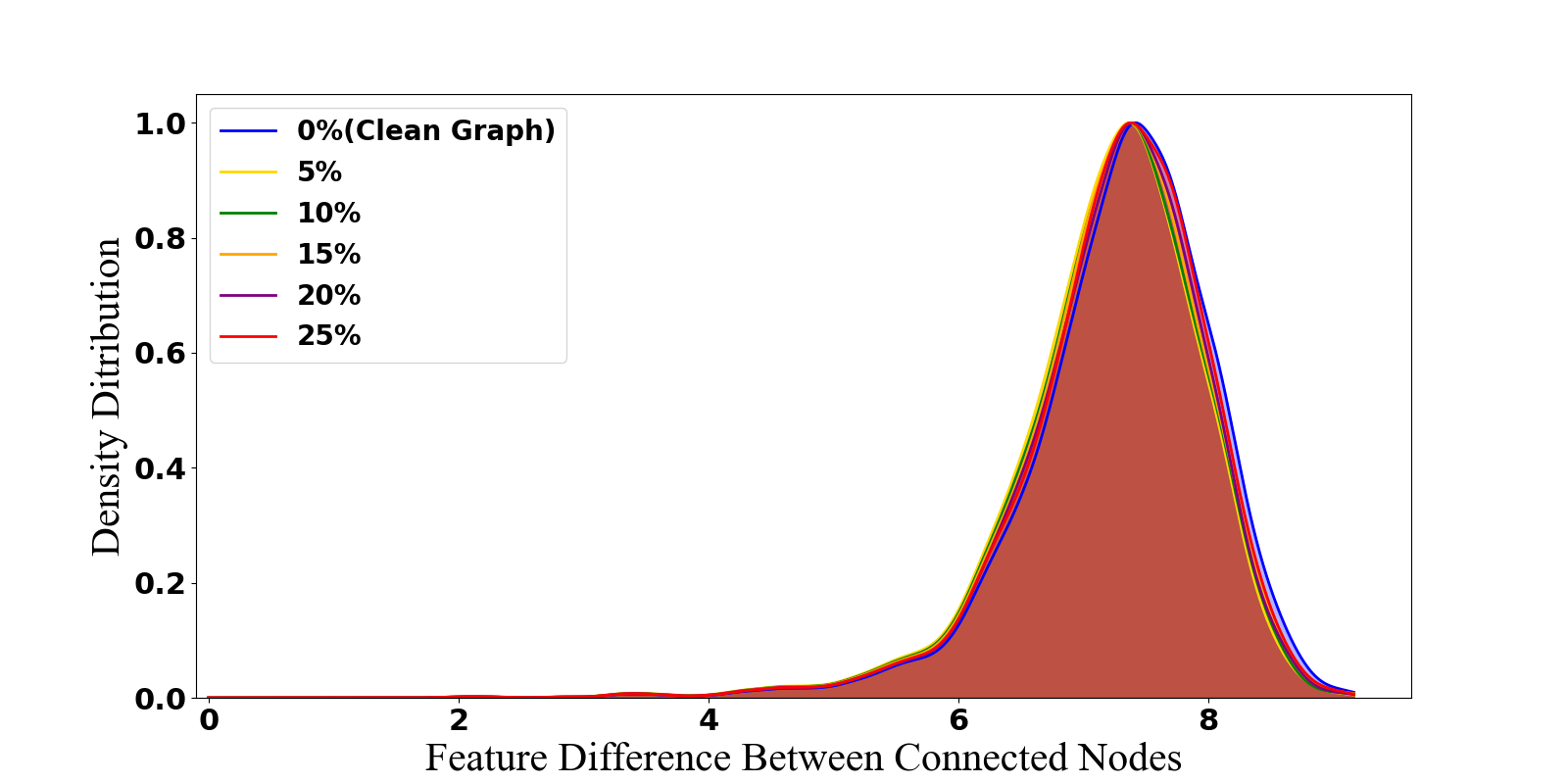}
		\end{minipage}
	}
	\hfill
	\subfigure[Feature-Smoothness: CE-PGD vs. TopFeaRe] % 为子图添加标题
	{
		\begin{minipage}[b]{.3\linewidth}
			\centering
			\includegraphics[width=1.25in, height=0.7in]{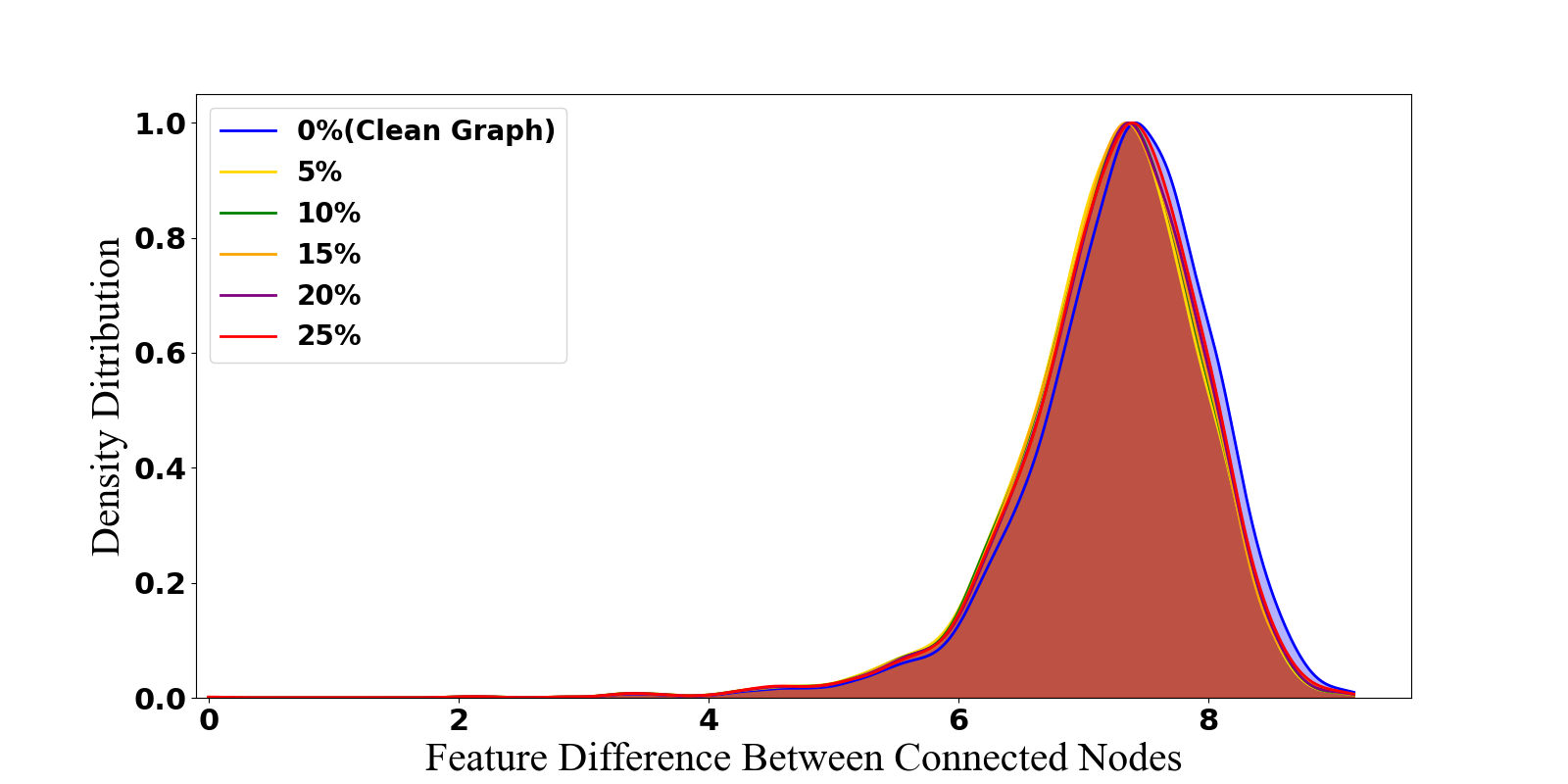}
		\end{minipage}
	}
	\hfill
	\subfigure[Feature-Smoothness: DICE vs. TopFeaRe] % 为子图添加标题
	{
		\begin{minipage}[b]{.3\linewidth}
			\centering
			\includegraphics[width=1.25in, height=0.7in]{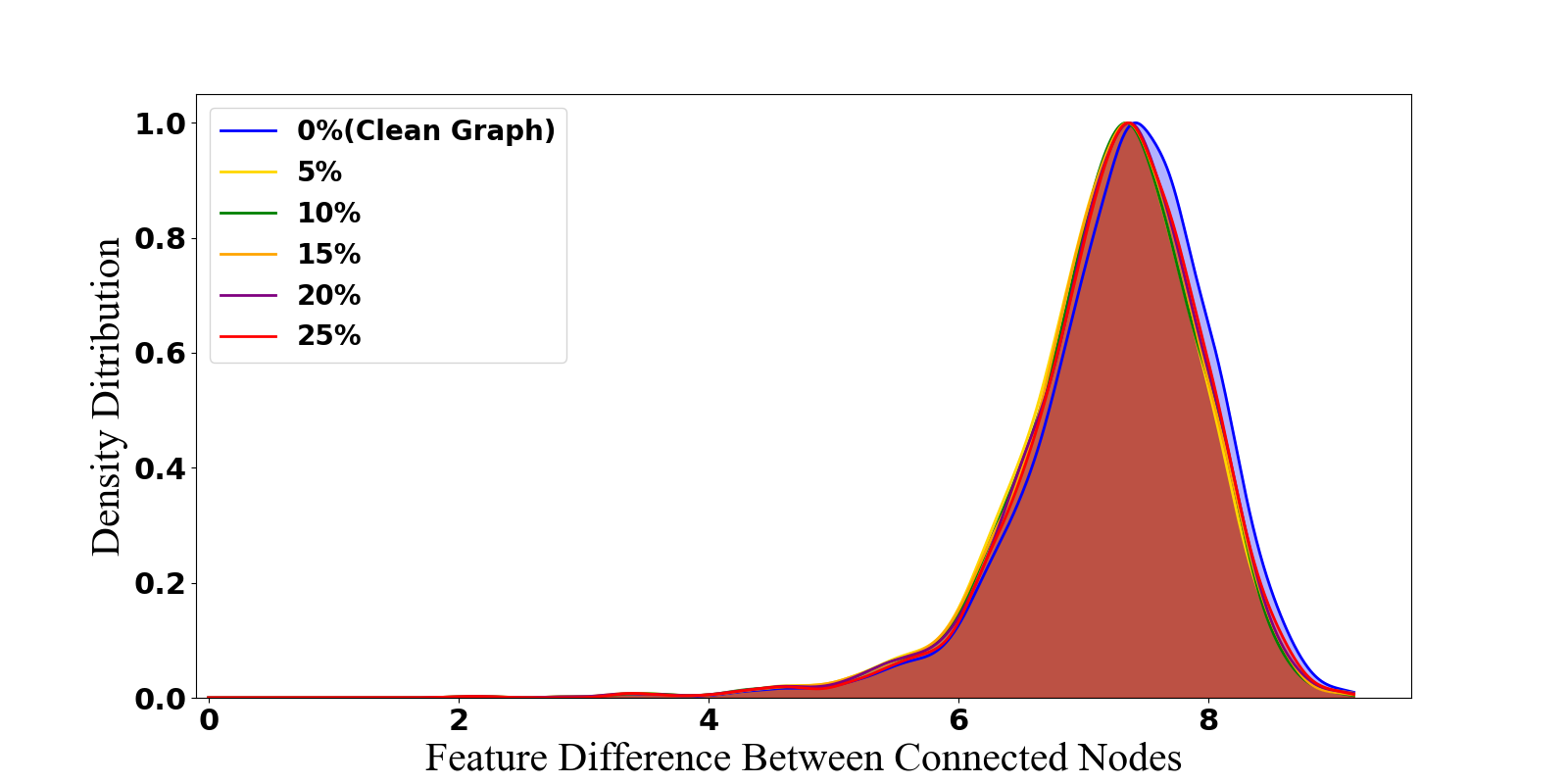}
		\end{minipage}
	}
	
	\caption{Feature-Smoothness variation on Citeseer}
	\label{fig:fs citeseer}
\end{figure}

\begin{table*}[htbp]
\centering
%\caption{ACCURACY OF NODE CLASSIFICATION USING NETTACK (\textbf{BOLD}-THE BEST)}
\caption{Accuracy of node classification using Nettack}
\label{tab:Citeseer_Nettack}
\begin{tabular}{|>{\centering\arraybackslash}m{0.6cm}|>{\centering\arraybackslash}m{1.95cm}|>{\centering\arraybackslash}m{1.98cm}|>{\centering\arraybackslash}m{1.98cm}|>{\centering\arraybackslash}m{1.98cm}|>{\centering\arraybackslash}m{1.98cm}|>{\centering\arraybackslash}m{1.98cm}|>{\centering\arraybackslash}m{1.98cm}|}
\hline
\textbf{\makecell[c]{Data\\set}} & \diagbox{{\textbf{Mod.}}}{{\textbf{Per.}}} & \textbf{0} & \textbf{1} & \textbf{2} & \textbf{3} & \textbf{4} & \textbf{5}\\
\hline

\multirow{6}{*}{\makecell[c]{Cite\\seer}} & GCN & 0.8344$\pm$0.0344 & 0.8309$\pm$0.0283 & 0.8236$\pm$0.0210 & 0.8330$\pm$0.0202 & 0.8375$\pm$0.0201 & 0.8300$\pm$0.0162\\ 
\cline{2-8} 
& GAT & 0.7303$\pm$0.0095 & 0.7298$\pm$0.0111 & 0.7322$\pm$0.0115 & 0.7323$\pm$0.0150 & 0.7358$\pm$0.0148 & 0.7277$\pm$0.0186\\ 
\cline{2-8}
& HANG & 0.8390$\pm$0.0188 & 0.8371$\pm$0.0223 & 0.8372$\pm$0.0219 & 0.8356$\pm$0.0235 & 0.8331$\pm$0.0229 & 0.8327$\pm$0.0342\\ 
\cline{2-8}
& GCN-SVD & 0.8297$\pm$0.0318 & 0.8174$\pm$0.0233 & 0.8141$\pm$0.0299 & 0.8186$\pm$0.0277 & 0.8216$\pm$0.0217 & 0.8138$\pm$0.0353\\ 
\cline{2-8} 
& GCN-Jaccard & 0.8374$\pm$0.0266 & 0.8356$\pm$0.0180 & 0.8282$\pm$0.0189 & 0.8311$\pm$0.0206 & 0.8269$\pm$0.0140 & 0.8362$\pm$0.0240\\ 
\cline{2-8}
& TopFeaRe &  {\textbf{0.8617$\pm$0.0247}} \(\textcolor{red}{(\uparrow \textbf{2.27\%})}\) & {\textbf{0.8544$\pm$0.0271}} \(\textcolor{red}{(\uparrow \textbf{1.73\%})}\) & {\textbf{0.8573$\pm$0.0247}} \(\textcolor{red}{(\uparrow \textbf{2.01\%})}\) & {\textbf{0.8588$\pm$0.0264}} \(\textcolor{red}{(\uparrow \textbf{2.32\%})}\) & {\textbf{0.8466$\pm$0.0286}} \(\textcolor{red}{(\uparrow \textbf{0.91\%})}\) & {\textbf{0.8491$\pm$0.0278}} \(\textcolor{red}{(\uparrow \textbf{1.29\%})}\) \\
\hline  
\end{tabular}
\end{table*}

\subsection{Computational Complexity}
Assume $E(n)$ is the edge function with respect to the node number $n$, the computational complexity depends on the number of edges, i.e. $O(E(n))$. Furthermore, consider that graph is perturbed pursuant to $RAP$, thus the total complexity can be formulated as $O(E(n + n * RAP))$. The time overhead is sketched in Fig. \ref{fig:TimeOverhead} as graph size enlarges at RAPs 5\%, 15\%.

\begin{figure}[bt]
\centering
\includegraphics[width=1.8in, height=1.2in]{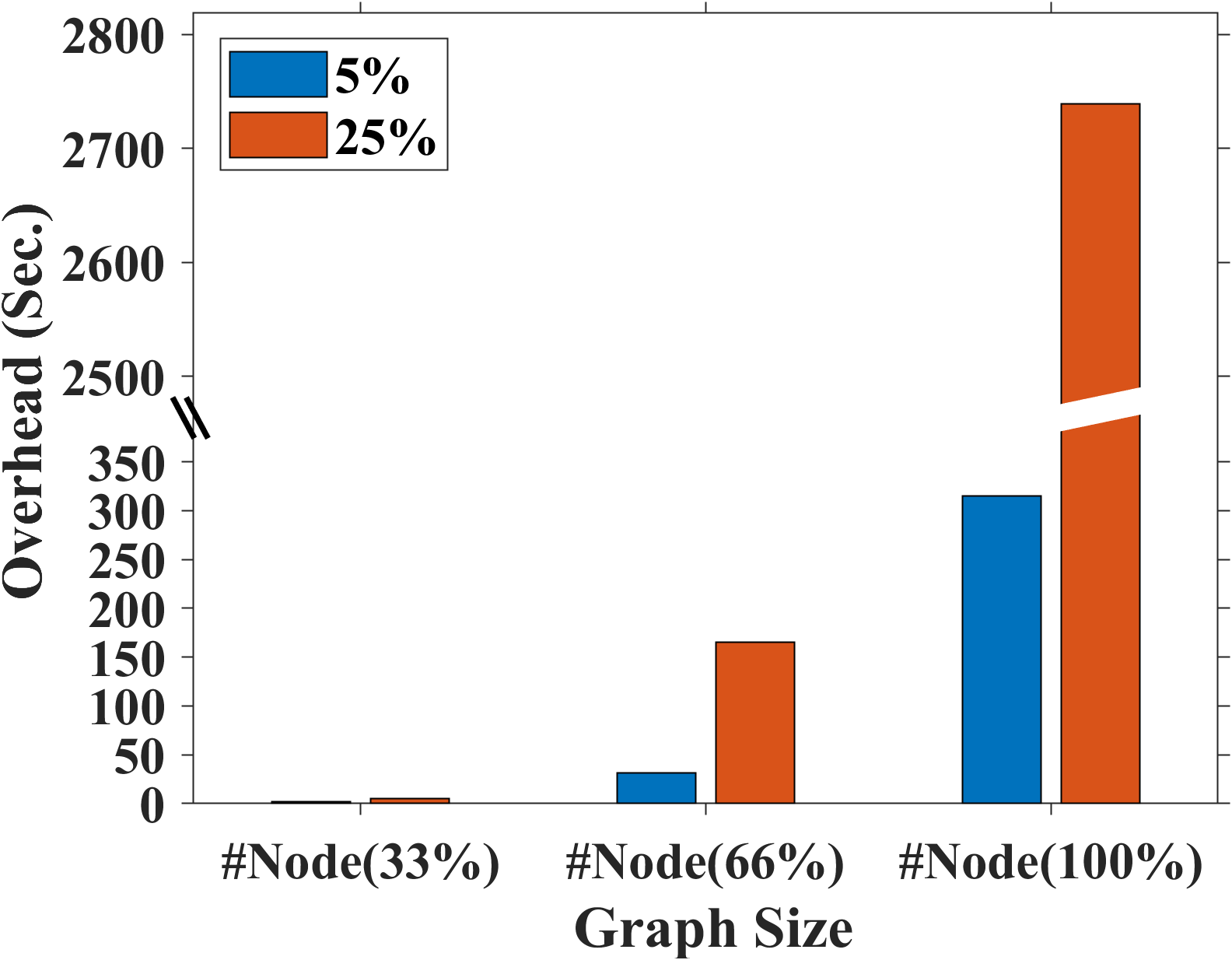}
\caption{The time overhead as graph (Citeseer) size enlarges}
\label{fig:TimeOverhead}
\end{figure}

\section{Related Work}
\begin{sloppypar}
Currently, graph data is facing serious adversarial attacks, primarily due to the complex dependencies (links) between graph regime. For example, from a global perspective, the graph structure can be manipulated to address the collective classification problem \cite{WangNeil19}. 
The attack Metattack employs a meta-learning strategy to attack nodes and edges, in addition to adjusting weights, it can rapidly raise the rank of the adjacency matrix and tend to establish connections between different nodes. Xu et al. \cite{pgd} utilize the popular attack PGD based on negative cross-entropy (CE-PGD) and the strong-attack CW (CW-PGD) to generate topological perturbations. Furthermore, from the perspective of network communities, DICE \cite{dice} attacks social networks by deliberately removing internal connections and establishing external connections, thereby reducing the effectiveness of node classification. More seriously, for training graph data, even if it is inaccessible, certain key attributes (such as node degree and subgraph structure) can still be explored through Property Inference Attacks (PIA) \cite{PIA} and Structural Member Inference Attacks (SMIA) \cite{SMIA}. He et al. \cite{HeMichael21} proposed a method for stealing graph links from training datasets under black-box access to GNN models. Overall, large graphs are regarded as complex systems with the entanglement of topology and features, and the numerous interactions (information transmission) between nodes make it extremely difficult to defend against persistent adversarial perturbations.
\par
Correspondingly, the adversarial defense mainly focuses on two aspects: graph itself and GNN model. The former enhances adversarial robustness by modifying the graph topology and node features. For example, from a preprocessing perspective, Wu et al. \cite{jaccard} removed dissimilar edges by comparing them with a preset threshold based on Jaccard similarity. GCN-SVD \cite{svd}, also as a preprocessing method, utilizes Singular Value Decomposition (SVD) to decompose the adjacency matrix of the perturbed graph, obtaining a low-rank approximation that represents a cleaner graph. Another research direction aims to resist adversarial perturbations by enhancing the robustness of GNN (neural-flow) models. For instance, the core idea of ProGNN \cite{JinMa20} is to combine the training process of GNNs with the structural properties of the graph, simultaneously learning a new structural graph and a robust GNN on the perturbed graph to ensure that the model can withstand attacks to some extent. At present, the neural ODEs have been successfully applied to GNNs, achieving this by modeling information exchange. References \cite{r1} and \cite{r2} treat the message passing process as a heat diffusion model, while references \cite{r3} and \cite{r4} view it as Beltrami diffusion. Additionally, for graph nodes, reference \cite{r5} employs an oscillator model to describe nodes and conducts the message passing process under the guidance of coupled ODEs.
\end{sloppypar}

\section{Conclusion}
We gain an in-depth understanding of attack patterns of GAAs and develop corresponding defense mechanism. GAAs tend to disrupt the connection and feature relationships among nodes, leading us to believe that a 1D simplified model is insufficient to adequately describe graph changes in the context. Specifically, for a given graph, we adopt weighting approach to construct a 2D model with two dynamic variables: one to describe topology adversarial resilience and the other to feature adversarial resilience. We demonstrate this 2D simplified model effectively captures the essential characteristics of disrupted graphs under three non-targeted and one targeted attacks, serving as a general paradigm for understanding attack patterns. On the other hand, inspired by the 2D model, we propose a graph purification strategy that effectively defends against various GAAs, showcasing its effectiveness across different datasets and models. In summary, our work sheds light on the potential adversarial defenses from both theoretical and practical angles as the future direction.
%%-------------------------------------------------------------------------------

%%-------------------------------------------------------------------------------
%
%The USENIX latex style is old and very tired, which is why
%there's no \textbackslash{}acks command for you to use when
%acknowledging. Sorry.
%
%\textbf{Do not include any acknowledgements in your submission which may deanonymize you (e.g., because of specific affiliations or grants you acknowledge)}

%-------------------------------------------------------------------------------
% optional clearing of the page
%\cleardoublepage
\appendix

\section*{Open Science}
%\textbf{Within up to one page, this appendix must list all artifacts necessary to evaluate the contribution of the paper and make clear how the review committees can access each artifact. This appendix must have exactly this title, otherwise you will risk desk rejection. }

In our work, all the experiments of ours and baselines are performed in an open-source benchmark for adversarial attack and defense, thus, the experimental results are reproducible. %We publicize our source code at Github to enhance sharing on our scientific findings once this work accepted, also guarantee its availability and functionality, as stated by the Artifact Evaluation committee. 

\textbf{Data Collection and Use.} The five datasets used in our experiments, i.e. Cora\_ML, Cora, Citeseer Amazon Photo, and PubMed are all open-sourced and commonly-used in academic domain. 
%Polblogs is a political-blogosphere relationship graph without the blog features involved: \url{https://pytorch-geometric.readthedocs.io/en/latest/generated/torch\_geometric.datasets.PolBlogs.html}  
%They all can be downloaded directly from Internet:\url{http://www.cs.umd.edu/~sen/lbc-proj/LBC.html}. 
The five datasets are only used for the experiment evaluation and performance analytics in this paper.

\textbf{Code and Datasets.}
%We open-source the code at Github: \url{https://github.com/Anony4OpenScience/TopFeaRe}.
We open-source code and datasets at Zenodo: \url{https://doi.org/10.5281/zenodo.17920431}.

\section*{Acknowledgments}
This work has been supported by Strategic Priority Research Program of CAS under Grant No. XDB0680301. 

% optional clearing of the page
%\cleardoublepage
\bibliographystyle{plain}
\bibliography{Mybib}

\appendix
\section {Boundary of Perturbation-Domain} \label{Prf_PerturbDomain}
\textbf{Proof.}
According to Theorem \ref{Theorem: AsymptoticStability}, there exists a diagonal matrix $M$ such that the perturbed output $\mathord{\buildrel{\lower3pt\hbox{$\scriptscriptstyle\rightharpoonup$}}\over y} \left( t \right)$ and its associated non-linear mapped input $ \phi \left( {\mathord{\buildrel{\lower3pt\hbox{$\scriptscriptstyle\rightharpoonup$}} \over y} \left( t \right)} \right) $ satisfy the condition $ \mathord{\buildrel{\lower3pt\hbox{$\scriptscriptstyle\rightharpoonup$}} 
\over \phi } \left( {\mathord{\buildrel{\lower3pt\hbox{$\scriptscriptstyle\rightharpoonup$}} 
\over y} \left( t \right)} \right)^T  \cdot \left[ {\mathord{\buildrel{\lower3pt\hbox{$\scriptscriptstyle\rightharpoonup$}} 
\over y} \left( t \right) - M\mathord{\buildrel{\lower3pt\hbox{$\scriptscriptstyle\rightharpoonup$}} 
\over \phi } \left( {\mathord{\buildrel{\lower3pt\hbox{$\scriptscriptstyle\rightharpoonup$}} 
\over y} \left( t \right)} \right)} \right] \geqslant 0 $. The diagonal elements of the diagonal $ M $ are influenced by $\phi \left( {\mathord{\buildrel{\lower3pt\hbox{$\scriptscriptstyle\rightharpoonup$}} \over y} \left( t \right)} \right) $. Using Eq. \eqref{eq:12}, $\frac{\left({{x}_{Gra}}\right)^{h}}{1+ \left({{x}_{Gra}}\right)^{h}}$ is deemed as the nonlinear perturbation input, and correspondingly $\phi \left( {\mathord{\buildrel{\lower3pt\hbox{$\scriptscriptstyle\rightharpoonup$}} \over y} \left( t \right)} \right) = \frac{\left({{x}_{Gra}}\right)^{h}}{1+ \left({{x}_{Gra}}\right)^{h}}$ and $ \mathord{\buildrel{\lower3pt\hbox{$\scriptscriptstyle\rightharpoonup$}}\over y} \left( t \right) = {x}_{Gra}$. Given that $ \phi \left( {\mathord{\buildrel{\lower3pt\hbox{$\scriptscriptstyle\rightharpoonup$}} \over y} \left( t \right)} \right)^T > 0 $ and $ {\mathord{\buildrel{\lower3pt\hbox{$\scriptscriptstyle\rightharpoonup$}} 
\over y} \left( t \right) - M\mathord{\buildrel{\lower3pt\hbox{$\scriptscriptstyle\rightharpoonup$}} 
\over \phi } \left( {\mathord{\buildrel{\lower3pt\hbox{$\scriptscriptstyle\rightharpoonup$}} 
\over y} \left( t \right)} \right)} \geqslant 0 $, the inequality $\frac{\phi \left( {\mathord{\buildrel{\lower3pt\hbox{$\scriptscriptstyle\rightharpoonup$}} \over y} \left( t \right)} \right)}{\mathord{\buildrel{\lower3pt\hbox{$\scriptscriptstyle\rightharpoonup$}} \over y} \left( t \right)} = \frac{\left({{x}_{Gra}}\right)^{h-1}}{1+ \left({{x}_{Gra}}\right)^{h}} \leqslant k$ is inferred. 

\section{Asymptotic Stability of Two-Dimensional Dynamic-Variation Mapping Equation} 
\label{Prf_CondnsFormlaAsymStab}
\textbf{Proof.} According to the definition in Eq. \eqref{eq:19}, it is clear that $-\mathcal{M}$ with a single element, its eigenvalue is simply the element itself. Therefore, the eigenvalue of $\mathbf{A}_r$ is $-\mathcal{M}$, which implies that $\mathbf{A}_r$ is a Hurwitz matrix. At this point, there exists ${\chi }_{r} \geq 0$, and when the condition $(1-\mathcal{M}{\chi }_{r}) \neq 0$ is satisfied, the one-dimensional dynamic equation can be computed as
\begin{equation}
	\begin{aligned}
		\label{eq:22}
		{ \widetilde{G}}_{r}(s) &= -\mathcal{M} + (I+s{\chi}_{r} ){(sI-\mathbf{A}_r) }^{-1}\mathbf{B}_{r} \\ 
		&= \frac{1}{{k}_{r}}+(1-{\chi }_{r}\mathcal{M}){(s+\mathcal{M})}^{-1}\left \langle {\gamma }_{r}\right \rangle+{\chi }_{r}\left \langle {\gamma }_{r}\right \rangle,
	\end{aligned}
\end{equation}
where $s=\sigma +\omega j$ and the complex unit $j = \sqrt{-1}$. First, it is clear that ${\widetilde{G}}_{r}(s)+{\widetilde{G}}_{r}(-s)$ is not identically zero. Moreover, since $\mathcal{M}$ is positive, it can be computed that
\begin{equation}
	\begin{aligned}
		\label{eq:23}
		{\widetilde{G}}_{r}(\omega j)+{{\widetilde{G}}_{r}}^{T}(-\omega j)=\frac{2}{{k}_{r}}+2\left \langle {\gamma }_{r}\right \rangle\frac{\mathcal{M}+{\chi }_{r}{\omega }^{2}}{{\mathcal{M}}^{2}+{\omega }^{2}}.
	\end{aligned}
\end{equation}
Similarly, it can be obtained that
\begin{equation}
	\begin{aligned}
		\label{eq:24}
		{\widetilde{G}}_{q}(\omega j)+{{\widetilde{G}}_{q}}^{T}(-\omega j)=\frac{2}{{k}_{q}}+2\left \langle {\gamma }_{q}\right \rangle\frac{\mathcal{C}+{\chi }_{q}{\omega }^{2}}{{\mathcal{C}}^{2}+{\omega }^{2}}. 
	\end{aligned}
\end{equation}
Therefore, for all real numbers $\omega$, the matrices ${ \widetilde{G}}_{r}(s)$ and ${ \widetilde{G}}_{q}(s)$ are positive definite and strictly positive real.

%\section{Surface of Asymptotic Stability} \label{Surface_ASEP}
%Fig. \ref{fig:q} exhibits the ASEP surface with dataset Cora\_ML.
%\begin{figure}[tb]
%	\centering
%	\hspace*{0pt} 
%	\subfigure[Metattack] 
%	{
%		\begin{minipage}[b]{.3\linewidth}
%			\centering
%			\includegraphics[width=1.4in, height=1.in]{Figure/reductionfigure/meta_cora_ml_f_.png}
%		\end{minipage}
%	}
%	\hfill 
%	\subfigure[CE-PGD] 
%	{
%		\begin{minipage}[b]{.3\linewidth}
%			\centering
%			\includegraphics[width=1.4in, height=1.in]{Figure/reductionfigure/pgd_cora_ml_f_.png}
%		\end{minipage}
%	}
%	\hfill 
%	\subfigure[DICE] 
%	{
%		\begin{minipage}[b]{.3\linewidth}
%			\centering
%			\includegraphics[width=1.4in, height=1.in]{Figure/reductionfigure/dice_cora_ml_f_.png}
%		\end{minipage}
%	}
%	\caption{ASEP surface with Cora\_ML}
%	\label{fig:q}
%\end{figure}

\section {Configuration} \label{EperimentConfig}

\textbf{Datasets.} The experiments are performed on five commonly-used scalable realistic graph datasets%\footnote{https://pytorch-geometric.readthedocs.io/en/latest/cheatsheet/data\_che\\atsheet.html}
: Cora, Cora\_ML, Citeseer, Amazon Photo, and PubMed. %The detailed statistics information is sketched in Table \ref{tab:1}.

\textbf{Baselines.} %Our goal is to enhance the adversarial resilience (robustness) of perturbed graphs by locating the critical state of adversarial resilience regarding the entanglement of graph topology and node feature. 
To evaluate our approach, we employ four types of graph adversarial attacks, including three non-targeted attacks Metattack \cite{r4}, CE-PGD \cite{pgd}, and DICE \cite{dice}, as well as the targeted attack Nettack\cite{nettack}. For the target models, we incorporate three neural network architectures-focused and two preprocessing-based models: %This comprehensive evaluation aims to demonstrate the effectiveness and scalability of our defense strategy.
%\begin{itemize}
%	\begin{sloppypar}
%		\item \textbf{GCN \cite{gcnKipf}}: As the representative of GNNs, it has been designed to extract features in various downstream tasks.
%		\item \textbf{GAT \cite{GATVelickovic}}: By utilizing multiple attention layers to dynamically adjust the weights of neighboring nodes, it can better learn the relationships and feature information.
%		\item \textbf{HANG \cite{hang}}: It is a recently-proposed graph neural flow method, and resorting to Hamiltonian neural flows to learn robust node embedding.
%		\item \textbf{GCN-SVD \cite{svd}}: It serves as a preprocessing method to defend against adversarial attacks by improving the robustness of GCN through low-rank approximation. % of the perturbed graph.
%		\item \textbf{GCN-Jaccard \cite{jaccard}}: By removing edges with low Jaccard similarity smaller than a predefined threshold, the graph topology is optimized to minimize the distinctions in node features caused by adversarial attacks.
%		\item \textbf{Mid-GCN \cite{HuangJin25}}: This approach is novel and recently proposed, it designs a mid-pass filtering GCN model to leverage mid-frequency signals (i.e. Laplacian eigenvalue around 1) to mitigate the adversarial attacks.
%	\end{sloppypar}
%\end{itemize}
i) \textbf{GCN \cite{gcnKipf}}: As the representative of GNNs, it has been designed to extract features in various downstream tasks; ii) \textbf{GAT \cite{GATVelickovic}}: By utilizing multiple attention layers to dynamically adjust the weights of neighboring nodes, it can better learn the relationships and feature information; iii) \textbf{HANG \cite{hang}}: It is a recently-proposed graph neural flow method, and resorting to Hamiltonian neural flows to learn robust node embedding; iv) \textbf{GCN-SVD \cite{svd}}: It serves as a preprocessing method to defend against adversarial attacks by improving the robustness of GCN through low-rank approximation of the perturbed graph; v) \textbf{GCN-Jaccard \cite{jaccard}}: By removing edges with low Jaccard similarity smaller than a predefined threshold, the graph topology is optimized to minimize the distinctions in node features caused by adversarial attacks; and vi) \textbf{Mid-GCN \cite{HuangJin25}}: It designs a mid-pass filtering GCN to leverage mid-frequency signals (i.e. Laplacian eigenvalue around 1) to mitigate the adversarial attacks.
\par
\begin{sloppypar}
\textbf{Execution Settings.} For non-targeted attacks, we set the rate of adversarial perturbation (RAP), i.e. the ratio of perturbed edges over all edges, to range from 0\% to 25\% in steps of 5\%. For targeted attacks, we select those nodes whose degree is greater than 10 in the test set for perturbation, with the number of perturbations ranging from 1 to 5 in steps of 1. For each graph, we randomly select 10\% of the nodes for model training, 10\% for validation, and 80\% for testing. The accuracy of node classification will be averaged over ten-time experiments. For baselines' settings, GCN \cite{gcnKipf}, GAT \cite{GATVelickovic}, and HANG \cite{hang} all use their default parameters. For GCN-SVD \cite{svd}, we choose the optimal rank reduction number from \{20, 40, 60, 80, 100\}. For GCN-Jaccard \cite{jaccard}, the threshold is selected from \{0.01, 0.02, 0.04, 0.06, 0.08, 0.1\} to find the optimal value. For our TopFeaRe, $\alpha$ is chosen from \{0.75, 0.8, 0.85, 0.9, 0.95, 0.99\} for the optimal value, with $\mathcal{W}_{J}$ and $\mathcal{W}_{C}$ setting to 0.3 and 0.7. The hardware running environment is: NVIDIA A800 80GB PCIe.
\end{sloppypar}

\section{Comparison with Baselines} \label{Exp:Preprocess&Neural}
Tables \ref{tab:Cora1}-\ref{tab:PubMed} provide the comparative experiments between ours and those baselines using preprocessing-based purification on datasets Cora, Citeseer, Amazon Photo, and PubMed. Tables \ref{tab:Combine-Cora_ML}-\ref{tab:Combine-PubMed} provide the comparative experiments between ours and those baselines using network-optimized defense on datasets Cora\_ML, Cora, Amazon Photo, and PubMed.
\begin{table*}[htbp]
\centering
\caption{Accuracy of node classification on Cora }
\label{tab:Cora1}
% \begin{tabular}{|c|c|c|c|c|c|c|c|} % 三列，均为居中对齐
\begin{tabular}{|>{\centering\arraybackslash}m{0.6cm}|>{\centering\arraybackslash}m{1.95cm}|>{\centering\arraybackslash}m{1.98cm}|>{\centering\arraybackslash}m{1.98cm}|>{\centering\arraybackslash}m{1.98cm}|>{\centering\arraybackslash}m{1.98cm}|>{\centering\arraybackslash}m{1.98cm}|>{\centering\arraybackslash}m{1.98cm}|}
\hline
\textbf{GAA} & \diagbox{{\textbf{Mod.}}}{{\textbf{RAP}}} & \textbf{0\%} & \textbf{5\%} & \textbf{10\%} & \textbf{15\%} & \textbf{20\%} & \textbf{25\%}\\ % 表头
\hline

\multirow{3}{*}{\makecell[c]{Meta\\ttack}} & GCN-SVD & 0.7801$\pm$0.0091 & 0.7622$\pm$0.0145 & 0.7478$\pm$0.0103 & 0.7090$\pm$0.0250 & 0.6849$\pm$0.0216 & 0.5800$\pm$0.1602\\ % 第一行
\cline{2-8} % 只画第二行到第八列的横线
& GCN-Jaccard &0.8221$\pm$0.0099&0.7819$\pm$0.0133&0.7561$\pm$0.0134&0.7066$\pm$0.0288&0.6878$\pm$0.0179&0.6491$\pm$0.0248\\ % 第二行
\cline{2-8} % 只画第二行到第八列的横线
& TopFeaRe &  {\textbf{0.8326$\pm$0.0083}} \(\textcolor{red}{(\uparrow \textbf{1.05\%})}\)& {\textbf{0.8032$\pm$0.0128}} \(\textcolor{red}{(\uparrow \textbf{2.13\%})}\)& {\textbf{0.7886$\pm$0.0131}} \(\textcolor{red}{(\uparrow \textbf{3.25\%})}\)& {\textbf{0.7807$\pm$0.0120}} \(\textcolor{red}{(\uparrow \textbf{7.17\%})}\)& {\textbf{0.7709$\pm$0.0126}} \(\textcolor{red}{(\uparrow \textbf{8.31\%})}\) & {\textbf{0.7521$\pm$0.0159}} \(\textcolor{red}{(\uparrow \textbf{10.30\%})}\)\\ % 第三行

\hline

\multirow{3}{*}{\makecell[c]{CE-\\PGD}} & GCN-SVD & 0.7801$\pm$0.0091 & 0.7682$\pm$0.0067 & 0.7624$\pm$0.0130 & 0.7555$\pm$0.0120 & 0.7488$\pm$0.0072 & 0.7440$\pm$0.0092\\ % 第一行
\cline{2-8} % 只画第二行到第八列的横线
& GCN-Jaccard & 0.8221$\pm$0.0099 & 0.8188$\pm$0.0102 & 0.8179$\pm$0.0092 & 0.8134$\pm$0.0090 & 0.8093$\pm$0.0067 & 0.8091$\pm$0.0096\\ % 第二行
\cline{2-8} % 只画第二行到第八列的横线
& TopFeaRe &  {\textbf{0.8326$\pm$0.0083}} \(\textcolor{red}{(\uparrow \textbf{1.05\%})}\)& {\textbf{0.8194$\pm$0.0105}} \(\textcolor{red}{(\uparrow \textbf{0.06\%})}\)& {\textbf{0.8186$\pm$0.0081}} \(\textcolor{red}{(\uparrow \textbf{0.07\%})}\)& {\textbf{0.8166$\pm$0.0061}} \(\textcolor{red}{(\uparrow \textbf{0.32\%})}\)& {\textbf{0.8158$\pm$0.0099}} \(\textcolor{red}{(\uparrow \textbf{0.65\%})}\) & {\textbf{0.8154$\pm$0.0061}} \(\textcolor{red}{(\uparrow \textbf{0.63\%})}\)\\ % 第三行

\hline

\multirow{3}{*}{DICE} & GCN-SVD & 0.7801$\pm$0.0091 & 0.7235$\pm$0.0080 & 0.7096$\pm$0.0121 & 0.6898$\pm$0.0165 & 0.6775$\pm$0.0063 & 0.6609$\pm$0.0094\\ % 第一行
\cline{2-8} % 只画第二行到第八列的横线
& GCN-Jaccard & 0.8221$\pm$0.0099 & 0.8102$\pm$0.0102 & 0.8024$\pm$0.0100 & 0.7894$\pm$0.0094 & 0.7771$\pm$0.0115 & 0.7678$\pm$0.0078\\ % 第二行
\cline{2-8} % 只画第二行到第八列的横线
& TopFeaRe &  {\textbf{0.8326$\pm$0.0083}} \(\textcolor{red}{(\uparrow \textbf{1.05\%})}\)& {\textbf{0.8152$\pm$0.0080}} \(\textcolor{red}{(\uparrow \textbf{0.50\%})}\)& {\textbf{0.8084$\pm$0.0080}} \(\textcolor{red}{(\uparrow \textbf{0.60\%})}\)& {\textbf{0.7970$\pm$0.0074}} \(\textcolor{red}{(\uparrow \textbf{0.76\%})}\)& {\textbf{0.7881$\pm$0.0104}} \(\textcolor{red}{(\uparrow \textbf{1.10\%})}\) & {\textbf{0.7807$\pm$0.0091}} \(\textcolor{red}{(\uparrow \textbf{1.29\%})}\)\\ % 第三行

\hline
\end{tabular}
\end{table*}

\begin{table*}[htbp]
\centering
\caption{Accuracy of node classification on Citeseer }
\label{tab:Citeseer1}
% \begin{tabular}{|c|c|c|c|c|c|c|c|} % 三列，均为居中对齐
\begin{tabular}{|>{\centering\arraybackslash}m{0.6cm}|>{\centering\arraybackslash}m{1.95cm}|>{\centering\arraybackslash}m{1.98cm}|>{\centering\arraybackslash}m{1.98cm}|>{\centering\arraybackslash}m{1.98cm}|>{\centering\arraybackslash}m{1.98cm}|>{\centering\arraybackslash}m{1.98cm}|>{\centering\arraybackslash}m{1.98cm}|}
\hline
\textbf{GAA} & \diagbox{{\textbf{Mod.}}}{{\textbf{RAP}}} & \textbf{0\%} & \textbf{5\%} & \textbf{10\%} & \textbf{15\%} & \textbf{20\%} & \textbf{25\%}\\ % 表头
\hline

\multirow{3}{*}{\makecell[c]{Meta\\ttack}} & GCN-SVD & 0.6831$\pm$0.0126 & 0.6686$\pm$0.0149 & 0.6456$\pm$0.0216 & 0.5634$\pm$0.1718 & 0.5832$\pm$0.0254 & 0.5708$\pm$0.0285\\ % 第一行
\cline{2-8} % 只画第二行到第八列的横线
& GCN-Jaccard & \textbf{0.7264$\pm$0.0126}&0.7098$\pm$0.0148&0.6802$\pm$0.0197&0.6469$\pm$0.0264&0.6227$\pm$0.0161&0.5956$\pm$0.0310\\ % 第二行
\cline{2-8} % 只画第二行到第八列的横线
& TopFeaRe &  {\textbf{0.7264$\pm$0.0135}} \(\textcolor{red}{(\uparrow \textbf{0\%})}\)& {\textbf{0.7159$\pm$0.0170}} \(\textcolor{red}{(\uparrow \textbf{0.61\%})}\)& {\textbf{0.7005$\pm$0.0162}} \(\textcolor{red}{(\uparrow \textbf{2.03\%})}\)& {\textbf{0.6885$\pm$0.0170}} \(\textcolor{red}{(\uparrow \textbf{4.16\%})}\)& {\textbf{0.6754$\pm$0.0147}} \(\textcolor{red}{(\uparrow \textbf{5.27\%})}\) & {\textbf{0.6647$\pm$0.0189}} \(\textcolor{red}{(\uparrow \textbf{6.91\%})}\)\\ % 第三行

\hline

\multirow{3}{*}{\makecell[c]{CE-\\PGD}} & GCN-SVD & 0.6831$\pm$0.0126 & 0.6773$\pm$0.0106 & 0.6763$\pm$0.0127 & 0.6960$\pm$0.0067 & 0.6875$\pm$0.0142 & 0.6842$\pm$0.0119\\ % 第一行
\cline{2-8} % 只画第二行到第八列的横线
& GCN-Jaccard & \textbf{0.7264$\pm$0.0126} & 0.7184$\pm$0.0140 & 0.7167$\pm$0.0126 & 0.7137$\pm$0.0128 & 0.7136$\pm$0.0134 & 0.7111$\pm$0.0144\\ % 第二行
\cline{2-8} % 只画第二行到第八列的横线
& TopFeaRe &  {\textbf{0.7264$\pm$0.0135}} \(\textcolor{red}{(\uparrow \textbf{0\%})}\)& {\textbf{0.7226$\pm$0.0137}} \(\textcolor{red}{(\uparrow \textbf{0.42\%})}\)& {\textbf{0.7226$\pm$0.0114}} \(\textcolor{red}{(\uparrow \textbf{0.59\%})}\)& {\textbf{0.7221$\pm$0.0113}} \(\textcolor{red}{(\uparrow \textbf{0.84\%})}\)& {\textbf{0.7185$\pm$0.0111}} \(\textcolor{red}{(\uparrow \textbf{0.49\%})}\) & {\textbf{0.7146$\pm$0.0111}} \(\textcolor{red}{(\uparrow \textbf{0.35\%})}\)\\ % 第三行

\hline

\multirow{3}{*}{DICE} & GCN-SVD & 0.6831$\pm$0.0126 & 0.6555$\pm$0.0161 & 0.6461$\pm$0.0230 & 0.6232$\pm$0.0155 & 0.6165$\pm$0.0178 & 0.5978$\pm$0.0213\\ % 第一行
\cline{2-8} % 只画第二行到第八列的横线
& GCN-Jaccard & \textbf{0.7264$\pm$0.0126} & \textbf{0.7185$\pm$0.0136} & 0.7081$\pm$0.0123 & 0.6964$\pm$0.0141 & 0.6834$\pm$0.0137 & 0.6777$\pm$0.0123\\ % 第二行
\cline{2-8} % 只画第二行到第八列的横线
& TopFeaRe &  {\textbf{0.7264$\pm$0.0135}} \(\textcolor{red}{(\uparrow \textbf{0\%})}\)& {0.7126$\pm$0.0156} \(\textcolor{red}{(\downarrow \textbf{0.59\%})}\)& {\textbf{0.7116$\pm$0.0141}} \(\textcolor{red}{(\uparrow \textbf{0.35\%})}\)& {\textbf{0.7008$\pm$0.0134}} \(\textcolor{red}{(\uparrow \textbf{0.44\%})}\)& {\textbf{0.6964$\pm$0.0179}} \(\textcolor{red}{(\uparrow \textbf{1.30\%})}\) & {\textbf{0.6847$\pm$0.0128}} \(\textcolor{red}{(\uparrow \textbf{0.70\%})}\)\\ % 第三行
\hline
\end{tabular}
\end{table*}

\begin{table*}[htbp]
\centering
\caption{Accuracy of node classification on Amazon Photo}
\label{tab:Amazon Photo}
% \begin{tabular}{|c|c|c|c|c|c|c|c|} % 三列，均为居中对齐
\begin{tabular}{|>{\centering\arraybackslash}m{0.6cm}|>{\centering\arraybackslash}m{1.95cm}|>{\centering\arraybackslash}m{1.98cm}|>{\centering\arraybackslash}m{1.98cm}|>{\centering\arraybackslash}m{1.98cm}|>{\centering\arraybackslash}m{1.98cm}|>{\centering\arraybackslash}m{1.98cm}|>{\centering\arraybackslash}m{1.98cm}|}
	\hline
	\textbf{GAA} & \diagbox{{\textbf{Mod.}}}{{\textbf{RAP}}} & \textbf{0\%} & \textbf{5\%} & \textbf{10\%} & \textbf{15\%} & \textbf{20\%} & \textbf{25\%}\\ % 表头
	\hline
	
	\multirow{3}{*}{\makecell[c]{Meta\\ttack}} & GCN-SVD & 0.8713$\pm$0.0047 & 0.8238$\pm$0.0235 & 0.8092$\pm$0.0249 & 0.7651$\pm$0.0827 & 0.6632$\pm$0.1165 & 0.6567$\pm$0.1482\\ % 第一行
	\cline{2-8} % 只画第二行到第八列的横线
	& GCN-Jaccard & 0.9232$\pm$0.0144&0.8327$\pm$0.0399&0.8217$\pm$0.0169&0.7332$\pm$0.1163&0.6344$\pm$0.1348&0.6496$\pm$0.1505\\ % 第二行
	\cline{2-8} % 只画第二行到第八列的横线
	& TopFeaRe &  {\textbf{0.9243$\pm$0.0118}} \(\textcolor{red}{(\uparrow \textbf{0.11\%})}\)& {\textbf{0.8486$\pm$0.0682}} \(\textcolor{red}{(\uparrow \textbf{1.59\%})}\)& {\textbf{0.8532$\pm$0.0338}} \(\textcolor{red}{(\uparrow \textbf{3.15\%})}\)& {\textbf{0.7754$\pm$0.1088}} \(\textcolor{red}{(\uparrow \textbf{1.03\%})}\)& {\textbf{0.7108$\pm$0.1246}} \(\textcolor{red}{(\uparrow \textbf{4.76\%})}\) & {\textbf{0.6847$\pm$0.1278}} \(\textcolor{red}{(\uparrow \textbf{2.80\%})}\)\\ % 第三行
	
	\hline
	
	\multirow{3}{*}{\makecell[c]{CE-\\PGD}} & GCN-SVD & 0.8713$\pm$0.0047 & 0.8756$\pm$0.0029 & 0.8743$\pm$0.0094 & 0.8771$\pm$0.0030 & 0.8672$\pm$0.0037 & 0.8680$\pm$0.0042\\ % 第一行
	\cline{2-8} % 只画第二行到第八列的横线
	& GCN-Jaccard & 0.9232$\pm$0.0144 & 0.9111$\pm$0.0039 & 0.8998$\pm$0.0231 & 0.8989$\pm$0.0043 & 0.8831$\pm$0.0057 & 0.8406$\pm$0.1294\\ % 第二行
	\cline{2-8} % 只画第二行到第八列的横线
	& TopFeaRe &  {\textbf{0.9243$\pm$0.0118}} \(\textcolor{red}{(\uparrow \textbf{0.11\%})}\)& {\textbf{0.9198$\pm$0.0143}} \(\textcolor{red}{(\uparrow \textbf{0.87\%})}\)& {\textbf{0.9091$\pm$0.0072}} \(\textcolor{red}{(\uparrow \textbf{0.93\%})}\)& {\textbf{0.9061$\pm$0.0051}} \(\textcolor{red}{(\uparrow \textbf{0.72\%})}\)& {\textbf{0.8976$\pm$0.0076}} \(\textcolor{red}{(\uparrow \textbf{1.45\%})}\) & {\textbf{0.8957$\pm$0.0075}} \(\textcolor{red}{(\uparrow \textbf{2.77\%})}\)\\ % 第三行
	
	\hline
	
	\multirow{3}{*}{DICE} & GCN-SVD & 0.8713$\pm$0.0047 & 0.8467$\pm$0.0064 & 0.8358$\pm$0.0048 & 0.8077$\pm$0.0163 & 0.7650$\pm$0.0273 & 0.7377$\pm$0.0671\\ % 第一行
	\cline{2-8} % 只画第二行到第八列的横线
	& GCN-Jaccard & 0.9232$\pm$0.0144 & 0.9099$\pm$0.0173 & 0.8995$\pm$0.0095 & 0.8917$\pm$0.0093 & 0.8763$\pm$0.0156 & 0.8481$\pm$0.0273\\ % 第二行
	\cline{2-8} % 只画第二行到第八列的横线
	& TopFeaRe &  {\textbf{0.9243$\pm$0.0118}} \(\textcolor{red}{(\uparrow \textbf{0.11\%})}\)& {\textbf{0.9189$\pm$0.0093}} \(\textcolor{red}{(\uparrow \textbf{0.90\%})}\)& {\textbf{0.9094$\pm$0.0138}} \(\textcolor{red}{(\uparrow \textbf{0.99\%})}\)& {\textbf{0.8975$\pm$0.0113}} \(\textcolor{red}{(\uparrow \textbf{0.58\%})}\)& {\textbf{0.8943$\pm$0.0111}} \(\textcolor{red}{(\uparrow \textbf{1.80\%})}\) & {\textbf{0.8803$\pm$0.0164}} \(\textcolor{red}{(\uparrow \textbf{3.22\%})}\)\\ % 第三行
	\hline
\end{tabular}
\end{table*}

\begin{table*}[htbp]
\centering
\caption{Accuracy of node classification on PubMed }
\label{tab:PubMed}
% \begin{tabular}{|c|c|c|c|c|c|c|c|} % 三列，均为居中对齐
	\begin{tabular}{|>{\centering\arraybackslash}m{0.6cm}|>{\centering\arraybackslash}m{1.95cm}|>{\centering\arraybackslash}m{1.98cm}|>{\centering\arraybackslash}m{1.98cm}|>{\centering\arraybackslash}m{1.98cm}|>{\centering\arraybackslash}m{1.98cm}|>{\centering\arraybackslash}m{1.98cm}|>{\centering\arraybackslash}m{1.98cm}|}
		\hline
		\textbf{GAA} & \diagbox{{\textbf{Mod.}}}{{\textbf{RAP}}} & \textbf{0\%} & \textbf{5\%} & \textbf{10\%} & \textbf{15\%} & \textbf{20\%} & \textbf{25\%}\\ % 表头
		\hline
		
		\multirow{3}{*}{\makecell[c]{Meta\\ttack}} & GCN-SVD & 0.8455$\pm$0.0007 & 0.8457$\pm$0.0014 & 0.8457$\pm$0.0014 & 0.8427$\pm$0.0009 & 0.8164$\pm$0.0024 & 0.7887$\pm$0.0025\\ % 第一行
		\cline{2-8} % 只画第二行到第八列的横线
		& GCN-Jaccard & 0.8685$\pm$0.0007&0.7981$\pm$0.0057&0.6381$\pm$0.0088&0.5588$\pm$0.0079&0.5320$\pm$0.0062&0.4450$\pm$0.0246\\ % 第二行
		\cline{2-8} % 只画第二行到第八列的横线
		& TopFeaRe &  {\textbf{0.8824$\pm$0.0054}} \(\textcolor{red}{(\uparrow \textbf{1.39\%})}\)& {\textbf{0.8672$\pm$0.0054}} \(\textcolor{red}{(\uparrow \textbf{2.15\%})}\)& {\textbf{0.8579$\pm$0.0095}} \(\textcolor{red}{(\uparrow \textbf{1.22\%})}\)& {\textbf{0.8514$\pm$0.0110}} \(\textcolor{red}{(\uparrow \textbf{0.87\%})}\)& {\textbf{0.8365$\pm$0.0083}} \(\textcolor{red}{(\uparrow \textbf{2.01\%})}\) & {\textbf{0.8177$\pm$0.0038}} \(\textcolor{red}{(\uparrow \textbf{2.90\%})}\)\\ % 第三行
		
		\hline
		
		\multirow{3}{*}{\makecell[c]{CE-\\PGD}} & GCN-SVD & 0.8455$\pm$0.0007 & 0.8476$\pm$0.0016 & 0.8469$\pm$0.0009 & 0.8448$\pm$0.0005 & 0.8387$\pm$0.0005 & 0.8374$\pm$0.0007\\ % 第一行
		\cline{2-8} % 只画第二行到第八列的横线
		& GCN-Jaccard & 0.8685$\pm$0.0007 & 0.8539$\pm$0.0004 & 0.8416$\pm$0.0004 & 0.8305$\pm$0.0011 & 0.8261$\pm$0.0004 & 0.8154$\pm$0.0016\\ % 第二行
		\cline{2-8} % 只画第二行到第八列的横线
		& TopFeaRe &  {\textbf{0.8824$\pm$0.0054}} \(\textcolor{red}{(\uparrow \textbf{1.39\%})}\)& {\textbf{0.8644$\pm$0.0095}} \(\textcolor{red}{(\uparrow \textbf{1.05\%})}\)& {\textbf{0.8586$\pm$0.0038}} \(\textcolor{red}{(\uparrow \textbf{1.17\%})}\)& {\textbf{0.8549$\pm$0.0025}} \(\textcolor{red}{(\uparrow \textbf{1.01\%})}\)& {\textbf{0.8476$\pm$0.0068}} \(\textcolor{red}{(\uparrow \textbf{0.89\%})}\) & {\textbf{0.8457$\pm$0.0050}} \(\textcolor{red}{(\uparrow \textbf{0.83\%})}\)\\ % 第三行
		
		\hline
		
		\multirow{3}{*}{DICE} & GCN-SVD & 0.8455$\pm$0.0007 & 0.8408$\pm$0.0010 & 0.8332$\pm$0.0010 & 0.8332$\pm$0.0011 & 0.8258$\pm$0.0012 & 0.8206$\pm$0.0016\\ % 第一行
		\cline{2-8} % 只画第二行到第八列的横线
		& GCN-Jaccard & 0.8685$\pm$0.0007 & 0.8530$\pm$0.0008 & 0.8351$\pm$0.0012 & 0.8236$\pm$0.0012 & 0.8103$\pm$0.0010 & 0.7926$\pm$0.0012\\ % 第二行
		\cline{2-8} % 只画第二行到第八列的横线
		& TopFeaRe &  {\textbf{0.8824$\pm$0.0054}} \(\textcolor{red}{(\uparrow \textbf{1.39\%})}\)& {\textbf{0.8623$\pm$0.0060}} \(\textcolor{red}{(\uparrow \textbf{0.93\%})}\)& {\textbf{0.8579$\pm$0.035}} \(\textcolor{red}{(\uparrow \textbf{2.28\%})}\)& {\textbf{0.8517$\pm$0.045}} \(\textcolor{red}{(\uparrow \textbf{1.85\%})}\)& {\textbf{0.8459$\pm$0.0019}} \(\textcolor{red}{(\uparrow \textbf{2.01\%})}\) & {\textbf{0.8386$\pm$0.015}} \(\textcolor{red}{(\uparrow \textbf{1.80\%})}\)\\ % 第三行
		\hline
	\end{tabular}
\end{table*}

%\section{Comparison with Neural Network-Optimized Baselines} \label{Exp:Combine}
%Tables \ref{tab:Combine-Cora_ML}-\ref{tab:Combine-PubMed} provide the comparative experiments between ours and those baselines using network-optimized defense on datasets Cora\_ML, Cora, Amazon Photo, and PubMed.

\begin{table*}[htbp]
\centering
%\caption{ACCURACY OF NODE CLASSIFICATION UNDER COMBINATION WITH Ours USING METATTACK (\textbf{BOLD}-THE BEST)}
\caption{Accuracy of node classification on Cora\_ML under combination with Ours using Metattack}
\label{tab:Combine-Cora_ML}
% \begin{tabular}{|c|c|c|c|c|c|c|c|} % 三列，均为居中对齐
\begin{tabular}{|>{\centering\arraybackslash}m{0.6cm}|>{\centering\arraybackslash}m{1.95cm}|>{\centering\arraybackslash}m{1.98cm}|>{\centering\arraybackslash}m{1.98cm}|>{\centering\arraybackslash}m{1.98cm}|>{\centering\arraybackslash}m{1.98cm}|>{\centering\arraybackslash}m{1.98cm}|>{\centering\arraybackslash}m{1.98cm}|}
\hline
\textbf{\makecell[c]{Data\\set}} & \diagbox{{\textbf{Mod.}}}{{\textbf{RAP}}} & \textbf{0\%} & \textbf{5\%} & \textbf{10\%} & \textbf{15\%} & \textbf{20\%} & \textbf{25\%}\\ % 表头
\hline

\multirow{6}{*}{\makecell[c]{Cora\\\_ML}} & GCN & 0.8505$\pm$0.0111 & 0.7915$\pm$0.0096 & 0.7353$\pm$0.0152 & 0.6856$\pm$0.0193 & 0.6333$\pm$0.0339 & 0.5665$\pm$0.0471\\ 
\cline{2-8} 
& TopFeaRe &  {\textbf{0.8516$\pm$0.0104}} \(\textcolor{red}{(\uparrow \textbf{0.11\%})}\)& {\textbf{0.8293$\pm$0.0088}} \(\textcolor{red}{(\uparrow \textbf{3.78\%})}\)& {\textbf{0.8219$\pm$0.0114}} \(\textcolor{red}{(\uparrow \textbf{8.66\%})}\)& {\textbf{0.8150$\pm$0.0107}} \(\textcolor{red}{(\uparrow \textbf{12.94\%})}\)& {\textbf{0.8086$\pm$0.0104}} \(\textcolor{red}{(\uparrow \textbf{17.53\%})}\) & {\textbf{0.7920$\pm$0.0128}} \(\textcolor{red}{(\uparrow \textbf{22.55\%})}\)\\ 
\cline{2-8}
& GAT & 0.8466$\pm$0.0108&0.8005$\pm$0.0165&0.7590$\pm$0.0150&0.7238$\pm$0.0172&0.6778$\pm$0.0351&0.6174$\pm$0.0500\\ 
\cline{2-8} 
& GAT-Our &  {\textbf{0.8476$\pm$0.0097}} \(\textcolor{red}{(\uparrow \textbf{0.10\%})}\)& {\textbf{0.8213$\pm$0.0142}} \(\textcolor{red}{(\uparrow \textbf{2.08\%})}\)& {\textbf{0.8146$\pm$0.0124}} \(\textcolor{red}{(\uparrow \textbf{5.56\%})}\)& {\textbf{0.8069$\pm$0.0085}} \(\textcolor{red}{(\uparrow \textbf{8.31\%})}\)& {\textbf{0.7960$\pm$0.0094}} \(\textcolor{red}{(\uparrow \textbf{11.82\%})}\) & {\textbf{0.7730$\pm$0.0140}} \(\textcolor{red}{(\uparrow \textbf{15.56\%})}\)\\
\cline{2-8}
& HANG & 0.8419$\pm$0.0096&0.8287$\pm$0.0100&0.8175$\pm$0.0105&0.8029$\pm$0.0062&0.7927$\pm$0.0170&0.7737$\pm$0.0218\\ 
\cline{2-8} 
& HANG-Our &  {\textbf{0.8466$\pm$0.0099}} \(\textcolor{red}{(\uparrow \textbf{0.47\%})}\)& {\textbf{0.8315$\pm$0.0099}} \(\textcolor{red}{(\uparrow \textbf{0.28\%})}\)& {\textbf{0.8204$\pm$0.0100}} \(\textcolor{red}{(\uparrow \textbf{0.29\%})}\)& {\textbf{0.8131$\pm$0.0063}} \(\textcolor{red}{(\uparrow \textbf{1.02\%})}\)& {\textbf{0.8060$\pm$0.0089}} \(\textcolor{red}{(\uparrow \textbf{1.33\%})}\) & {\textbf{0.7973$\pm$0.0137}} \(\textcolor{red}{(\uparrow \textbf{2.36\%})}\)\\
\cline{2-8}
& MidGCN & 0.8340$\pm$0.0021&0.8001$\pm$0.0037&0.7895$\pm$0.0037&0.7780$\pm$0.0055&0.7566$\pm$0.0057&0.7417$\pm$0.0072\\ 
\cline{2-8} 
& MidGCN-Our &  {\textbf{0.8349$\pm$0.0016}} \(\textcolor{red}{(\uparrow \textbf{0.09\%})}\)& {\textbf{0.8035$\pm$0.0044}} \(\textcolor{red}{(\uparrow \textbf{0.34\%})}\)& {\textbf{0.7943$\pm$0.0047}} \(\textcolor{red}{(\uparrow \textbf{0.48\%})}\)& {\textbf{0.7829$\pm$0.0074}} \(\textcolor{red}{(\uparrow \textbf{0.49\%})}\)& {\textbf{0.7645$\pm$0.0057}} \(\textcolor{red}{(\uparrow \textbf{0.79\%})}\) & {\textbf{0.7422$\pm$0.0085}} \(\textcolor{red}{(\uparrow \textbf{0.05\%})}\)\\

\hline

\end{tabular}
\end{table*}

\begin{table*}[htbp]
\centering
%\caption{ACCURACY OF NODE CLASSIFICATION UNDER COMBINATION WITH Ours USING METATTACK (\textbf{BOLD}-THE BEST)}
\caption{Accuracy of node classification on Cora under combination with Ours using Metattack}
\label{tab:Combine-Cora}
% \begin{tabular}{|c|c|c|c|c|c|c|c|} % 三列，均为居中对齐
\begin{tabular}{|>{\centering\arraybackslash}m{0.6cm}|>{\centering\arraybackslash}m{1.95cm}|>{\centering\arraybackslash}m{1.98cm}|>{\centering\arraybackslash}m{1.98cm}|>{\centering\arraybackslash}m{1.98cm}|>{\centering\arraybackslash}m{1.98cm}|>{\centering\arraybackslash}m{1.98cm}|>{\centering\arraybackslash}m{1.98cm}|}
\hline
\textbf{\makecell[c]{Data\\set}} & \diagbox{{\textbf{Mod.}}}{{\textbf{RAP}}} & \textbf{0\%} & \textbf{5\%} & \textbf{10\%} & \textbf{15\%} & \textbf{20\%} & \textbf{25\%}\\ % 表头

\hline
\multirow{6}{*}{Cora} & GCN & 0.8300$\pm$0.0100 & 0.7535$\pm$0.0184 & 0.6683$\pm$0.0340 & 0.5908$\pm$0.0408 & 0.5055$\pm$0.0330 & 0.4397$\pm$0.0358\\ 
\cline{2-8} 
& TopFeaRe &  {\textbf{0.8324$\pm$0.0101}} \(\textcolor{red}{(\uparrow \textbf{0.24\%})}\)& {\textbf{0.8032$\pm$0.0128}} \(\textcolor{red}{(\uparrow \textbf{4.97\%})}\)& {\textbf{0.7886$\pm$0.0131}} \(\textcolor{red}{(\uparrow \textbf{12.03\%})}\)& {\textbf{0.7807$\pm$0.0120}} \(\textcolor{red}{(\uparrow \textbf{18.99\%})}\)& {\textbf{0.7709$\pm$0.0126}} \(\textcolor{red}{(\uparrow \textbf{26.54\%})}\) & {\textbf{0.7521$\pm$0.0159}} \(\textcolor{red}{(\uparrow \textbf{31.24\%})}\)\\ 
\cline{2-8}
& GAT & 0.8328$\pm$0.0101&0.7965$\pm$0.0175&0.7784$\pm$0.0190&0.7136$\pm$0.0463&0.6655$\pm$0.0390&0.6178$\pm$0.0441\\ 
\cline{2-8} 
& GAT-Our &  {\textbf{0.8388$\pm$0.0110}} \(\textcolor{red}{(\uparrow \textbf{0.60\%})}\)& {\textbf{0.8122$\pm$0.0122}} \(\textcolor{red}{(\uparrow \textbf{1.57\%})}\)& {\textbf{0.8019$\pm$0.0121}} \(\textcolor{red}{(\uparrow \textbf{2.35\%})}\)& {\textbf{0.7974$\pm$0.0109}} \(\textcolor{red}{(\uparrow \textbf{8.38\%})}\)& {\textbf{0.7790$\pm$0.0126}} \(\textcolor{red}{(\uparrow \textbf{11.35\%})}\) & {\textbf{0.7704$\pm$0.0155}} \(\textcolor{red}{(\uparrow \textbf{15.26\%})}\)\\
\cline{2-8}
& HANG & 0.8192$\pm$0.0092&0.7894$\pm$0.0137&0.7724$\pm$0.0185&0.7419$\pm$0.0256&0.7154$\pm$0.0245&0.6946$\pm$0.0438\\ 
\cline{2-8} 
& HANG-Our &  {\textbf{0.8346$\pm$0.0090}} \(\textcolor{red}{(\uparrow \textbf{1.54\%})}\)& {\textbf{0.8036$\pm$0.0117}} \(\textcolor{red}{(\uparrow \textbf{1.42\%})}\)& {\textbf{0.7867$\pm$0.0145}} \(\textcolor{red}{(\uparrow \textbf{1.43\%})}\)& {\textbf{0.7852$\pm$0.0110}} \(\textcolor{red}{(\uparrow \textbf{4.33\%})}\)& {\textbf{0.7731$\pm$0.0132}} \(\textcolor{red}{(\uparrow \textbf{5.77\%})}\) & {\textbf{0.7671$\pm$0.0190}} \(\textcolor{red}{(\uparrow \textbf{7.25\%})}\)\\
\cline{2-8}
& MidGCN & 0.8310$\pm$0.0041&0.7581$\pm$0.1582&0.7888$\pm$0.0093&0.7543$\pm$0.0111&0.7550$\pm$0.0135&0.7443$\pm$0.0083\\ 
\cline{2-8} 
& MidGCN-Our &  {\textbf{0.8323$\pm$0.0036}} \(\textcolor{red}{(\uparrow \textbf{0.13\%})}\)& {\textbf{0.8148$\pm$0.0037}} \(\textcolor{red}{(\uparrow \textbf{5.67\%})}\)& {\textbf{0.7944$\pm$0.0078}} \(\textcolor{red}{(\uparrow \textbf{0.56\%})}\)& {\textbf{0.7641$\pm$0.0118}} \(\textcolor{red}{(\uparrow \textbf{0.98\%})}\)& {\textbf{0.7645$\pm$0.0049}} \(\textcolor{red}{(\uparrow \textbf{0.95\%})}\) & {\textbf{0.7469$\pm$0.0121}} \(\textcolor{red}{(\uparrow \textbf{0.26\%})}\)\\

\hline

\end{tabular}
\end{table*}

\begin{table*}[htbp]
\centering
%\caption{ACCURACY OF NODE CLASSIFICATION UNDER COMBINATION WITH Ours USING METATTACK (\textbf{BOLD}-THE BEST)}
\caption{Accuracy of node classification on Amazon Photo under combination with Ours using Metattack}
\label{tab:Combine-AmazonPhoto}
% \begin{tabular}{|c|c|c|c|c|c|c|c|} % 三列，均为居中对齐
\begin{tabular}{|>{\centering\arraybackslash}m{0.6cm}|>{\centering\arraybackslash}m{1.95cm}|>{\centering\arraybackslash}m{1.98cm}|>{\centering\arraybackslash}m{1.98cm}|>{\centering\arraybackslash}m{1.98cm}|>{\centering\arraybackslash}m{1.98cm}|>{\centering\arraybackslash}m{1.98cm}|>{\centering\arraybackslash}m{1.98cm}|}
\hline
\textbf{\makecell[c]{Data\\set}} & \diagbox{{\textbf{Mod.}}}{{\textbf{RAP}}} & \textbf{0\%} & \textbf{5\%} & \textbf{10\%} & \textbf{15\%} & \textbf{20\%} & \textbf{25\%}\\ % 表头

\hline
\multirow{6}{*}{\makecell[c]{Amaz\\on P\\hoto}} & GCN & 0.9215$\pm$0.0034 & 0.8242$\pm$0.0774 & 0.8281$\pm$0.0581 & 0.7529$\pm$0.1425 & 0.6295$\pm$0.1496 & 0.6108$\pm$0.1526\\ 
\cline{2-8} 
& TopFeaRe &  {\textbf{0.9243$\pm$0.0118}} \(\textcolor{red}{(\uparrow \textbf{0.28\%})}\)& {\textbf{0.8486$\pm$0.0682}} \(\textcolor{red}{(\uparrow \textbf{2.44\%})}\)& {\textbf{0.8532$\pm$0.0338}} \(\textcolor{red}{(\uparrow \textbf{2.51\%})}\)& {\textbf{0.7754$\pm$0.1088}} \(\textcolor{red}{(\uparrow \textbf{2.25\%})}\)& {\textbf{0.7108$\pm$0.1246}} \(\textcolor{red}{(\uparrow \textbf{8.13\%})}\) & {\textbf{0.6847$\pm$0.1278}} \(\textcolor{red}{(\uparrow \textbf{7.39\%})}\)\\ 
\cline{2-8}
& GAT & 0.9354$\pm$0.0018&0.9173$\pm$0.0039&0.7824$\pm$0.1926&0.7239$\pm$0.1575&0.6528$\pm$0.1737&0.5811$\pm$0.1824\\
\cline{2-8} 
& GAT-Our &  {\textbf{0.9448$\pm$0.0018}} \(\textcolor{red}{(\uparrow \textbf{0.94\%})}\)& {\textbf{0.9257$\pm$0.0030}} \(\textcolor{red}{(\uparrow \textbf{0.84\%})}\)& {\textbf{0.8598$\pm$0.1236}} \(\textcolor{red}{(\uparrow \textbf{7.74\%})}\)& {\textbf{0.8079$\pm$0.1337}} \(\textcolor{red}{(\uparrow \textbf{8.40\%})}\)& {\textbf{0.7731$\pm$0.1729}} \(\textcolor{red}{(\uparrow \textbf{12.03\%})}\) & {\textbf{0.7614$\pm$0.1384}} \(\textcolor{red}{(\uparrow \textbf{18.03\%})}\)\\
\cline{2-8}
& HANG & 0.9347$\pm$0.0032&0.9162$\pm$0.0054&0.9151$\pm$0.0039&0.9093$\pm$0.0068&0.9016$\pm$0.0056&0.8992$\pm$0.0061\\ 
\cline{2-8} 
& HANG-Our &  {\textbf{0.9441$\pm$0.0025}} \(\textcolor{red}{(\uparrow \textbf{0.94\%})}\)& {\textbf{0.9245$\pm$0.0045}} \(\textcolor{red}{(\uparrow \textbf{0.83\%})}\)& {\textbf{0.9236$\pm$0.0033}} \(\textcolor{red}{(\uparrow \textbf{0.85\%})}\)& {\textbf{0.9203$\pm$0.0058}} \(\textcolor{red}{(\uparrow \textbf{1.10\%})}\)& {\textbf{0.9184$\pm$0.0048}} \(\textcolor{red}{(\uparrow \textbf{1.68\%})}\) & {\textbf{0.9052$\pm$0.0052}} \(\textcolor{red}{(\uparrow \textbf{0.60\%})}\)\\
\cline{2-8}
& MidGCN &      0.7725$\pm$0.0025&0.7439$\pm$0.0049&0.6595$\pm$0.0082&0.5708$\pm$0.0170&0.4875$\pm$0.0088&0.4675$\pm$0.0113\\ 
\cline{2-8} 
& MidGCN-Our &  {\textbf{0.7903$\pm$0.0053}} \(\textcolor{red}{(\uparrow \textbf{1.78\%})}\)& {\textbf{0.7473$\pm$0.0064}} \(\textcolor{red}{(\uparrow \textbf{0.34\%})}\)& {\textbf{0.6651$\pm$0.0067}} \(\textcolor{red}{(\uparrow \textbf{0.56\%})}\)& {\textbf{0.5969$\pm$0.0102}} \(\textcolor{red}{(\uparrow \textbf{2.61\%})}\)& {\textbf{0.5020$\pm$0.0088}} \(\textcolor{red}{(\uparrow \textbf{1.45\%})}\) & {\textbf{0.4759$\pm$0.0101}} \(\textcolor{red}{(\uparrow \textbf{0.84\%})}\)\\

\hline

\end{tabular}
\end{table*}

\begin{table*}[htbp]
\centering
%\caption{ACCURACY OF NODE CLASSIFICATION UNDER COMBINATION WITH Ours USING METATTACK (\textbf{BOLD}-THE BEST)}
\caption{Accuracy of node classification on PubMed under combination with Ours using Metattack}
\label{tab:Combine-PubMed}
% \begin{tabular}{|c|c|c|c|c|c|c|c|} % 三列，均为居中对齐
\begin{tabular}{|>{\centering\arraybackslash}m{0.6cm}|>{\centering\arraybackslash}m{1.95cm}|>{\centering\arraybackslash}m{1.98cm}|>{\centering\arraybackslash}m{1.98cm}|>{\centering\arraybackslash}m{1.98cm}|>{\centering\arraybackslash}m{1.98cm}|>{\centering\arraybackslash}m{1.98cm}|>{\centering\arraybackslash}m{1.98cm}|}
\hline
\textbf{\makecell[c]{Data\\set}} & \diagbox{{\textbf{Mod.}}}{{\textbf{RAP}}} & \textbf{0\%} & \textbf{5\%} & \textbf{10\%} & \textbf{15\%} & \textbf{20\%} & \textbf{25\%}\\ % 表头

\hline
\multirow{6}{*}{\makecell[c]{PubM\\ed}} & GCN & 0.8340$\pm$0.0021 & 0.8001$\pm$0.0037 & 0.7895$\pm$0.0037 & 0.7780$\pm$0.0055 & 0.7566$\pm$0.0057 & 0.7417$\pm$0.0072\\ 
\cline{2-8} 
& TopFeaRe &  {\textbf{0.8824$\pm$0.0054}} \(\textcolor{red}{(\uparrow \textbf{4.84\%})}\)& {\textbf{0.8672$\pm$0.0054}} \(\textcolor{red}{(\uparrow \textbf{6.71\%})}\)& {\textbf{0.8579$\pm$0.0095}} \(\textcolor{red}{(\uparrow \textbf{6.84\%})}\)& {\textbf{0.8514$\pm$0.0110}} \(\textcolor{red}{(\uparrow \textbf{7.34\%})}\)& {\textbf{0.8365$\pm$0.0083}} \(\textcolor{red}{(\uparrow \textbf{7.99\%})}\) & {\textbf{0.8177$\pm$0.0038}} \(\textcolor{red}{(\uparrow \textbf{7.60\%})}\)\\ 
\cline{2-8}

& GAT & 0.8583$\pm$0.0012&0.8271$\pm$0.0032&0.6808$\pm$0.0236&0.5419$\pm$0.0128&0.3640$\pm$0.0416&0.3165$\pm$0.0394\\
\cline{2-8} 
& GAT-Our &  {\textbf{0.8640$\pm$0.0004}} \(\textcolor{red}{(\uparrow \textbf{0.57\%})}\)& {\textbf{0.8348$\pm$0.0061}} \(\textcolor{red}{(\uparrow \textbf{0.77\%})}\)& {\textbf{0.7923$\pm$0.0140}} \(\textcolor{red}{(\uparrow \textbf{11.15\%})}\)& {\textbf{0.6546$\pm$0.0271}} \(\textcolor{red}{(\uparrow \textbf{11.27\%})}\)& {\textbf{0.6070$\pm$0.0091}} \(\textcolor{red}{(\uparrow \textbf{24.30\%})}\) & {\textbf{0.5848$\pm$0.0220}} \(\textcolor{red}{(\uparrow \textbf{26.83\%})}\)\\
\cline{2-8}
& HANG & 0.8712$\pm$0.0006&0.8001$\pm$0.0037&0.7895$\pm$0.0037&0.7780$\pm$0.0055&0.7566$\pm$0.0057&0.7417$\pm$0.0072\\ 
\cline{2-8} 
& HANG-Our &  {\textbf{0.8711$\pm$0.0008}} \(\textcolor{red}{(\downarrow \textbf{0.01\%})}\)& {\textbf{0.8665$\pm$0.0027}} \(\textcolor{red}{(\uparrow \textbf{6.64\%})}\)& {\textbf{0.8637$\pm$0.0014}} \(\textcolor{red}{(\uparrow \textbf{7.42\%})}\)& {\textbf{0.8571$\pm$0.0042}} \(\textcolor{red}{(\uparrow \textbf{7.91\%})}\)& {\textbf{0.8492$\pm$0.0023}} \(\textcolor{red}{(\uparrow \textbf{9.26\%})}\) & {\textbf{0.7940$\pm$0.0103}} \(\textcolor{red}{(\uparrow \textbf{5.23\%})}\)\\
\cline{2-8}
& MidGCN &      0.8340$\pm$0.0021&0.8001$\pm$0.0037&0.7895$\pm$0.0037&0.7780$\pm$0.0055&0.7566$\pm$0.0057&0.7417$\pm$0.0072\\ 
\cline{2-8} 
& MidGCN-Our &  {\textbf{0.8581$\pm$0.0053}} \(\textcolor{red}{(\uparrow \textbf{2.41\%})}\)& {\textbf{0.8403$\pm$0.0062}} \(\textcolor{red}{(\uparrow \textbf{4.02\%})}\)& {\textbf{0.8111$\pm$0.0127}} \(\textcolor{red}{(\uparrow \textbf{2.16\%})}\)& {\textbf{0.7959$\pm$0.0219}} \(\textcolor{red}{(\uparrow \textbf{1.79\%})}\)& {\textbf{0.7784$\pm$0.0123}} \(\textcolor{red}{(\uparrow \textbf{2.18\%})}\) & {\textbf{0.7532$\pm$0.0308}} \(\textcolor{red}{(\uparrow \textbf{1.15\%})}\)\\

\hline

\end{tabular}
\end{table*}

\end{document}